\documentclass[11pt]{article}
\pdfoutput=1
\usepackage[utf8]{inputenc}
\usepackage[T1]{fontenc}
\usepackage[margin=1in]{geometry}
\usepackage{amsmath,amsthm,amssymb,thmtools,thm-restate,booktabs,color,doi,graphicx,latexsym,url,xcolor,xspace}
\usepackage[numbers,sort&compress]{natbib}
\usepackage{microtype,hyperref}
\usepackage[inline]{enumitem}
\usepackage{subcaption}
\setlist[itemize]{label=--}
\setlist[enumerate]{label=(\arabic*),labelindent=\parindent,leftmargin=*}

\usepackage{sectsty}
\allsectionsfont{\boldmath}

\definecolor{citecolor}{HTML}{0000C0}
\definecolor{urlcolor}{HTML}{000080}

\hypersetup{
    colorlinks=true,
    linkcolor=black,
    citecolor=citecolor,
    filecolor=black,
    urlcolor=urlcolor,
    pdftitle={Near-Optimal Leader Election in Population Protocols on Graphs}, %
    pdfauthor={Dan Alistarh, Joel Rybicki, Sasha Voitovych} %
}

\newtheorem{theorem}{Theorem}
\newtheorem{lemma}[theorem]{Lemma}
\newtheorem{corollary}[theorem]{Corollary}
\newtheorem{proposition}[theorem]{Proposition}

\newenvironment{myabstract}
{\list{}{\listparindent 1.5em%
		\itemindent    \listparindent
		\leftmargin    1cm
		\rightmargin   1cm
		\parsep        0pt}%
	\item\relax}
{\endlist}

\newenvironment{mycover}
{\list{}{\listparindent 0pt
		\itemindent    \listparindent
		\leftmargin    1cm
		\rightmargin   1cm
		\parsep        0pt}%
	\raggedright
	\item\relax}
{\endlist}

\newcommand{\myemail}[1]{\,$\cdot$\, {\small #1}}
\newcommand{\myaff}[1]{\,$\cdot$\, {\small #1}\par\medskip}

\hypersetup{
 linkcolor=black
}

\newcommand{\namedref}[2]{\hyperref[#2]{#1~\ref*{#2}}}
\newcommand{\sectionref}[1]{\namedref{Section}{#1}}
\newcommand{\figureref}[1]{\namedref{Figure}{#1}}
\newcommand{\tableref}[1]{\namedref{Table}{#1}}
\newcommand{\equationref}[1]{\hyperref[#1]{Eq~(\ref*{#1})}}
\newcommand{\theoremref}[1]{\hyperref[#1]{Theorem~\ref*{#1}}}
\newcommand{\theoremrefshort}[1]{\hyperref[#1]{Thm.~\ref*{#1}}}

\newcommand{\lemmaref}[1]{\hyperref[#1]{Lemma~\ref*{#1}}}
\newcommand{\noteref}[1]{\hyperref[#1]{note~\ref*{#1}}}
\newcommand{\appendixref}[1]{\hyperref[#1]{Appendix~\ref*{#1}}}
\newcommand{\corollaryref}[1]{\hyperref[#1]{Corollary~\ref*{#1}}}

\renewcommand{\vec}[1]{\mathbf{#1}}
\newcommand{\var}[1]{\texttt{#1}}

\DeclareMathOperator*{\E}{\mathbf{E}}

\DeclareMathOperator*{\polylog}{polylog}

\DeclareMathOperator*{\dist}{dist}
\DeclareMathOperator*{\Geom}{Geom}

\newcommand{\Bin}[2]{\operatorname{Bin} \left(#1,#2\right)}

\newcommand{\tbroadcast}{\vec B}
\newcommand{\tbroadcastg}{\tbroadcast(G)}

\newcommand{\finit}{\mathsf{init}}
\newcommand{\fout}{\mathsf{out}}

\newcommand{\hit}[1]{\vec{H}^{\mathcal P}(#1)}
\newcommand{\chit}[1]{\vec{H}(#1)}
\newcommand{\stateparam}{h}

\newcommand{\exec}[1]{\operatorname{Exec}(#1)}
\newcommand{\stable}[1]{\operatorname{Stable}(#1)}

\newcommand{\alllowcount}[1]{\operatorname{LowCount}(#1)}

\begin{document}

\begin{mycover}
	{\huge\bfseries\boldmath Near-Optimal Leader Election in Population Protocols on Graphs \par}

	\bigskip
	\bigskip

\textbf{Dan Alistarh}
\myemail{dan.alistarh@ista.ac.at}
\myaff{Institute of Science and Technology Austria}

\textbf{Joel Rybicki}
\myemail{joel.rybicki@hu-berlin.de}
\myaff{Humboldt University of Berlin}

\textbf{Sasha Voitovych}
\myemail{sasha.voitovych@mail.utoronto.ca}
\myaff{University of Toronto}

\bigskip
\end{mycover}

\medskip
\begin{myabstract}
  \noindent\textbf{Abstract.}
  In the \emph{stochastic population protocol} model, we are given a connected graph  with $n$ nodes, and in every time step, a scheduler samples an edge of the graph uniformly at random and the nodes connected by this edge interact. A fundamental task in this model is \emph{stable leader election}, in which all nodes start in an identical state and the aim is to reach a  configuration in which (1) exactly one node is elected as leader and (2) this node remains as the unique leader no matter what sequence of interactions follows. On \emph{cliques}, the complexity of this problem has recently been settled: time-optimal protocols stabilize in $\Theta(n \log n)$ expected steps using $\Theta(\log \log n)$ states, whereas protocols that use $O(1)$ states require $\Theta(n^2)$ expected steps.

In this work, we investigate the complexity of stable leader election on graphs. We provide the first non-trivial time lower bounds on general graphs, showing that, when moving beyond cliques, the complexity of stable leader election can range from $O(1)$ to $\Theta(n^3)$ expected steps. We describe a protocol that is time-optimal on many graph families, but uses polynomially-many states. In contrast, we give a near-time-optimal protocol that uses only $O(\log^2n)$ states that is at most a factor $O(\log n)$ slower. Finally, we observe that for many graphs the constant-state protocol of Beauquier et al. [OPODIS 2013] is at most a factor $O(n \log n)$ slower than the fast polynomial-state protocol, and among constant-state protocols, this protocol has near-optimal \emph{average case complexity} on dense random graphs.
\end{myabstract}

\thispagestyle{empty}
\setcounter{page}{0}
\newpage

\section{Introduction}

Leader election is one of the most fundamental symmetry-breaking problems in distributed computing~\cite{angluin1980local}: given a distributed system consisting of $n$ identical nodes, the goal is to designate exactly one node as a leader and all others as followers.  In this work, we study the computational complexity of leader election in the \emph{stochastic population protocol model}, a popular model of distributed computation among a population of (initially) indistinguishable agents that reside on a graph and interact in an unpredictable, random manner~\cite{angluin2006computation,aspnes2009introduction}.

\subsection{The stochastic population model on graphs}

In the stochastic population protocol model, or simply the \emph{population model}, the system is described by a finite, connected graph $G = (V,E)$ with $n$ nodes. Each node represents an agent, corresponding to a finite state automaton.  Initially, all nodes are identical and anonymous.

\paragraph{Model of computation}
In the population model, computation proceeds asynchronously, in a series of random pairwise interactions between neighbouring nodes in the graph $G$. In each discrete time step,
\begin{enumerate}
\item the scheduler samples an ordered pair $(u,v)$ \emph{uniformly at random} among all pairs of nodes connected by an~edge in $G$, 
\item the selected nodes $u$ and $v$ interact by exchanging information and updating their local states, and
  \item every node maps their local state to an output value. 
\end{enumerate}
When the scheduler selects the ordered pair $(u,v)$ of nodes upon an interaction step, we say that $u$ is the \emph{initiator} of the interaction and $v$ is the \emph{responder}. The algorithm is described by a state transition function, which is typically given by a collection of local update rules of the form $A + B \to C + D$, where $A$ and $B$ are the states of the initiator and the responder at the start of an interaction, and $C$ and $D$ are the resulting states after the interaction.

\paragraph{Stable leader election}
In the case of leader election, nodes have two possible output values to indicate whether they are a \emph{leader} or a \emph{follower}.
The goal is to design the local update rules so that the system reaches a \emph{stable configuration} in which (1) exactly one node $v \in V$ is elected as the leader and all other nodes are followers and (2) the node $v$ remains as the unique leader no matter what sequence of interactions follows (i.e., all  configurations reachable from a stable configuration have the same output).

\paragraph{Complexity measures}
The \emph{time complexity} is measured by \emph{stabilization time}, which is the \emph{total number of interaction steps} needed to reach a stable configuration. The typical aim is to guarantee that stabilization time is small both in expectation and with high probability. Finally, we measure \emph{space complexity} as the maximum number of distinct node states employed by the protocol.

\subsection{Prior work on leader election in the population model}

Already the foundational work on population protocols raised the question how the structure of the interaction graph influences the computational power~\cite{angluin2006computation,aspnes2009introduction,angluin2008self} and complexity~\cite{angluin2005stably} of stable computation in the population model. In particular, the complexity of stable leader election on general interaction graphs has remained an open problem. Instead, most work in this area has focused on a special case where the interaction graph is restricted to be a clique~\cite{elsaesser-survey,alistarh-survey}. While this special case naturally corresponds to \emph{well-mixed systems}, it is often too simplistic when modelling systems where the interaction patterns among agents are influenced by some underlying spatial structure.

Leader election is known to be an especially important problem in the population model: for instance, the early work of Angluin, Aspnes and Eisenstat~\cite{angluin2008fast-computation} showed that having a leader can be useful in the population model on cliques: semilinear predicates can be stably computed in $O(n \log^5 n)$ time, and randomized LOGSPACE computation can be performed with small error~\cite{angluin2008fast-computation}. The above result has motivated a vast amount of follow-up work on the complexity of leader election on \emph{cliques}~\cite{doty2018stable,alistarhg15,alistarh2017time,berenbrink2018simple,gasieniec2018almost,gasieniec2018fast,gkasieniec2020enhanced,sudo2019logarithmic,sudo2020leader,berenbrink2020optimal,elsaesser-survey,alistarh-survey}. By now, the complexity of leader election on the clique is well-understood: there exists a protocol that solves leader election in $\Theta(n \log n)$ expected steps using $\Theta(\log \log n)$ states per node~\cite{berenbrink2020optimal}, which is optimal.  To elect a leader in the clique model, all protocols require $\Omega(n \log n)$ expected steps~\cite{sudo2020leader}, any $o(\log \log n)$-state protocol requires $n^2/\polylog n$ expected steps~\cite{alistarh2017time} and the time complexity bound for constant-state protocols is $\Theta(n^2)$ expected steps~\cite{doty2018stable}.

Somewhat surprisingly, much less is known about the complexity of leader election on general interaction graphs. Angluin, Aspnes, Fischer and Jiang~\cite{angluin2008self} showed that \emph{self-stabilizing} leader election is not generally possible on all connected interaction graphs. At the same time, Beauquier, Blanchard and Burman~\cite{beauquier2013self} showed that there exists a constant-state protocol that solves stable leader election as long as self-stabilization is not required. Subsequently, research on leader election in the population model has largely fallen into two  categories: (1) work that tries to understand computational complexity and space-time complexity trade-offs of leader election under uniform random pairwise interactions on the clique~\cite{doty2018stable,alistarhg15,berenbrink2018simple,gasieniec2018almost,gasieniec2018fast,gkasieniec2020enhanced,sudo2019logarithmic,sudo2020leader,berenbrink2020optimal,timeoptimalLE}, and (2) work that aims to understand in which interaction graphs and under what model assumptions leader election can be solved in, e.g., a fault-tolerant manner~\cite{angluin2008self,beauquier2013self,chen2019self,chen2020ssle,sudo2016same,sudo2020loosely,sudo2021self,yokota2020time}.

An interesting question left open by this line of work is the \emph{computational complexity} of stable leader election, without the requirement of self-stabilization, on \emph{general interaction  graphs}~\cite{alistarh2021graphical}. One reason is that algorithmic~\cite{elsaesser-survey,alistarh-survey,gkasieniec2020enhanced,sudo2019logarithmic,berenbrink2020optimal,timeoptimalLE} and lower bound techniques~\cite{doty2018stable,alistarh2017time,sudo2020leader} developed for the clique model do not readily extend to the case of general interaction graphs. More broadly, establishing tight bounds for randomized leader election is known to be challenging even in well-studied synchronous models of distributed computing, such as the LOCAL and CONGEST models~\cite{peleg00locality, kutten2015complexity}.

\subsection{Limits of existing techniques}

Existing \emph{upper bound techniques} on the clique naturally rely on the fact that every pair of nodes can potentially interact.  Specifically, fast and space-efficient algorithms~\cite{gkasieniec2020enhanced,berenbrink2020optimal,timeoptimalLE} combine (1) fast information dissemination dynamics, typical for the clique, with (2) careful time-keeping across ``juntas of nodes'' to obtain space-efficient phase clocks. It is not straightforward to generalize either of these techniques to e.g.\ sparse or poorly-connected~graphs.

 The only existing work to explicitly consider the complexity of non-self-stabilizing leader election on graphs is by Alistarh, Gelashvili and Rybicki~\cite{alistarh2021graphical}, whose general goal is to find general ways of porting clique-based algorithms to regular interaction graphs. (Chen and Chen~\cite{chen2020ssle} considered complexity of \emph{self-stabilizing} leader election in regular graphs, but this is computationally harder than stable leader election~\cite{angluin2008self}.)
Specifically for regular graphs, Alistarh et al.~\cite{alistarh2021graphical} gave a leader election protocol that stabilizes in $1/\phi^2 \cdot n \polylog n$ steps in expectation and with high probability and uses $1/\phi^2 \cdot \polylog n$ states, where $\phi$ is the conductance of the interaction graph. While this approach yields to fairly efficient leader election protocols in graphs with high conductance, it performs poorly in low-conductance graphs. For example, on cycles the protocol uses $n^2 \polylog n$ states and requires $\Omega(n^3)$ steps to stabilize.

Alistarh et al.~\cite{alistarh2021graphical} also showed that the constant-state protocol of Beauquier et al.~\cite{beauquier2013self} stabilizes in the order of $Dmn^2 \log n$ steps in expectation and with high probability on any graph with diameter~$D$ and $m$ edges. This upper bound can be further refined to $O(\vec C(G) \cdot n \log n)$, where $\vec C(G)$ is the cover time of a classic random walk on the graph $G$, by leveraging the recent results of Sudo, Shibata, Nakamura, Kim and Masuzawa~\cite{sudo2021self}. However, beyond the case of cliques, there are no results indicating whether this bound could be improved. 

Specifically, existing \emph{lower bound techniques} for population protocols on the clique~\cite{doty2018stable,alistarh2017time,sudo2020leader} do not directly generalize to general graphs.  
In particular, such approaches  usually rely on the fact that short executions can lead to ``populous'' configurations which have large ``leader generating'' sets of nodes; then, by carefully interleaving interactions between nodes in such sets, short executions can be extended to create new leaders. This suggests that short executions are unlikely to yield stable configurations. 
However, to create new leaders, existing arguments require the set of nodes to be connected in the underlying graph. This is straightforward on the clique, but non-trivial for general graphs.

The situation seems even more challenging when trying to establish space-time complexity trade-offs, such as showing that constant-state protocols cannot run in sublinear time. In this case, the only known approach is the \emph{surgery technique}~\cite{doty2018stable,alistarh2017time}, which requires keeping track of the distribution of certain states that can be used to generate a leader. On \emph{general graphs}, one would therefore also need to 
keep track of the \emph{spatial distribution} of  states created by the protocol, which appears highly complex for general protocols and interaction graphs.

To illustrate the difficulty of extending the above techniques to general interaction graphs, a useful exercise is to consider the case of star graphs: there is a $O(1)$-state protocol that elects a leader in a single interaction in any star of size $n$. Thus, the lower bound for constant-state protocols or the general lower bound of $\Omega(n \log n)$ expected steps cannot hold in general, as in some graphs, the graph structure can be used to break symmetry~fast using a small number of states.

\subsection{Our contributions}

In this work, we give new upper bounds and lower bounds for stable leader election in the population protocol model on general graphs. For many graph families, we obtain either tight or almost tight bounds; please see \tableref{tab:summary} for a summary of our results. We detail our main results in the following.

\begin{table}[t]
{
  \begin{tabular}{@{}lllll@{}}
    \toprule
Graph family             & Stabilization time  & States & Reference \\ \midrule
General        & $O(\tbroadcast(G) + n\log n )$   & $O(n^4)$         &    \theoremref{thm:le-in-broadcast-time}    \\
                       & $O(\tbroadcast(G)\cdot\log n)$   & $O(\log^2 n)$      &    \theoremref{thm:fast-le}     \\
                       & $O( \chit{G} \cdot n\log n)$   & $O(1)$         & \cite{beauquier2013self} + \cite{sudo2021self} + \theoremref{thm:6-state-le-cover-time}  \\ \midrule
Renitent*        & $\Omega(\tbroadcast(G))$   &   $\infty$      &  \theoremref{thm:constructive-lb} &     \\
                         \midrule
                         
Regular      & $O(\phi^{-2} n \log^6 n)$   & $O( \phi^{-2}\log^7 n)$      &  \cite{alistarh2021graphical}     \\
& $O(\phi^{-1}n \log n)$   & $O(n^3)$      &    \theoremref{thm:le-in-broadcast-time}     \\
& $O(\phi^{-1} n\log^2 n)$   & $O(\log n \cdot \log (\phi^{-1} \cdot \log n))$      &    \corollaryref{cor:fast-le-regular}     \\ 

& $O(\phi^{-1} n^2 \log^2 n)$   & $O(1)$         &    \theoremref{thm:6-state-le-cover-time}   \\ \midrule

Cliques  & $\Theta(n \log n)$ & $\Theta(\log \log n)$ & \cite{sudo2020leader,alistarh2017time,berenbrink2020optimal} \\
& $\Theta(n^2)$ & $O(1)$ & \cite{doty2018stable} \\ \midrule
Dense random**  & $\Theta(n \log n)$   & $O(n^4)$         & \theoremref{thm:dense-lower-bound-general} + \ref{thm:le-in-broadcast-time}  \\
                      & $O(n \log^2 n)$   & $O(\log^2 n)$         &   \theoremref{thm:fast-le}  \\
& $o(n^2)$ impossible & $O(1)$         &  \theoremref{thm:nsquared-lb} \\ \
& $O(n^2 \log^2 n)$  & $O(1)$         &  \theoremref{thm:6-state-le-cover-time} \\ \midrule
Stars        & $O(1)$   & $O(1)$  &  Trivial \\ \bottomrule
\end{tabular}
  \caption{Complexity bounds for stable leader election. Stabilization time refers to the expected number of steps required to reach a stable configuration. Here $\tbroadcast(G)$ is a characteristic of information dissemination dynamics in the population model defined in \sectionref{sec:broadcast}. The quantity $\chit{G}$ is the worst-case expected hitting time of a simple, classic random walk on $G$.  
  The $O(\log^2n)$ and $O(1)$-state protocols also stabilize in the reported time w.h.p. For regular graphs, $\phi$ denotes the conductance of the graph. (*) We define the class of renitent graphs in \sectionref{sec:lower}. (**) For dense random graphs, the bounds are for \emph{average-case complexity}, when the input graph is an Erd\"os-Re\'nyi random graph $G \sim G_{n,p}$ for any constant $p>0$. Here, the upper bounds follow from the fact that broadcast time is $O(n \log n)$ with high probability in these graphs.
  }
\label{tab:summary}
}
\end{table}

\paragraph{Bounds on information propagation in the population model}

We phrase our upper bounds in terms of \emph{worst-case expected broadcast time} $\tbroadcastg$ on the graph $G$.  Informally, $\tbroadcastg$ denotes the maximum expected time until a broadcast originating from a single node reaches all other nodes in the graph $G$. This process is often called ``one-way epidemics'' in the population protocol literature~\cite{angluin2008fast-computation}. %

In  \sectionref{sec:broadcast}, we establish the worst-case broadcast time upper bounds of $O(mD + m\log n)$ and $O(m/\beta \cdot \log n)$ for any $m$-edge graph with diameter $D$ and edge expansion $\beta$. Thus, $\tbroadcastg \in O(n^3)$ for any connected graph, but it can often be much smaller. We also provide lower bounds on the time that information propagates to a given distance $k$. These bounds are then used to lower bound both broadcast time and leader election time for general protocols.

While for regular graphs these dynamics correspond to well-studied asynchronous rumour spreading~\cite{chierichetti2010rumour,giakkoupis2016asynchrony}, when the graph is not regular, the dynamics in the population model behave differently: the interaction rate of a node depends on its degree. 

\paragraph{Fast space-efficient leader in close-to-broadcast time}
We first observe that, if we disregard space complexity, there exists a simple protocol that solves leader election in $O(\tbroadcastg + n \log n)$ expected steps, on any connected graph $G$: 
nodes can generate unique identifiers, and then broadcast them to elect a leader. While this protocol is time-optimal on many graphs, generating unique identifiers will require polynomially-many states, and thus, results in a high space complexity.

Our first contribution is a \emph{space-efficient protocol} that elects a leader in $O(\tbroadcastg \cdot \log n)$ steps in expectation and with high probability using only $O(\log n \cdot \stateparam(G))$ states, where $\stateparam(G) \in O(\log n)$ is a parameter depending on the broadcast time $\tbroadcastg$. 
Contrasting to the identifier-based approach, this space-eficient protocol achieves exponentially smaller space complexity of $O(\log^2n)$, with a factor $O(\log n)$ increase in stabilization time.

Our protocol builds on a time-optimal approach on the clique by Sudo, Fukuhito, Izumi, Kakugawa and Masuzawa~\cite{timeoptimalLE}, and significantly improves upon the state-of-the-art on general graphs. Specifically, Alistarh, Gelashvili and Rybicki~\cite{alistarh2021graphical} gave a protocol for leader election on $\Delta$-regular graphs that stabilizes in $O(n/\phi^{2} \cdot  \log^6 n)$ expected steps and uses $O(\log^7 n / \phi^2)$ states per node, where $\phi =\beta/\Delta$ is the conductance of the graph. Our protocol has stabilization time $O( n/ \phi \cdot \log^2 n)$ on regular graphs; this improves the dependency on the conductance $\phi$ by a linear factor and the polylogarithmic dependence from $\log^6n$ to $\log^2n$. In terms of space complexity, we get an \emph{exponential} improvement in conductance, as the parameter $\stateparam$ in the space complexity bound satisfies $\stateparam \in O(\log\log n + \log(1/\phi))$ in regular graphs.

We emphasize that our protocol also works in \emph{non-regular graphs}, and guarantees that the elected leader has degree $\Theta(\Delta)$ with high probability. Our protocol has high-degree nodes driving a space-efficient and approximate distributed phase clock: nodes with degree $\Theta(\Delta)$ generate ``clock ticks'' roughly every $\tbroadcast(G)$ steps with high probability. With this in place, we devise a protocol in which high-degree nodes participate in a tournament that lasts for $O(\log n)$ phases, and each phase lasts for $O(\tbroadcast(G))$ steps.

\paragraph{Time lower bounds for general protocols}

On the negative side, we show how to construct families of graphs in which leader election and broadcast have the same asymptotic time complexity. Our approach is based on a probabilistic indistinguishability argument similar in spirit to the lower bound argument of Kutten, Pandurangan, Peleg, Robinson and Trehan~\cite{kutten2015complexity} for randomized leader election in the synchronous LOCAL and CONGEST models. However, in the population model, communication patterns are asynchronous and stochastic instead of synchronous, so we need a more refined approach to establish the lower bounds.

Roughly speaking, we show that if (a) the nodes of the graph can be divided into constantly many subsets $V_1, \ldots, V_K$ such that the local neighbourhoods of these sets are isomorphic up to some distance $\ell$ and (b) there are sets whose distance-$\ell$ neighbourhoods are disjoint, then any leader election protocol must propagate information at least up to distance $\ell$ to reach a stable configuration. If  this propagation takes at least $f(n)$ steps with at least a constant probability, then we get a lower bound of order $f(n)$ for the expected stabilization time. We call such graphs $f$-renitent (see \sectionref{sec:lower} for a formal definition).

In general, it is fairly straightforward to construct graphs with diameter $\Theta(D)$ and $\Theta(m)$ edges, which are $\Omega(Dm)$-renitent for any  $1 \le D \le n$ and $n \le m \le n^2$. Moreover, in these graphs broadcast time is $\Theta(Dm)$. Our proof  works for a general variant of the population model, in which we do not restrict the state space of the protocol and give each node an \emph{infinite} stream of uniform, fair random bits that assign unique identifiers for each node with probability 1 at the start of the execution.  Finally, we also show that in \emph{any} sufficiently dense graph, leader election requires $\Omega(n \log n)$ expected steps. This part of the argument extends the lower bound argument of Sudo and Masuzawa~\cite{sudo2020leader} from cliques to dense graphs.

\paragraph{Worst-case and average-case complexity of constant-state protocols}
As a baseline result, we observe that the constant-state protocol of Beauquier et al.~\cite{beauquier2013self} stabilizes in $O(\chit{G} \cdot n \log n)$ steps in expectation and with high probability, where $\chit{G}$ is the worst-case expected hitting time of \emph{a classic random walk} on $G$. For any connected graph, it is known that $\chit{G} \in O(n^3)$, and if the graph is regular, then $\chit{G} \in O(n^2)$ holds~\cite{levin2017markov}. 

It follows from the analyses of Alistarh et al.~\cite{alistarh2021graphical} and Sudo et al.~\cite{sudo2021self} that this protocol stabilizes in $O(\chit{G} \cdot n \log n)$ steps in expectation and with high probability. %
In the preliminary conference version of this paper~\cite{alistarh2022leader}, we erroneously reported a running time in terms of $O(\tbroadcastg \cdot n \log n)$; the analysis was unfortunately flawed. It remains an interesting question to determine if we can bound the hitting times of tokens in the population model also as a function of the broadcast time $\tbroadcastg$, since this would immediately imply a bound for the constant-state protocol as a function of the broadcast time as well. This would allow for a cleaner comparison between the running times of the different protocols.

As our final contribution, we show that, in the class of constant-state protocols, the \emph{average-case complexity} of the protocol by Beauquier et al.~\cite{beauquier2013self} on dense random graphs is optimal up to a $O(\log n)$ factor. More formally, we show that the stabilization time of any leader election protocol, that works on all connected graphs, cannot be $o(n^2)$ on any connected Erd\"os-R\'enyi random graph $G \sim G_{n,p}$ with high probability for any constant $p>0$. This is tight up to a logarithmic factor, as the hitting time satisfies $\chit{G} \in O(n)$ asymptotically almost surely in these graphs~\cite{lowe2014hitting}. Therefore, the 6-state protocol stabilizes in $O(n^2 \log n)$ steps with %
probability $1-o(1)$ %
on a random graph $G \sim G_{n,p}$. 
As $\tbroadcastg \in \Omega(n)$ for any graph $G$, this implies that, on these graphs, constant-state protocol of Beauquier et al.~\cite{beauquier2013self} stabilizes slower than optimal by a factor of at most $O(n\log n)$.

To show the lower bound, we extend the \emph{surgery technique}, used to prove space-time lower bounds for population protocols so far only in the clique~\cite{doty2018stable,alistarh2017time}, to the case of (highly) dense random graphs.

\section{Preliminaries}

We now establish some key definitions and notation used throughout the paper.

\subsection{Graphs}

Let $G = (V,E)$ be a undirected graph, where $V = V(G)$ is the set of nodes and $E(G) = E$ is the set of edges of the graph. We use $n = |V(G)|$ to denote the number of nodes and $m = |E(G)|$ the number of edges.
The degree $\deg(v)$ of a node $v \in V$ is the number of edges incident to it.
We use $\Delta = \max\{ \deg(v) : v \in V\}$ to denote the maximum degree and $\delta = \min \{ \deg(v) : v \in V\}$ the minimum degree of the graph. Unless otherwise specified, we assume all graphs are connected.

Given a nonempty node set $S \subseteq V(G)$, the edge boundary $\partial S$ of $S$ is the set
\[
\partial S = \left\{ \{u,v\} \in E(G) : u \in S, v \in V(G) \setminus S \right\}.
\]
The \emph{edge expansion} of the graph $G$ is given by
\[
\beta(G) = \min \left\{ \frac{|\partial S|}{|S|} : \emptyset \neq S \subseteq V(G), |S| \le n/2 \right \}.
\]
We define $G[S]$ to be a subgraph of $G$ induced on vertices of $S$.

The distance between two vertices $u$ and $v$ is denoted by $\dist(u,v)$. The radius-$r$ neighbourhood of $u$ is $B_r(u) = \{ v \in V : \dist(u,v) \le r \}$ and 
\[
B_r(U) = \bigcup_{u \in U} B_r(u).
\]
For $r=1$, we use the short-hand $B(u)=B_1(u)$. The diameter $D(G)$ is given by $D(G) = \max \{ \dist(u,v) : u,v \in V(G)\}$. For any two graphs $G$ and $H$, we write $G \simeq H$ if they are isomorphic.

For random graphs, we use the Erd\H{o}s--R\'enyi random graph model $G_{n,p}$. In this model, a random graph $G \sim G_{n,p}$ is sampled as follows. We start with $n$ nodes and for each $u,v \in V$ such that $u \neq v$, we add the edge $\{u,v\}$ with probability $p$ independently of all other edges.

\subsection{Population protocols on graphs}
A (stochastic) \emph{schedule} on a graph $G$ is an infinite sequence $(e_t)_{t \ge 1}$ of ordered pairs of nodes $(v,u)$, where each $e_t$ is sampled independently and uniformly at random among all pairs of nodes connected by an edge in $G$ (there are $2m$ such pairs). The order of nodes in the pair is used to distinguish between initiator and a responder.  A \emph{protocol} is a tuple $\mathcal{A} = (\Lambda, \Xi, \Sigma_{\textrm{in}}, \Sigma_{\textrm{out}}, \finit, \fout)$, where
$\Lambda$ is the set of states,
$\Xi \colon \Lambda \times \Lambda \to \Lambda \times \Lambda$ is the state transition function,
$\Sigma_{\textrm{in}}$ and $\Sigma_{\textrm{out}}$ are the sets of input and output labels, respectively,
$\finit \colon \Sigma_{\textrm{in}} \to \Lambda$ is the initialization function, and
$\fout \colon \Lambda \to \Sigma_{\textrm{out}}$ is the output function.

A configuration is a map $x \colon V \to \Lambda$, where $x(v)$ is the state of the node $v$ in configuration $x$.  For any $e = (u,v)$ and configurations $x$ and $x'$, we write $x \Rightarrow_{e} x'$ if $x'(u),x'(v) =  \Xi(x(u),x(v))$ and $x'(w) = x(w)$ for all $w \in V \setminus \{u,v\}$. For any sequence $\sigma = (e_1, \ldots, e_t)$ we write $x_0 \Rightarrow_{\sigma} x_t$ if $x_i \Rightarrow_{e_{i+1}} x_{i+1}$ for each $i \ge 0$. We say that $x'$ is reachable from $x$ on  $G$ if there exists some $k \ge 1$ and $\sigma = (e_1, \ldots, e_k)$ such that $x \Rightarrow_\sigma x'$. Given input $f \colon V(G) \to \Sigma_{\textrm{in}}$, a protocol and  a schedule $(e_t)_{t \ge 1}$, an \emph{execution} is the infinite sequence $(x_t)_{t \ge 0}$ of configurations, where $x_0 = \finit \circ f$ is the initial configuration and $x_t \Rightarrow_{e_{t+1}} x_{t+1}$ for $t \ge 0$.
Note that throughout the paper, the time step $t$ denotes the total number of pairwise interactions that have occurred so far. 

In the case of leader election, we assume that the input is a constant function, unless otherwise specified. That is, all nodes start in the same initial state. We say that a configuration $x$ is \emph{correct} if  $\fout(x(v)) = \mathsf{leader}$ for exactly one node $v \in V$ and for all $u \in V \setminus \{ v \}$ we have $\fout(x(u)) = \mathsf{follower}$. A configuration $x$ is \emph{stable} if for every configuration $x'$ reachable from $x$ we have $\fout(x(v)) = \fout(x'(v))$ for every node $v \in V$. The \emph{stabilization time} of a leader election protocol $\mathcal{A}$ is the minimum $t$ such that $x_t$ is stable and correct. The \emph{state complexity} of a protocol is  $|\Lambda|$, the number of distinct states.

Some of the protocols we consider are non-uniform in the following sense: the state space and transition function of the protocol can depend on parameters that capture high-level structural information about the population and the interaction graph (e.g., number of nodes and edges, broadcast time or the maximum degree). However, upon initialization, all nodes receive exactly the same information. For example, nodes do not initially know their own degree or identity in the interaction graph.

\subsection{Probability-theoretic tools}

Let $X$ and $Y$ be real-valued random variables defined on the same probability space. We say that $X$ \emph{stochastically dominates} $Y$, written as $Y \preceq X$, if $\Pr[X \ge x] \ge \Pr[Y \ge x]$ for all $x \in \mathbb{R}$. We start with three concentration bounds. The first is a folklore result; see e.g.~\cite{cannonne2019poisson} for a proof. The second is also standard Chernoff bounds for sums of Bernoulli random variables. The third result gives tail bounds on the sums of geometric random variables, via Janson~\cite[Theorems 2.1 and 3.1]{janson2018tail}.

\begin{lemma}\label{lemma:poisson-concentration-c}
  Let $X \sim \operatorname{Poisson(\lambda)}$ be a Poisson random variable with mean $\lambda$. Then
   \begin{enumerate}[label=(\alph*)]
   \item $\Pr[X \ge c \lambda ] \le \exp\left(-\lambda \cdot (c-1)^2/c\right)$ for $c \ge 1$,
   \item $\Pr[X \le c \lambda ] \le \exp\left(-\lambda \cdot (1-c)^2/(2-c)\right)$ for $c \le 1$.
   \end{enumerate}
\end{lemma}

\begin{lemma}
    \label{lemma:bernoulli-chernoff}
    Let $X = X = Y_1 + \ldots + Y_k$ be a sum of independent Bernoulli random variables with $\Pr[Y_i = 1] = p_i$. Then 
\begin{enumerate}[label=(\alph*)]
\item $\Pr[X \ge (1+\lambda) \cdot \E[X] ] \le \exp\left(-\E[X] \cdot \lambda^2/3\right)$ for any $\lambda \ge 1$, and
\item $\Pr[X \le (1 - \lambda) \cdot \E[X]] \le \exp\left(-\E[X] \cdot \lambda^2/2 \right)$ for any $\lambda \le 1$.
\end{enumerate}
Note that in the special case when $p_i = p$ for all $1 \le i \le k$, the sum $X \sim \Bin{k}{p}$ is a Binomial random variable.
\end{lemma}

\begin{lemma}\label{lemma:sum-of-geometric}
  Let $p_1, \ldots, p_k \in (0,1]$ and $X = Y_1 + \ldots + Y_k$ be a sum of independent geometric random variables with $Y_i \sim \operatorname{Geom}(p_i)$. Define $p = \min \{ p_i : 1 \le i \le k \}$ and $c(\lambda) = \lambda - 1 - \ln \lambda$. Then
  \begin{enumerate}[label=(\alph*)]
  \item $\Pr[ X \ge \lambda \cdot \E[X] ] \le \exp\left(- p \cdot \E[X] \cdot c(\lambda)\right)$ for any $\lambda \ge 1$, and

  \item  $\Pr[ X \le \lambda \cdot \E[X] ] \le \exp\left(- p \cdot \E[X] \cdot c(\lambda)\right)$ for any $0 < \lambda \le 1$.
  \end{enumerate}
\end{lemma}

\begin{lemma}[Wald's identity]\label{lemma:wald}
  Let $(X_i)_{i \ge 1}$ be a sequence of real-valued independent and identically distributed random variables and $N$ a non-negative integer-valued random variable independent of $(X_i)_{i \ge 1}$.
  If $N$ and all $X_i$ have finite expectation, then $\E[X_1 + \cdots + X_N] = \E[N] \cdot \E[X_1]$.
\end{lemma}

\section{Bounds on information propagation}\label{sec:broadcast}

Our results will rely on notions of \emph{broadcast time} and \emph{propagation time} in the population model. For this, we define the following infection process on a graph $G$: initially, each node $v \in V$ holds a unique message. In every step, when nodes $u$ and $v$ randomly interact they inform each other about \emph{all messages} they have so far received.

The \emph{distance-$k$ propagation time} is the minimum time until some message has reached a node at distance $k$ from its source. The \emph{broadcast time} is the expected time until all nodes in the network are aware of all messages. Propagation times are used in our lower bounds, whereas broadcast time appears in our upper bounds. Before we formalize these notions below, we briefly discuss some work on related, but different stochastic information propagation dynamics on graphs.

\subsection{Information propagation in related models}
Many variants of the above broadcasting process have been studied in settings ranging from information dissemination~\cite{demers1987epidemic,giakkoupis2016asynchrony,acan2017push,sauerwald2010mixing,karp2000-randomized,chierichetti2010rumour,chierichetti2010almost} to models of epidemics~\cite{fates2008examples,ottino2017takeover,liggett2005interacting}. For example, in the synchronous push-pull model~\cite{demers1987epidemic,karp2000-randomized}, Chierichetti, Lattanzi and Panconesi~\cite{chierichetti2010rumour} first showed that broadcast succeeds with high probability in $O(\log^4 n/\phi^6)$ rounds on graphs of conductance $\phi$. Subsequently, they improved the running time bound to $O(\log n/\phi \cdot \log^2(1/\phi) )$ rounds~\cite{chierichetti2010almost}. Finally, Giakkoupis~\cite{giakkoupis2011tight} showed that the push-pull algorithm succeeds in $O(\log n/\phi)$ rounds with high probability, and showed that for all $\phi \in \Omega(1/n)$ there is a family of graphs in which this bound is tight.

In the asynchronous setting, Acan, Collevecchio, Mehrabian and Wormald~\cite{acan2017push} and Giakkoupis, Nazari and Woelfel~\cite{giakkoupis2016asynchrony} studied broadcasting in the continuous-time push-pull model, where each \emph{node} has a (probabilistic) Poisson clock that rings at unit rate. They showed that on graphs in which the protocol runs in $T$ rounds, the asynchronous protocol runs in $O(T + \log n)$ continuous time. Ottino-Löffler, Scott and Strogatz~\cite{ottino2017takeover} studied a discrete-time infection model that is similar to this asynchronous setting, and characterized broadcast time in cliques, stars, lattices and Erd\H{o}s--R\'enyi random graphs.

Although the interaction patterns in the stochastic population model and the above asynchronous models are the same for \emph{regular graphs}, they are different in general graphs. In the population model, instead of sampling a node and then one of its neighbours in each step, our scheduler samples an \emph{edge}.
In the continuous-time setting, this corresponds to having an independent Poisson clock at each \emph{edge} rather than each \emph{node} in the network. Thus, high-degree nodes interact more often than low-degree nodes in the population model.

\subsection{Information propagation in the population model} 
We now define information propagation dynamics in our setting. Let $(e_t)_{t \ge 1}$ be a stochastic schedule on a graph $G = (V,E)$. For each node $v \in V$, let $I_0(v) = \{ v \}$. For $t \ge 0$, define
\[
I_{t+1}(v) = \begin{cases}
  I_t(v) \cup I_t(u) & \textrm{if } e_{t+1} = (u,v) \textrm{ or } e_{t+1} = (v,u) \\
  I_t(v) & \textrm{otherwise.}
\end{cases}
\]
Following Sudo and Masuzawa~\cite{sudo2020leader}, we say that $I_t(v)$ is the \emph{set of influencers} of node $v$ at the end of step $t$. Nodes in $I_t(v)$ are nodes that can (in principle) influence the state of node $v$ at step $t$. The above dynamics can be seen as a rumour spreading process, where each node starts with a unique message, and whenever two nodes interact, they inform each other about each message they hold. 

\paragraph{Broadcast and propagation time}
 Let $T(v,u) = \min \{ t : v \in I_t(u) \}$ be the minimum time until node $u$ is influenced by node $v$. The \emph{broadcast time from source $v$} is
 \[
 T(v) = \max \{ T(v,u) : u \in V(G) \}.
 \]
 We define the \emph{worst-case expected broadcast time} on $G$ to be
 \[
 \tbroadcast(G) = \max\{ \E[T(v)] : v \in V\}.
 \]
For each $k \ge 0$ and $u \in V$, let
\[
T_k(u) = \min \{ T(u,v) : v \in V, \dist(u,v)=k \}.
\]
The \emph{distance-$k$ propagation time} in $G$ is $T_k(G) = \min \{ T_k(u) : u \in V \}$. If there are no nodes at distance $k$ from node $u$, then $T_k(u) = \infty$.
Moreover, $T_k(G) = \infty$ for all $k > D(G)$.
Note that the distance-$k$ propagation time gives lower bound for the expected broadcast time as for each $1 \le k \le D(G)$ we have
\[
\E[T_k(G)] \le \E[T_D(G)] \le \tbroadcast(G).
\]

\paragraph{Sampling edge sequences}
For a finite sequence $\rho \in E^k$ of $k$ edges, let $X(\rho)$ be the number of steps until the scheduler has sampled each edge from $\rho$ in order. Note that
\[
X(\rho) = Y_1 + \cdots + Y_k
\]
is a sum of i.i.d.\ geometric random variables, where $Y_i \sim \operatorname{Geom}(1/m)$ is the number of steps until the $i$th edge of $\rho$ is sampled after sampling the $(i-1)$th edge in the sequence $\rho$. The next lemma follows immediately from \lemmaref{lemma:sum-of-geometric}.

\begin{lemma}\label{lemma:rho-concentration}
  Let $c(\lambda) = \lambda - 1 - \ln \lambda$. 
  For any
  $\rho \in E^k$, we have $\E[X(\rho)] = km$ and
  \begin{enumerate}[label=(\alph*)]
  \item $\Pr[X(\rho) > \lambda k m] \le \exp\left(-kc(\lambda)\right)$ for $\lambda \ge 1$, and
  \item $\Pr[X(\rho) < \lambda k m] \le \exp\left(-kc(\lambda)\right)$ for $0 < \lambda \le 1$.
  \end{enumerate}
\end{lemma}

With the above lemma, it is fairly straightforward to establish 
the following upper bound on the worst-case expected broadcast time $\tbroadcastg$. We give the details in the next section.

\begin{restatable}{theorem}{broadcastupper}
  \label{thm:broadcast-upper}\label{thm:broadcast-upper-bounds}
  Let $G$ be a graph with $n$ nodes, $m$ edges, edge expansion $\beta$ and diameter $D$. Then the worst-case expected broadcast time satisfies
  \[
\tbroadcastg \in O\left( m \cdot \min \left\{ \frac{\log n}{\beta}, \log n + D \right\} \right).
  \]
\end{restatable}

Note that there are graphs in which $\ln n / \beta > D$, e.g. cycles, and $\ln n /\beta < D$, e.g., cliques. We will later give leader election protocols whose stabilization time is bounded as a function of $\tbroadcastg$ on any graph $G$. In general, for any increasing function $T$ between $\Omega(n \log n)$ and $O(n^3)$, we can find families of graphs in which both the expected broadcast time and leader election time are $\Theta(T)$. We give the construction in \sectionref{sec:lower}.

\subsection{Upper-bounding the broadcast time}

In this section, we give the proof of \theoremref{thm:broadcast-upper}. We first derive an upper bound of order $m \cdot (\ln n + D)$ and then a bound of order $m \cdot \ln n / \beta$.

\begin{lemma}\label{lemma:diameter-tail-bound}
  There exists a constant $n_0$ such that for all $n > n_0$ and $\lambda \ge 1$,
  \[
  \Pr[T(G) > \lambda m \cdot \max \{ 6 \cdot \ln n, D \} ] \le 1/n^{\lambda}.
  \]
\end{lemma}
\begin{proof}
  Let $u, v \in V$. We first bound the probability that the propagation time $T(u,v)$ is large.
  Consider a shortest path $\rho$ of length $k$ between $u$ and $v$.
  Observe that $T(u,v) \preceq X(\rho)$.
  Define
  \[
  t = \lambda m \cdot \max \{ 6 \ln n, D \} \quad \textrm{ and } \quad \eta = (\lambda/k) \cdot \max \{ 6 \cdot \ln n, D \}.
  \]
  Note that there exists a constant $n_0$ so that for all $n > n_0$ the inequality $\eta - 1 - \ln \eta > \eta/2$ holds. Suppose $n > n_0$. Observing that $\eta \ge 1$ and $\eta \cdot \E[X(\rho)] = \eta k m = t$, we can apply 
  \lemmaref{lemma:rho-concentration} to get
\begin{align*}
\Pr[T(u,v) > t] &\le \Pr[X(\rho) > t] = \Pr[ X(\rho) > \eta \cdot \E[X(\rho)]] \\
&\le \exp( - k (\eta - 1 - \ln \eta) ) 
\le \exp( - k \eta / 2 ) \le \exp(-3\lambda \ln n) = 1/n^{3\lambda},
\end{align*}
where the second to last step follows, as  $\eta - 1 - \ln \eta > \eta/2$ holds.
Using union bound and $\lambda \ge 1$, we get
\begin{align*}
  \Pr[T(G) > t ] &= \Pr[ \max\{ T(u,v) : u,v \in V \} > t ] \\
  &\le \sum_{ u,v \in V} \Pr[T(u,v) > t] \le
  1/n^\lambda.
\end{align*}
\end{proof}

\begin{lemma} \label{lemma:diameter-broadcast}
  There exists a constant $n_0$ such that for all $n \ge n_0$,
  \[
  \tbroadcastg \le m \cdot \max \{ 6 \cdot \ln n, D \} + 2.
  \]
\end{lemma}
\begin{proof}
  Let $t^* = m \cdot \max \{ 6 \cdot \ln n, D \}$. Using \lemmaref{lemma:diameter-tail-bound}, we get 
  \begin{align*}
    \tbroadcastg &= \sum_{t=0}^\infty \Pr[T(G) > t]
    =\sum_{t=0}^{t^*} \Pr[ T(G) > t ] + \sum_{t=t^*+1}^{\infty} \Pr[ T(G) > t] \\
    &\le t^* +  1 + \sum_{t=t^* + 1}^{\infty} \Pr[T(G) > t] \\
    &\le t^* +  1 + \int_{0}^\infty \Pr[T(G) > (1+x)\cdot t^*] \ dx \\
    &\le t^* + 1 + \int_{0}^\infty \frac{1}{n^{1+x}} \ dx \le t^* + 2. \qedhere
  \end{align*}
\end{proof}

Next, we establish an upper bound of order $m \log n / \beta$ on the broadcast time using edge expansion. We will use a similar proof strategy as above.

\begin{lemma}\label{lemma:expansion-broadcast-whp}
  If $G$ has edge expansion $\beta > 0$ and $n \ge 2$, then for any $\lambda \ge 1$ the broadcast time satisfies
  \[
\Pr\left[T(G) \ge \lambda \cdot \frac{ 2m \cdot \log n}{\beta}\right] \le \left(\frac{1}{n-1}\right)^{ \lambda - e  - \ln \lambda}.
  \]
\end{lemma}
\begin{proof}
  As the claim is vacuous for $n=2$, we may assume that $n > 2$. Consider a node $v \in V$. We first bound the broadcast time  $T(v)$ from node $v \in V$. Let
$S_t = \{ u : v \in I_t(u) \}$
be the set of nodes influenced by node $v$ at time step $t$.

Note that for any time step $t \ge 0$, the event $|S_{t+1}| > |S_t|$ occurs if and only if the scheduler samples an edge from the boundary $\partial S_t$ of the set $S_t$.
By definition of edge expansion, $|\partial S_t| \ge \beta |S_t|$ when $|S_t|\le n/2$ and $|\partial S_t| \ge \beta |V \setminus S_t|$ otherwise.
Since the scheduler samples edges in each step independently from all other steps, the probability of this event is at least $\beta i / m$ when $|S_t|=i \le n/2$ and at least $\beta (n-i)/m$ when $n/2 \le i \le n-1$. Thus, the number $X_{i}$ of steps it takes from the set of influenced nodes to grow from $i$ to $i+1$ is stochastically dominated by the geometric random variable $Y_{i} \sim \operatorname{Geom}(p_{i})$ with $p_{i} = \beta i / m$ for $1 \le i \le  n/2$  and $p_{i+1} = \beta (n-i)/m$ for $n/2 \le i \le n-1$. 

Let $Y = Y_1 + \cdots + Y_{n-1}$. The minimum time until all nodes are influenced by $v$ satisfies
$T(v)  \preceq Y$. By linearity of expectation, we get
\[
\E[Y] = \E\left[\sum_{i=1}^{n-1} Y_i\right] = \sum_{i=1}^{n-1} \E\left[Y_i\right]  = \sum_{i=1}^{n-1} \frac{1}{p_i} = \frac{m}{\beta} \cdot C(n),
\]
where $H_{n-1} \le C(n) \le 2H_n$ and $H_n$ denotes the $n$th harmonic number. 
Since  each $p_{i+1} \ge \beta/m$,  \lemmaref{lemma:sum-of-geometric} yields
\begin{align*}
  \Pr\left[T(v) \ge \frac{2\lambda m \log n}{\beta} \right] &\le \Pr\left[Y \ge  \frac{2\lambda m \log n}{\beta} \right] \le \Pr\left[Y \ge  \lambda \E[Y] \right] \\
&\le \exp\left(- H_{n-1} \cdot (\lambda - 1 - \ln \lambda)\right) \le \left(\frac{1}{n-1}\right)^{\lambda - 1 - \ln \lambda}. 
\end{align*}
where in the last inequality we used the fact that the harmonic number satisfies $H_n \ge \ln n$. Then by the union bound we get
\begin{align*}
  \Pr\left[T(G) \ge \lambda \cdot \frac{ 2m \cdot \log n}{\beta}\right] &\le n \left(\frac{1}{n-1}\right)^{\lambda - 1 - \ln \lambda} 
\le \left(\frac{1}{n-1}\right)^{\lambda - 2 - \ln \lambda - \varepsilon}
\end{align*}
for any $\varepsilon > 0$ such that $1 + \varepsilon \ge \ln n/\ln(n-1)$.
Since $n \ge 3$, the claim follows by setting $\varepsilon = e-2$.
\end{proof}

\begin{lemma}\label{lemma:expansion-broadcast}
There exists a constant $\lambda_0 \ge 2$ such that for any graph $G$ with edge expansion $\beta > 0$, the expected broadcast time satisfies
\[
\tbroadcastg \le (2\lambda_0 m \log n)/\beta + 2.
\]
\end{lemma}
\begin{proof}
  Choose $\lambda_0 \ge 2$ such that $\lambda - e - \ln \lambda \ge \lambda/2 \ge 1$ for all $\lambda \ge \lambda_0$. Set $t^* = (2\lambda_0 m \log n)/\beta$. Now
  \begin{align*}
    \tbroadcastg = \E[T(G)] &= \sum_{t=0}^\infty \Pr[T(G) > t] \\
    &= \sum_{t=0}^{t^*} \Pr[T(G) > t] +  \sum_{t=t^*+1}^{\infty} \Pr[T(G) > t] \\
    &\le t^* + 1 + \sum_{t=t^*+1}^{\infty} \Pr[T(G) > t].
  \end{align*}
  The claim now follows by using \lemmaref{lemma:expansion-broadcast-whp} with $\lambda(x) = (1+x)\lambda_0$, as we get
  \begin{align*}
    \sum_{t=t^*+1}^{\infty} \Pr[T(G) > t] &\le \int_{0}^\infty \Pr[T(G) > t^*(1+x)] \ dx  \\
    &\le  \int_{0}^\infty \left(\frac{1}{n-1}\right)^{(1+x)\lambda_0 - e - \ln((1+x)\lambda_0)} \ dx  \\
    &\le  \int_{0}^\infty \left(\frac{1}{n-1}\right)^{(1+x)\lambda_0/2} \ dx \\
    &\le \int_{0}^\infty \left(\frac{1}{n-1}\right)^{(1+x)} \ dx \le 1 . \qedhere
  \end{align*}
\end{proof}

Now \lemmaref{lemma:diameter-broadcast} and \lemmaref{lemma:expansion-broadcast} together imply the bounds given in \theoremref{thm:broadcast-upper}.
Finally, we briefly remark that broadcast time on dense random regular graphs is $O(n \log n)$ with high probability.

\begin{lemma}
  \label{remark:gnp-broadcast}
  Let $p > 0$ be a constant and $G \sim G_{n,p}$ with $G$ conditioned on being connected. Then $\tbroadcast(G) \in O(n\log n)$ in expectation and with high probability. 
\end{lemma}
\begin{proof} 

In the following, we say that an event holds with exponential probability if it happens with probability at least $1-\exp\left(-n^{\Omega(1)}\right)$.
It is well-known that sufficiently dense random graphs have a small spectral gap. For example, from \cite[Theorem 1.1]{Hoffman2019SpectralGO} it follows that the spectral gap of a normalized Laplacian is $o(1)$ with exponential probability for Erd\"os-R\'enyi graphs $G_{n,p}$ with $p \in \Theta(1)$.

From the Cheeger inequality, it follows that conductance $\phi$ of the graph is $1 - o(1)$, which implies that $m/\beta \in n(1 + o(1))$, since an Erd\"os-R\'enyi graph with $p \in \Theta(1)$ is almost regular: all degrees are within $[(1 - \varepsilon)(n-1)p, (1 + \varepsilon)(n-1)p]$ with exponential probability. Additionally, $m/\beta \in O(n^3)$ deterministically for any connected graph $G$. Thus, \theoremref{thm:broadcast-upper} implies that $\tbroadcast(G) \in O(m/\beta \cdot \log n)$ in expectation and with high probability for $G \sim G_{n,p}$ conditioned on connectivity. 
\end{proof}

\subsection{Lower bounds for propagation and broadcast time}

We now establish lower bounds on the propagation time. This implies lower bounds for broadcast time. Note that $T(G) \ge n/2$ is a trivial lower bound, as every node needs to interact at least once for the all nodes to become influenced by the source node. We start with a simple bound that applies to any graph.

\begin{lemma}\label{lemma:nlogn-bc-lb}
For any graph $G$ with maximum degree $\Delta$, we have
\[
\tbroadcast(G) \ge  m/\Delta \cdot \ln(n-1).
\]
\end{lemma}
\begin{proof}
  Let $v \in V$. In any given step, the probability that the number of nodes influenced by node $v$ increases from $i$ to $i+1$ by one is at most $p_{i+1} = \Delta i/m$. Let $X_{i+1}$ be the number of steps it takes from the set of nodes influenced by $v$ to grow from $i$ to $i+1$. Now $X_{i+1}$ stochastically dominates the geometric random variable $Y_{i+1} \sim \Geom(p_{i+1})$. Therefore, 
  \[
  \E[T(v)] = \sum_{i=1}^{n-1} \E[X_{i+1}] \ge \sum_{i=1}^{n-1} \E[Y_{i+1}]  = \sum_{i=1}^{n-1} \frac{m}{\Delta i} = \frac{m}{\Delta} \cdot H_{n-1},
  \]
  where the harmonic number $H_{n-1}$ satisfies $H_{n-1} \ge \ln (n-1)$.
\end{proof}

For the lower bounds, we bound the \emph{distance-$k$ propagation times} $T_k(G)$ using the notion of an obstructing set, which acts as a bottleneck for information propagation. For any $r,k \in \mathbb{N}$ such that $r \le k$, we say that a set $K \subseteq E^r$ of length-$r$ edge sequences is an \emph{$(r,k)$-obstructing} set for node $v \in V$ if every path $\rho(v,u)$ from $v$ to a node $u$ with $\dist(u,v)=k$ contains some $\sigma \in K$ as a subsequence. The next lemma is useful for bounding propagation times.

\begin{restatable}{lemma}{obstructinglemma}
\label{lemma:obstructing-lemma}
Let $K$ be a $(r,k)$-obstructing set for $v$. Then for any $0 < \lambda \le 1$, we have
\[
\Pr[T_k(v) < \lambda rm] \le |K| \cdot e^{-r c(\lambda)},
\]
where $c(\lambda) = \lambda - 1 - \ln \lambda$.
\end{restatable}
\begin{proof}
  Let $P$ be the set of paths of from $v$ to any node $u$ at distance $k$ from $v$.
  Since any path $\rho \in P$ contains some $\sigma \in K$ as a subsequence,
  if the stochastic schedule contains $\rho$ as a subsequence, then it also contains $\sigma$ as a subsequence.
  Therefore, by using the union bound and \lemmaref{lemma:rho-concentration}, we get that 
  \begin{align*}
    \Pr[T_k(v) < \lambda rm] &\le \sum_{\rho \in P} \Pr[ X(\rho) < \lambda rm] \\
    &\le \sum_{\sigma \in K} \Pr[ X(\sigma) < \lambda rm] \\
    &\le |K| \cdot e^{-rc(\lambda)}. \qedhere
  \end{align*}
\end{proof}

\begin{lemma}\label{lemma:max-degree-propagation}
  If $G$ is a graph with maximum degree $\Delta$ and $k \ge \ln n$, then
  \[
  \Pr\left[T_k(G) < (km)/(\Delta e^3) \right] \le 1/n.
  \]
\end{lemma}
\begin{proof}
  Consider some node $v \in V$ and let $K$ be the set of all paths of length $k$ originating from $v$. If no such paths exist, then $T_k(v) = \infty$. Observe that $K$ is a $(k,k)$-obstructing set for $v$ and that $|K| \le \Delta^{k} = \exp( k\ln \Delta)$.
  Choose $\lambda = 1/(e^3 \cdot \Delta)$ and observe that
\begin{align*}
  \ln \Delta - c(\lambda) &< \ln \Delta - (\lambda - 1 - \ln \lambda) \\
  &= \ln \Delta - \lambda + 1 + \ln \lambda  \\
  &= \ln \Delta - \frac{1}{e^3 \Delta} + 1 - \ln(e^3 \Delta) < -2.
\end{align*}
Applying \lemmaref{lemma:obstructing-lemma} with $K$ and $\lambda$ yields
\begin{align*}
  \Pr\left[T_k(v) < \frac{km}{\Delta e^3} \right] &\le |K| \cdot e^{-k \cdot c(\lambda)} \\
  &\le
  e^{k \ln \Delta} \cdot e^{-kc(\lambda)} =
  e^{k(\ln \Delta - c(\lambda))} \le e^{-2k} \le 1/n^2,
\end{align*}
since $k \ge \ln n$.
The lemma follows by taking the union bound over all $v \in V$.
\end{proof}

\begin{restatable}{theorem}{regbroadcastlb}
\label{thm:regular-broadcast-lb}
  Let $G$ be a graph with $n$ nodes and diameter $D$. Then the following hold:
  \begin{itemize}
    \item If $G$ is regular, then $\tbroadcastg \in \Omega( n \cdot \max\{ D, \log n\} )$.
    \item If $G$ is a bounded-degree graph, then $\tbroadcastg \in \Theta( n \cdot \max\{ D, \log n\} )$.
  \end{itemize}
\end{restatable}

\begin{proof}
  Note that for both graph classes \lemmaref{lemma:nlogn-bc-lb} gives the lower bound $\Omega(n \log n)$.
  Thus, it suffices to consider the case $D \ge \ln n$. Suppose that $D \ge \ln n$ and let $t = (Dm)/(\Delta e^3)$.  \lemmaref{lemma:max-degree-propagation} implies that $\Pr[T_D(G) < t] \le 1/n$. Since $T_D(G) \le T(G)$, we have $\E[T_D(G)] \le \tbroadcastg$ by monotonicity of expectation. Now
\[
\tbroadcastg \ge \E[T_D(G)] \ge \E[T_D(G) \mid T_D(G) \ge t] \cdot \Pr[T_D(G) \ge t] \ge t(1 - 1/n). 
\]
The first lower bound follows from observing that for regular graphs $m = \Delta n/2$. The lower bound for the second claim follows by observing that $\Delta$ is a constant. \theoremref{thm:broadcast-upper-bounds} gives the upper~bound.
\end{proof}

\section{Baselines for stable leader election on graphs}\label{sec:baseline}

 In this section, we discuss two protocols, which act as our baselines for time complexity and space complexity of stable leader election on general graphs. First, we note that constant-state protocol given by Beauquier et al.~\cite{beauquier2013self} stabilizes in $O(\chit{G}\cdot n \log n)$ steps in expectation and with high probability, where $\chit{G}$ is the worst-case hitting time of a simple, \emph{classic random walk}~\cite{levin2017markov}. Second, we observe that if we allow for polynomially-many states, then there is a simple protocol that elects a leader in $O(\tbroadcastg + n \log n)$ expected steps. This protocol is time-optimal for a large class of graphs.

\subsection{First baseline: A space-efficient protocol}
Our baseline for \emph{space-efficient protocols} is the constant-state leader election protocol given by Beauquier et al.~\cite{beauquier2013self}. This protocol stabilizes in any connected graph in finite expected time.

Here, we first observe the following bound on the stabilization time as a function 
worst-case hitting time of a \emph{classic} random walk. This result essentially follows from the recent techniques developed by Sudo et al.~\cite{sudo2021self} combined with the algorithm of Beauquier et al.~\cite{beauquier2013self}.

\begin{restatable}{theorem}{lecovertime}
  \label{thm:6-state-le-cover-time}
  Given a nonempty set of leader candidates as input, there is a 6-state protocol that elect exactly one candidate as a leader in $O(\chit{G}\cdot n \log n)$ steps in expectation and with high probability, where $\chit{G}$ is the worst-case hitting time of a classic random walk on $G$.
\end{restatable}

\paragraph{Overview of the protocol}
The idea of the protocol is simple: as input, we are given a nonempty set of leader candidates. i.e., each node receives a single bit as local input denoting whether it is a leader candidate or not. At the start of the execution, each leader candidate creates a ``black token''. In each interaction, the selected nodes swap their tokens with each other. Whenever two black tokens meet, one of them is colored white and the other token is left black. Whenever a leader candidate receives a white token, this candidate becomes a follower and removes the token from the system. Eventually exactly one black token and leader candidate remain. That is, the protocol is always correct.

The formal description of this protocol appears in~\cite{beauquier2013self} and \cite{alistarh2021graphical}. In particular, \cite{beauquier2013self} gives a proof of correctness, i.e., that the protocol eventually stabilizes in any connected graph. Similarly to Alistarh et al.~\cite{alistarh2021graphical}, we exploit this property by using this constant-state protocol as a \emph{backup protocol} for faster protocols that may fail to elect a unique leader with a polynomially small probability.

\paragraph{Comparison to previous analyses}

The main idea for the analysis is that black and white tokens in the protocol perform random walks on the interaction graph $G$. Recently, Alistarh et al.~\cite{alistarh2021graphical} analysed this protocol using the meeting and hitting times of randomly walking tokens in the population model. They obtain a bound on these times through the diameter and the number of edges of the graph. Sudo et al.~\cite{sudo2021self} provide a more refined analysis of the meeting and hitting times in the population model by relating them to hitting times of a classic random walk, which are well-understood for large families of graphs. 

In this section, we adapt these results to obtain a upper bound on the stabilization time of the protocol of Beauquier et al.~\cite{beauquier2013self}. 
In \sectionref{sec:constant-lb}, we show that the \emph{average-case time complexity} of this protocol on dense random graphs is almost-optimal among constant-state protocols.

\paragraph{Classic random walks and random walks in the population model}

First, we give a definition of a simple \emph{random walk in the population model}. Let $(e_t)_{t \ge 1}$ be the sequence of (undirected) edges sampled by the stochastic scheduler. We use $X_t(v)$ to represent the position of a random walk started at node $v$ at time $t \ge 0$. The dynamics of the random walk are given by 
\[
X_{t+1}(v) =
\begin{cases}
    w & \text{ if } e_{t + 1} = \{X_t(v), w\} \\
    X_t(v) & \text{ otherwise}.
\end{cases}
\]
In other words, random walk travels to the other end of an edge, if edge containing its current position is sampled. 

In contrast, we use $X_t^{\mathcal S}(v)$ to denote the position of a \emph{classic} random started at $v$ at time $t\ge0$. For a classic random walk, we have $X_0^{\mathcal S}(v) =v$ and the conditional distribution $X_{t+1}^{\mathcal S}(v) \mid X_t^{\mathcal S}(v)$ is uniform on the neighbourhood of $X_t^{\mathcal S}(v)$. That is, if the classic random walk is at node $u$, then the next location of the random walk is sampled uniformly at random among the neighbours of $u$. 

\paragraph{Hitting times}
Let $\hit{u,v}$ denote the expected time for a random walk started at node $u$ to reach node $v$ in the population model. Define $\hit{G} = \max \{ \hit{u,v} : u,v \in V(G) \}$ to be the worst-case expected hitting time in the population model. Similarly, let $\chit{v,u}$ denote the expected hitting time from $u$ to $v$ of a classic random walk, and define $\chit{G}= \max \{ \chit{u,v} : u,v \in V(G) \}$.
Sudo et al.~\cite[Lemma 2]{sudo2021self} show the following relationship between $\chit{G}$ and $\hit{G}$.

\begin{lemma}\label{lemma:hit-broadcast}
  For any graph $G$, we have $\hit{G} \le  27n \cdot \chit{G}$.
\end{lemma}

\paragraph{Meeting times}
We say that a random walk \emph{meets} another random walk at time step $t$, if the two random walks are located at the opposite ends of the edge $e_t$ sampled at time step $t$. Let $\vec M(u,v)$ denote the expected time until the random walks started at nodes $u$ and $v$ meet.

Following Sudo et al.~\cite{sudo2021self}, we note that the meeting times in the population model can be bounded using the hitting times using essentially the same proof as the proof of Coppersmith, Tetali and Winkler~\cite[Theorem 2]{coppersmith1993collisions} for classic random walks on graphs, with minor modifications. This yields the next lemma; for the sake of completeness, we provide its proof in Appendix~\ref{apx:meet-broadcast}.

\begin{lemma}\label{lemma:meet-broadcast}
  For any $u \neq v$, we have $\vec M(u,v) \le 2\cdot \hit{G}$.
\end{lemma}

With the above lemma, we can now prove the next result, which slightly refines the bounds on hitting times given by Sudo et al.~\cite{sudo2021self} to also hold with high probability in addition to expectation. 

\begin{lemma}\label{lemma:all-hit-meet-broadcast}
  Suppose we start a random walk at every node of $G$.
  Then each random walk visits every node and meets every other random walk within $O(\chit{G}\cdot n \log n)$ time in expectation and with high probability.
\end{lemma}
\begin{proof}
  Let $v,u \in V$ and consider a random walk started at vertex $v$. By \lemmaref{lemma:hit-broadcast}, the expected hitting time $\hit{w,u}$ is upper bounded by $27n\cdot\hit{G}$ for any node $w \in V$. 

  Let $k \ge 3$ be a constant
and consider $k\log n$ time intervals of length $54n\cdot\hit{G}$. On each such interval, the random walk started at node $v$ hits node $u$ with probability at least $1/2$ by Markov inequality, regardless of its position at the beginning of the interval. Hence, it hits node $u$ in $54k \cdot \chit{G}\cdot n \log n $ steps with probability at least
  \[
  1 - \left(\frac{1}{2}\right)^{k\log n} \ge 1 - \frac{1}{n^k}.
\]
  Then, by union bound over all pairs of vertices, every random walk hits every node in $O( \chit{G}\cdot n \log n)$ steps in expectation and with high probability.
    
    Now let $v,u \in V$, and consider random walks started at vertices $v$ and $u$. By \lemmaref{lemma:meet-broadcast} and \lemmaref{lemma:hit-broadcast}, the expected meeting time $\vec M(v,u)$ is upper bounded by $54n\cdot \chit{G}$. Consider $k\log n$ intervals of length $108n\cdot\chit{G}$. On each interval, the random walk started at $v$ meets a random walk started $u$ with probability at least $1/2$ by Markov inequality, regardless of their positions at the beginning of the interval. Hence, they meet in $108k\cdot\chit{G}\cdot n \log n$ steps with probability at least
    \[
    1 - \left(\frac{1}{2}\right)^{k\log n} \ge 1 - \frac{1}{n^k}.
    \]
    Then, by union bound over all pairs of vertices, every pair of random walks meet in $O(\chit{G} \cdot n \log n)$ time in expectation and with high probability, since $k \ge 3$ was chosen to be an arbitrary constant.
\end{proof}

\paragraph{Analysis of the token-based protocol}

We can finally apply the above results to analyze the token dynamics of the constant-state leader election protocol. It is easy to see that each individual black or white token follows the simple random walk in the population model defined above.

For the purposes of our analysis, when a white token disappears, we replace it by a ``white ghost token''. Similarly, we treat a black token that turns white as a ``black ghost token'', while keeping the white token as well. This is done so that we can reason about the meeting/hitting times without worrying about the fact that black or white token may disappear before a meeting/hitting event occurs. With the above, we are now ready to prove \theoremref{thm:6-state-le-cover-time}.

\lecovertime*
\begin{proof}
   First, the number of steps until there remains exactly one black token in the system is $O(\chit{G} \cdot n \log n)$ steps in expectation and with high probability. This follows from \lemmaref{lemma:all-hit-meet-broadcast}, as all black tokens (or their ghosts) have met after $O(\chit{G} \cdot n \log n)$ steps (in expectation and w.h.p.), and every time two (non-ghost) black tokens meet, one of them turns into a white token. This implies only one black token remains in the system by this time. Indeed, if there were at least two black tokens, they would meet before that, causing one of them to turn into a white token. Moreover, the number of black tokens can never be less than one.

  Suppose $w$ white tokens remain after all but one black token have been eliminated. This means there are $w + 1$ leaders in the system. Now \lemmaref{lemma:all-hit-meet-broadcast} implies that each white token (or its ghost) will hit every one of the $w + 1$ leader nodes after $O(\chit{G}\cdot n \log n)$ steps in expectation and with high probability. Therefore, there is exactly one leader candidate left in the system after $O(\chit{G}\cdot n \log n)$ steps in expectation and with high probability. At this point the protocol has stabilized.
 \end{proof}

The bound in~\theoremref{thm:6-state-le-cover-time} simplifies, when a random Erd\"{o}s-Renyi graph is considered. The result below is a direct consequence of a bound given by L\"owe and Torres~\cite[Corollary 2.1]{lowe2014hitting}.

\begin{proposition}
    Let $G \sim G_{n,p}$ with $p \in \log^{\omega(1)}(n)/n$. Then, with probability $1-o(1)$, we have
    \[
    \chit{G} \in O(n).
    \]
\end{proposition}

\subsection{Second baseline: A time-efficient protocol}
We next discuss our baseline for \emph{fast protocols}. For this, we use a simple protocol that stabilizes in $O(\tbroadcast(G) + n \log n)$ expected steps using polynomially many states. The results in \sectionref{sec:lower} show that this protocol is time-optimal for a large class of graphs, as there are graphs where leader election requires $\Theta(\tbroadcastg)$ steps and $\tbroadcastg \in \Omega(n \log n)$.

In this protocol, we first generate unique identifiers with high probability by using the stochasticity of the scheduler and a large state space. Once we have unique identifiers, we can elect the node with the largest identifier as the leader by a broadcast process.

The only non-trivial part is to get finite expected stabilization time. This is achieved by interleaving the always-correct constant-state protocol with a broadcasting process: Once a node has generated its identifier, the node starts an instance of the constant-state protocol labelled with its own identifier and designating itself as a leader candidate in this instance. If a node encounters an instance of the constant-state protocol labeled with an identifier that is higher than the identifier of its current instance (or its own identifier), the node joins as a follower to the instance with the higher identifier. In the case that two or more nodes generated the same (highest) identifier, the constant-state protocol ensures that eventually only one leader candidate remains. 
Specifically, we show the next result.

\begin{restatable}{theorem}{lebroadcast}\label{thm:le-in-broadcast-time}
  There is a protocol that uses $O(n^4)$ states on general graphs and $O(n^3)$ states on regular graphs that elects a leader in $O(\tbroadcast(G) + n \log n)$ steps in expectation.
\end{restatable}

\paragraph{Description of the time-efficient protocol}
Let $\Lambda$ denote the set of states of the constant-state protocol $\mathcal{A}$ given by \theoremref{thm:6-state-le-cover-time}. Recall that the protocol can be given an input of a nonempty \emph{subset} of nodes out of which the leader is elected.

Let $\finit \colon \{ \mathsf{leader}, \mathsf{follower} \} \to \Lambda$ be the initialization function of $\mathcal{A}$ so that $\finit(\mathsf{leader})$ gives the state in which a node starts as a leader candidate and $\finit(\mathsf{follower})$ the state in which the node has been designated as a follower.
Each node $v \in V$ maintains two local variables:
\begin{itemize}
  \item $\var{state}(v) \in \Lambda$, and
  \item $\var{id}(v) \in \{1, \ldots, 2^{k+1} - 1\}$, where $k \in \Theta(\log n)$ is a parameter.
\end{itemize}
At the start of the protocol, each node $v$ initialises its state variables to
\[
\var{id}(v) \gets 1 \quad \textrm{ and } \quad \var{state}(v) \gets \finit(\mathsf{follower}).
\]
During an interaction $(v_0,v_1)$, where $v_0$ is the initiator and $v_1$ is the responder, node $v_i$ updates its state by applying the following rules in sequence:

\begin{enumerate}
\item If $\var{id}(v_i) < 2^k$, then set
  \[
  \var{id}(v_i) \gets 2 \cdot \var{id}(v_i) + i.
  \]
  If now $\var{id}(v_i) \ge 2^k$, then initialise
  \[ \var{state}(v_i) \gets \finit(\mathsf{leader}).
  \]

\item If $\var{id}(v_i) < \var{id}(v_{1-i})$ and $\var{id}(v_{1-i}) \ge 2^k$, then set
 \[
 \var{id}(v_i) \gets \var{id}(v_{1-i}) \quad \textrm{ and } \quad \var{state}(v_i) \gets \finit(\mathsf{follower}).
 \]
\item Update $\var{state}(v_i)$ as per the rules of the constant-state protocol $\mathcal{A}$, by using $\var{state}(v_0)$ and $\var{state}(v_1)$ as input for the transition function of $\mathcal{A}$. 

\end{enumerate}

In each time step, the output of node $v$ is defined to be the output of the constant-state protocol $\mathcal{A}$ in state $\var{state}(v)$.

\paragraph{Analysis of the time-efficient protocol}

We say that \emph{node $v$ starts an instance of $\mathcal{A}$ with the identifier $j$} if it executes $\var{state}(v) \gets \finit(\mathsf{leader})$ and $\var{id}(v) = j$ at some step $t>0$. This can only happen when applying Rule (1). If node executes $\var{state}(v) \gets \finit(\mathsf{follower})$ in Rule (2) we say that \emph{node $v$ joins an instance of $\mathcal{A}$}.

Note that a node can join many instances, but it if it joins a new instance, the new instance must have a higher identifier than the previous instance. Moreover, node can start an instance only once and this must happen before the node joins any instance. While two nodes can start an instance with the same identifier, they do so with a small probability, as shown next.

\begin{lemma}\label{lemma:unique-ids}
  Suppose node $u$ starts an instance with identifier $i$ and node $v$ starts an instance with identifier $j$. Then $\Pr[i = j] \le 1/2^k$. 
\end{lemma}
\begin{proof}
We now bound the probability of the event
that $u$ and $v$ generate the same identifier (which happens by applying the first rule $k$ times). For node $v$, the probability that it is an initiator on any of its interactions is $1/2$. This means that the identifier of $v$ will be uniformly distributed across the set $\{2^k,\ldots,2^{k+1} - 1\}$. While the identifiers are \emph{not} in general independent, we show that the identifiers will be the same with probability at most $1/2^k$. To this end, we distinguish three cases:
  \begin{enumerate}

  \item If $u$ and $v$ assign their $b$th bits at the same time step $t>0$, then one of them will be the initiator and one the responder. Therefore, the $b$th bit and the identifiers of $u$ and $v$ will be different with probability 1.

  \item If $u$ and $v$ never interact before assigning all of their $k$ bits in their identifiers, then their identifiers are different with probability $1/2^k$. In particular, the $b$th bits will be different with probability $1/2$.

\item If $v$ assigns its $b$th bit at an earlier time step $t>0$ than node $u$, then $v$ assigns the bit at position $b$ to 0 or 1 with equal probability, regardless of who is its interaction partner or previous bit assignments of $v$ and $u$. Thus, the probability that $u$ and $v$ have the same bit in position $b$ is at most $1/2$.
  \end{enumerate}
  From the above, we have that nodes $u$ and $v$ have the same $b$th bit, for $1 < b \le k+1$, in their generated identifiers for with probability at most $1/2$. Therefore, the probability that all of the $k$ assigned bits are equal is at most $1/2^k$.
\end{proof}

Define $M(t)$ to be the maximum $\var{id}$ value in the system at time step $t$ and $M = \max \{ M(t) : t \ge 0 \}$. Let $T$ be the minimum time until all nodes have either started an instance of the constant-state protocol $\mathcal{A}$ with identifier $M$ or joined an instance with identifier $M$. Note that by time $T$ all nodes start executing the same instance of the constant-state protocol $\mathcal{A}$, which is guaranteed to stabilize in finite expected time. In particular, after time $T$ if there are more than one leader, then $\mathcal{A}$ will reduce the number of leaders to one. From \lemmaref{lemma:unique-ids} it follows that the probability that there is more than one leader is small.

\begin{lemma}\label{lemma:join-time}
  We have $\E[T] \le kn + 2 \cdot \tbroadcast(G)$.
\end{lemma}
\begin{proof}
Let $v \in V$ be a node with $\deg(v) = \Delta$ and $X$ be the number of steps until node $v$ has been activated $k$ times. Clearly, node $v$ satisfies $\var{id}_X(v) \ge 2^k$ at time step $X$, as node $v$ has either executed Rule (1) of the algorithm $k$ times or it has satisfied the condition in Rule (2) by this time. Let $T(v)$ denote the number of steps until a broadcast initiated at node $v$ at time step $X$ reaches all nodes. Therefore, for all $t \ge X + T(v)$, we have that $\var{id}_t(u) \ge 2^k$ for all $u \in V$ and no node will satisfy the condition in Step (1) at time step $t$.

  Note that $M(t) = M(t')$ for all $t' \ge t \ge X + T(v)$. Now let $u$ that attained the maximum value at time step $X + T(v)$. The broadcast initiated from node $u$ at time step $X + T(v)$ will reach all nodes in the system by some random time $T(u)$. Hence, by time step $X + T(v) + T(u)$ all nodes have the same $\var{id}$ value. Hence, $T \le X + T(v) + T(u)$.
  By monotonicity and linearity of expectation,
  \[
  \E[T] \le \E[X] + \E[T(v)] + \E[T(u)].
  \]
  Since $X = Y_1 + \cdots + Y_k$ is the sum of $k$ geometric random variables with mean $m/\deg(v) = m/\Delta \le n$, we get that $\E[X] \le kn$. As $T(v) + T(v) \le 2 \tbroadcast(G)$ by definition of the broadcast time, we get that $\E[T] \le kn + 2 \tbroadcast(G)$.
\end{proof}

\lebroadcast*
\begin{proof}
Let $F$ be the event that more than one node starts an instance with identifier $M$ and let $v$ be any such node. The probability that some other node $u$ starts an instance with identifier is at most $\Pr[F] \le n/2^k$ by using \lemmaref{lemma:unique-ids} and the union bound over all nodes.
  Let $S$ be the stabilization time of protocol $\mathcal{A}$. Using \lemmaref{lemma:join-time}, the expected stabilization time of the overall protocol is
  \begin{align*}
  \E[T \mid \overline{F}] \cdot \Pr[\overline{F}] + \E[S \mid \overline{F}]\cdot \Pr[F]  &\le \E[T] + \E[S] \cdot \Pr[F] \\
  &\le kn + 2 \cdot \tbroadcastg + \E[S] \cdot n/2^k.
  \end{align*}
  Because the worst-case hitting time of a classic random walk on general connected graphs is $O(n^3)$ and $O(n^2)$ on regular graphs~\cite{levin2017markov}, by \theoremref{thm:6-state-le-cover-time}, we have $\E[S] \le Cn^4 \log n$ for some constant $C$ for any connected graph and $\E[S] \le Cn^3 \log n$ for any connected, regular graph.  
  For general graphs, we can set the parameter controlling the identifier space to $k = \lceil 4 \log n \rceil$ so that number of states is $O(n^4)$ and the expected stabilization time becomes $O(\tbroadcastg + n \log n)$. For regular graphs, it suffices to set $k = \lceil 3 \log n \rceil$ to get the desired bound.
\end{proof}

\section{Space-efficient and fast leader election}\label{sec:impulse}

We now give a leader election protocol whose stabilization time is parameterized by the worst-case expected broadcast time $\tbroadcast(G)$ and whose state complexity depends on the expansion properties of the graph. The approach is inspired by a time-optimal algorithm on the clique due to Sudo et al.~\cite{timeoptimalLE}, but with significant differences: for instance, our algorithm works on any connected graph, and guarantees that a high-degree node is elected as a leader.

\begin{restatable}{theorem}{fastle}\label{thm:fast-le}
  For any graph $G$ with maximum degree $\Delta$, there is a leader election protocol that uses $O(\log n \cdot \stateparam(G))$ states and stabilizes in $O( \tbroadcastg \cdot \log n)$ steps in expectation and with high probability, where $\stateparam(G) \in O( \log( \Delta/\beta \cdot \log n))$.
\end{restatable}

Observe that parameter $\stateparam(G) \in O(\log n)$, so the protocol uses $O(\log^2 n)$ states. Moreover, for graphs where the ratio $\Delta/\beta$ is small, we can obtain $o(\log^2 n)$ space complexity.  For example, in regular graphs we get the following bounds.

\begin{corollary}\label{cor:fast-le-regular}
In any regular graph with conductance $\phi = \beta/\Delta$, there is a leader election protocol that stabilizes in $O( 1/\phi \cdot n \log^2 n)$ steps in expectation and with high probability using $O(\log n \cdot (\log \log n - \log \phi))$ states.
\end{corollary}

The algorithm consists of three parts. First, we describe a space-efficient way for nodes to approximately count the number of \emph{local} interactions. The second part is a two-phase protocol that first removes low-degree nodes from the set of leader candidates and then reduces the number of high-degree leader candidates to one with high probability. Finally, to guarantee finite expected stabilization time, we use the constant-state token-based leader election protocol given in \theoremref{thm:6-state-le-cover-time} as a backup protocol to handle the unlikely cases where the fast part fails.

\subsection{Local approximate clocks on graphs} 

We first describe a subroutine that is used to trigger events every $\Theta(2^{\stateparam})$ expected interactions using exactly $\stateparam + 1$ local states, where $\stateparam \ge 1$ is a given parameter controlling the frequency of the triggered events. Each node $v$ maintains a variable $\var{streak}(v) \in \{0, \ldots, \stateparam\}$, which is initialized to 0.
When node $v$ interacts, it updates its streak counter as follows:
\begin{itemize}
\item If $v$ is the initiator, then set
  \[
  \var{streak}(v) \gets \var{streak}(v) + 1.
  \] Otherwise, set
  \[
  \var{streak}(v) \gets 0.
  \]
\item If $\var{streak}(v) = \stateparam$, then node $v$ is said to \emph{complete a streak} and set
  \[
  \var{streak}(v) \gets 0.
  \]
\end{itemize}
In the above clock protocol, a local event is triggered at a node whenever the node completes a streak.

\paragraph{Analysis of the clock protocol}

 Let $K$ denote the number of times a fixed node needs to interact until it first completes a streak. Here $K$ is the number of fair coin flips needed to observe $\stateparam$ consecutive heads, as the scheduler picks the role of initiator and responder uniformly at random and independently from previous interactions.

 We start with a technical result that approximates the distribution of $K$ using geometric random variables. Recall that
$X \preceq Y$ denotes that $Y$ stochastically dominates $X$.

\begin{restatable}{lemma}{kdominationlemma}
\label{lemma:k-domination}
  The random variable $K$ satisfies $Z_0 \preceq K \preceq Z_1 + \stateparam$, where $Z_0 \sim \Geom(2^{-\stateparam})$ and $Z_1 \sim \Geom(2^{-\stateparam-1})$.
\end{restatable}

The proof of \lemmaref{lemma:k-domination} is given in \appendixref{apx:k-domination-lemma}.
We use  $X(d)$ to denote the number of \emph{steps} (i.e., the number of node pairs sampled by the scheduler) until a fixed node of degree $d$ completes a streak. Note that high degree nodes have a higher probability to complete their streaks, as they interact more often.

The next lemma summarizes some useful properties of  $K$ and $X(d)$. The proof follows by application of concentration bounds on geometric random variables (\lemmaref{lemma:sum-of-geometric}) and Wald's identity (\lemmaref{lemma:wald}).

\begin{restatable}{lemma}{kexpectationandstuff}
\label{lemma:k-expectation-and-stuff}
  Let $1 \le d \le n$. The random variables $K$ and $X(d)$ satisfy the following:
  \begin{enumerate}[label=(\alph*)]
  \item The expected value of $K$ is $\E[K] = 2^{\stateparam+1} - 2$.
  \item The expected value of $X(d)$ is $\E[X(d)] = \E[K] \cdot m/d$.
  \item For any $0 \le \lambda \le 1$  we have $\Pr[X(d) \le \lambda \E[X(d)]] \le 4\lambda + 2^{1 - \lambda \stateparam}$. 
  \end{enumerate}
\end{restatable}

\begin{proof}
  For claim (a), let $s(\stateparam)$ denote the expected number of fair coin flips to obtain $\stateparam$ consecutive heads. Note that $s(1) = 2$ and $s(\stateparam) = s(\stateparam+1)/2 -1$, which implies that 
$s(\stateparam+1) = 2s(\stateparam) + 2$.
Solving this recurrence yields $\E[K] = s(\stateparam) = 2^{\stateparam+1} - 2$.

 For claim (b) and claim (c), note that $X(d)$ is the random sum
 \[
 X(d) = Y_1 + \cdots + Y_{K}
 \]
 of $K$ independent and identically distributed geometric random variables, where $Y_i \sim \Geom(d/m)$ is the number of steps until the fixed node of degree $d$ interacts for the $i$th time after its $(i-1)$th interaction for all $i \ge 1$. Since the random variable $K$ is independent of the sequence $(Y_i)_{i \ge 1}$, Wald's identity (\lemmaref{lemma:wald}) implies
 \[
 \E[X(d)] = \E[K] \cdot \E[Y_1] = \E[K] \cdot m/d,
 \]
 which establishes claim (b).

  Finally, to prove claim (c), set $k = \lceil 2\lambda \cdot \E[K] \rceil$ and define $X = Y_1' + \cdots + Y_k'$, where $Y_i' \sim \Geom(d/m)$ and  all $Y_i'$'s are independent. By linearity of expectation, we have
  \[
  \E[X] = \frac{km}{d} = \frac{k}{\E[K]}\cdot \E[X(d)] \ge 2\lambda \cdot \E[X(d)].
  \]
  Then if $K \ge k$, we know that $X(d)$ stochastically dominates $X$. This means that
  \begin{align*}
\Pr[X(d) \le \lambda\cdot \E[X(d)] \mid K \ge k] &\le \Pr[X \le \lambda \cdot \E[X(d)]] \le \Pr[X \le \E[X]/2] \\&\le \exp\left(-\frac{d}{m} \cdot \E[X] \cdot c(1/2)\right) \le \exp(-k/4)\\&\le \exp\left(- \lambda \cdot 2^{\stateparam-1}\right) \le 2^{-\lambda \stateparam},
\end{align*}
where the third inequality follows from \lemmaref{lemma:sum-of-geometric}b and last inequality follows from $2^{\stateparam-1} \ge \stateparam$. Using the law of total probability, we get 
  \begin{align*}
    \Pr[X(d) \le \lambda\cdot \E[X(d)]] &= \Pr[X(d) \ge \lambda\cdot \E[X(d)] \mid K < k] \cdot \Pr[K < k ] \\
    &+ \Pr[X(d) \ge \lambda\cdot \E[X(d)] \mid K \ge k] \cdot \Pr[K \ge k ]  \\
    &\le \Pr[K < k ] + \Pr[X(d) \le \lambda\cdot \E[X(d)] \mid K \ge k] \\
    &\le \Pr[K < k ] + 2^{-\lambda \stateparam}.
  \end{align*}
  Recall that for a geometric random variable $Z \sim \Geom(1/2^{\stateparam})$, we have
  \[
  \Pr[Z \le k]  = 1 - (1 - 1/2^\stateparam)^k.
  \]
  Note also that $k \le 2\lambda \cdot \E[K] + 1 \le \lambda 2^{\stateparam + 2} + 1$. The claim now follows using \lemmaref{lemma:k-domination} and applying the Bernoulli inequality to the term  $\Pr[K < k]$, as
  \begin{align*}
  \Pr[K < k] &\le \Pr[Z < k]  \le 1 - (1 - 1/2^\stateparam)^{k} \le \frac{k}{2^\stateparam} \le 4\lambda + 2^{-\stateparam}.
  \end{align*}
  Therefore, since $\lambda \le 1$, it follows that
  \[
    \Pr[X(d) \le \lambda\cdot \E[X(d)]] \le 4\lambda + 2^{-\stateparam} + 2^{-\lambda \stateparam} \le 4\lambda + 2^{1-\lambda \stateparam}. \qedhere
  \]
\end{proof}

Together with the above results, we can show that the number $R$ of interactions to complete $\ell \ge \ln n$ streaks is strongly concentrated around the interval $[\E[R]/2, 4\cdot \E[R]]$. Thus, the above process can be essentially used as space-efficient local clock which ticks at a desired (approximate) frequency. In the following, we write $c(\lambda) = \lambda - 1 - \ln \lambda$ for any $\lambda > 0$.

\begin{restatable}{lemma}{rboundslemma}
\label{lemma:r-bounds}
  Let $\ell \ge \ln n$ and $R$ be the number of interactions a node needs to complete $\ell$ streaks. Then   
  \begin{enumerate}[label=(\alph*)]
  \item The expected value of $R$ is $\E[R] = (2^{\stateparam+1} - 2)\ell$.
  \item $\Pr[R \le \lambda \cdot \E[R]/2] \le 1/n^{c(\lambda)}$ for all $0 < \lambda \le 1$
  \item $\Pr[ R \ge 4\lambda \cdot \E[R]] \le 1/n^{c(\lambda)}$ for all $\lambda \ge 1$
   \end{enumerate}
\end{restatable}
\begin{proof}
Note that (a) follows immediately from linearity of expectation and the fact that $R$ is a sum of $\ell$ i.i.d.\ copies of $K$. To show (b) and (c), define the following sums
\[
R_0 = Y_1 + \cdots + Y_\ell \quad \textrm{ and } \quad R_1 = Y'_1 + \cdots + Y'_\ell
\]
of independent geometric random variables, where $Y_i \sim \Geom(1/2^\stateparam)$ and $Y'_i \sim \Geom(1/2^{\stateparam+1})$ for each $1 \le i \le \ell$. By linearity of expectation, we have that
\[
\E[R_0] = \ell 2^\stateparam \quad \textrm{ and } \quad \E[R_1] = \ell 2^{\stateparam+1}.
\]
\lemmaref{lemma:k-domination} implies that $R_0 \preceq R \preceq R_1 + \ell \stateparam$. Note that $R_0 \preceq R$ implies
  \begin{align*}
    \Pr[R \le \lambda \cdot \E[R]/2] &= \Pr[R \le \lambda \cdot \ell(2^\stateparam - 1)] \le  \Pr[R \le \lambda \ell 2^\stateparam]\\
                                     &= \Pr[R \le \lambda \E[R_0]] \le \Pr[R_0 \le \lambda \E[R_0]]\\ &\le e^{-\ell c(\lambda)} \le n^{-c(\lambda)},
  \end{align*}
  where the last two inequalities follow from \lemmaref{lemma:sum-of-geometric} and $\ell \ge \ln n$. Similarly, $R \preceq R_1 + \ell \stateparam$ implies 
    \begin{align*}
    \Pr[R \ge 4\lambda \cdot \E[R]] &= \Pr[R \ge \lambda \ell (2^{\stateparam + 3} - 8)] \le  \Pr[R \ge \lambda \ell 2^{\stateparam+2}]\\
     &= \Pr[R \ge 2 \lambda \cdot \E[R_1]] \le \Pr[R_1 + \ell \stateparam \ge 2 \lambda \ell 2^{\stateparam+1}] \\
    &= \Pr[R_1 \ge 2 \lambda \ell 2^{\stateparam+1} - \ell \stateparam]
    \le \Pr[R_1 \ge \lambda \ell 2^{\stateparam+1}] \\
    &\le e^{-\ell c(\lambda)} \le n^{-c(\lambda)},
    \end{align*}
    where the third inequality follows from $\lambda \ge 1$ and $\ell \stateparam \le \ell 2^{\stateparam+1}$ and the last two from \lemmaref{lemma:sum-of-geometric} and $\ell \ge \ln n$.
\end{proof}

Finally, we examine the concentration of the number of steps until a node completes a certain number of streaks.

\begin{restatable}{lemma}{logstreakslemma}
\label{lemma:log-streaks}
Suppose $\stateparam \in \omega(1)$. Let $\ell \ge \ln n$ and $S = S(d,\ell)$ be the number of steps until a fixed node of degree $d$ completes $\ell$ streaks. Then for all sufficiently large $n$, %
\begin{enumerate}[label=(\alph*)]
\item $\E[S] = \E[K] \cdot \ell m/d = (2^{\stateparam+1}-2) \cdot \ell m/d$,
  \item $\Pr[ S \le \lambda^2 \cdot \E[S]/4 ] \le 2/n^{c(\lambda)}$ for any $\lambda \le 1$, and
  \item $\Pr[ S \ge 8\lambda^2 \cdot \E[S] ] \le 2/n^{c(\lambda)}$ for any $\lambda \ge 1$.
\end{enumerate}
\end{restatable}
\begin{proof}
  Let $(Y_i)_{i\ge 1}$ an infinite sequence of i.i.d.\ geometric random variables, where  $Y_i \sim \Geom(d/m)$. Note that $S$ is the random sum $S = Y_1 + \cdots + Y_R$,
where $R$ is the (random) number of interactions node takes to complete $\ell$ streaks.  Since $R$ is independent of the sequence $(Y_i)_{i \ge 1}$, Wald's inequality (\lemmaref{lemma:wald}) implies that $\E[S] = \E[R] \cdot \E[Y_1]$. Together with \lemmaref{lemma:r-bounds}a this implies claim (a) of the lemma, as 
\[
\E[S] = \E[R] \cdot \E[Y_1] = \left(2^{\stateparam+1}-2\right) \cdot \left(\frac{\ell m}{d}\right).
\]
To show claim (b), let $\lambda \le 1$ and define
\[
r_0 = \lfloor\lambda \cdot \E[R]/2 \rfloor \quad \textrm{ and } \quad S_0 = \sum_{i=1}^{r_0} Y_i.
\]
Note that since $\ell \ge \ln n$, for large enough $n$ we have
\[
\lambda \cdot \E[R]/4 \le r_0 \le \lambda \cdot \E[R]/2.
\]
Observe that $\E[S_0] \ge \lambda \E[S]/4$.
If the event  $R \ge r_0$ happens, then $S_0 \le S$. Hence,
\begin{align*}
  \Pr[S \le \lambda^2 \cdot \E[S]/4 \mid R \ge r_0] &\le \Pr[S_0 \le \lambda^2 \cdot \E[S]/4] \\
  &\le \Pr[S_0 \le \lambda \cdot \E[S_0]] \\
  &\le e^{-r_0 \cdot c(\lambda)} \le 1/n^{c(\lambda)},
\end{align*}
where the last two inequalities follow from \lemmaref{lemma:sum-of-geometric} and the fact that for large enough $n$ we have $r_0 \in \Omega(\ell 2^\stateparam) \subseteq \omega(\log n)$. By law of total probability,
\begin{align*}
  \Pr[ S \le \lambda^2 \cdot \E[S]/4] &= \Pr[S \le \lambda^2 \cdot \E[S]/4 \mid R \ge r_0] \cdot \Pr[ R \ge r_0 ] \\
  &+ \Pr[S \le \lambda^2/4 \cdot \E[S] \mid R < r_0] \cdot \Pr[ R < r_0 ] \\
  &\le \Pr[S \le \lambda^2/4 \cdot \E[S] \mid R \ge r_0]   +\Pr[ R < r_0 ] \\
  &\le 1/n^{c(\lambda)} + 1/n^{c(\lambda)},
\end{align*}
as the term $\Pr[ R < r_0]$ is at most $1/n^{c(\lambda)}$ by \lemmaref{lemma:r-bounds}.

Finally, the proof of the claim (c) follows a similar pattern.
Let $\lambda \ge 1$ and define
\[
r_1 = \lceil \lambda \cdot 4\E[R] \rceil \quad \textrm{ and } \quad S_1 = \sum_{i=1}^{r_1} Y_i.
\]
Note that since $\ell \ge \ln n$, for large enough $n$ we have
\[
\lambda \cdot 8\E[R] \ge r_0 \ge \lambda \cdot 4\E[R].
\]
Observe that $8\lambda \cdot \E[S] \ge \E[S_1]$. If $R < r_1$, then $S < S_1$. Now we get that
\begin{align*}
  \Pr[S \ge 8\lambda^2 \cdot \E[S] \mid R < r_1] &\le \Pr[S_1 \ge 8\lambda^2 \cdot \E[S]]  \\
  &\le \Pr[S_1 \ge \lambda \cdot \E[S_1]] \\
  &\le e^{-r_1 \cdot c(\lambda)} \le 1/n^{c(\lambda)},
\end{align*}
where the last two inequalities follow from \lemmaref{lemma:sum-of-geometric} and the fact that $r_1 \in \Omega(\ell 2^\stateparam) \subseteq \omega(\log n )$. By law of total probability,
\begin{align*}
  \Pr[ S \ge 8\lambda^2 \cdot \E[S]] &= \Pr[S \ge 8\lambda^2 \cdot \E[S] \mid R \ge r_1] \cdot \Pr[ R \ge r_1 ] \\
  &+ \Pr[S \ge 8\lambda^2 \cdot \E[S] \mid R < r_1] \cdot \Pr[ R < r_1 ] \\
  &\le \Pr[ R \ge r_1 ]  + \Pr[S \ge \lambda^2 \cdot \E[S] \mid R < r_1] \\
  &\le 1/n^{c(\lambda)} + 1/n^{c(\lambda)},
\end{align*}
as $\Pr[ R \ge r_1] \ge 1/n^{c(\lambda)}$ by \lemmaref{lemma:r-bounds} and the second term we bounded above. %
\end{proof}

\subsection{The fast leader election protocol}

With the time-keeping mechanism in place, we now describe and analyse the leader election protocol that reduces the number of leader candidates to \emph{one}, with high probability, in $O(\tbroadcastg \cdot \log n)$ steps.

Let $\tau \ge 1$ be an arbitrary fixed constant that controls the probability that the protocol fails (increasing $\tau$ decreases the probability of failure). Fix the parameters 
\[
\stateparam = 8 + \left \lceil \log \left( \frac{\tbroadcastg \cdot \Delta}{m} \right) \right \rceil \quad \textrm{ and } \quad 
 L =  \lceil 2\tau \log n \rceil,
 \]
where $\Delta$ denotes the maximum degree of the graph $G$.

Note that with this choice of parameters $X(d)\in O(\tbroadcast(G))$ in expectation and w.h.p. for any $d \le \Delta$ and also $X(d)\in \Theta(\tbroadcast(G))$ in expectation and w.h.p. for $d \in \Theta(\Delta)$. We also note that $\stateparam \in \Omega(\log \log n)$, as in any graph $\tbroadcastg \ge m \ln( n -1) / \Delta$; see \lemmaref{lemma:nlogn-bc-lb}.

\paragraph{The fast, space-efficient leader election protocol}
As a subroutine, each node runs the streak counter protocol with $\stateparam$ fixed as above. In addition, every node $v$ maintains two state variables $\var{status}(v) \in \{ \mathsf{leader}, \mathsf{follower} \}$ and a counter $\var{level}(v) \in \{0, \ldots, \alpha(\tau) \cdot L \}$, where $\alpha(\tau) > 1$ is a constant we fix later in the analysis.

Initially, each node $v$ initializes the variables to $\var{status}(v) \gets \mathsf{leader}$ and $\var{level}(v) \gets  0$. When $v$ interacts with $u$, node $v$ updates its state using the following rules applied in sequence:
\begin{enumerate}
\item If $v$ completes a streak and $\var{status}(v) = \mathsf{leader}$, then set
  \[ \var{level}(v) \gets \min \{ \var{level}(v) + 1 , \alpha(\tau) \cdot L \}.
  \]
\item If $\var{level}(v) < \var{level}(u)$ and $\var{level}(u) \ge L$, then set
  \[
  \var{status}(v) \gets \mathsf{follower}.
  \]
 \item If $\max \{ \var{level}(u),  \var{level}(v) \} \ge L$, then set
   \[
   \var{level}(v) = \max \{ \var{level}(u),  \var{level}(v) \}.
   \]
\end{enumerate}

\paragraph{Analysis of the fast, space-efficient leader election protocol}
We now analyse the above protocol. We say that a node is at level $\ell$ at time step $t$ if its $\var{level}$ variable is $\ell$ at time step $t$. A node is in the \emph{elimination phase} if it is at least at level $L$. Otherwise, it is in the \emph{waiting phase}. When a node $v$ in the waiting phase interacts with a node in the elimination phase, then $v$ moves to the elimination phase (as a follower).

Note that a node can only remain a leader and increase its level if it completes a streak. When a node increases its level without completing a streak, it must become a follower by Rules (2) and (3). Moreover, by Rule (2), in every time step one of the nodes in the graph with the highest level must be a leader. Thus, the protocol guarantees that there is always at least one leader in every step. Finally, as $\tbroadcastg$ is the worst-case expected broadcast time, Rule~(3) implies that if some node is at level $\ell \ge L$ at step given step, then within $\tbroadcast(G)$ expected steps all nodes are at level at least $\ell$.

The first step in the analysis considers fixed pairs of nodes, characterizing the period of time after which at least one node from a given pair drops out of contention. 

\begin{restatable}{lemma}{pairwisepruninglemma}
\label{lemma:pairwise-pruning}
  Let $u$ and $v$ be nodes with degree at least $d$. If $u$ and $v$ have level at least $L$ and less than $\alpha(\tau) \cdot L$ at step $t$, then at least one of them is a follower at time step $t + 16 \cdot ( \E[X(d)] +  \tbroadcastg)$ with probability at least $1/8$.
\end{restatable}
\begin{proof} 
  Let $L \le \ell < \alpha(\tau) \cdot L$ be the maximum of the levels of $u$ and $v$ at step $t$. Let $X$ be the minimum number of steps until the one of the nodes increments its level to $\ell+1$ after step $t$. Since  both nodes have degree at least $d$, we get that $X \preceq X(d)$, which implies 
  \[
  \Pr[X \ge 16 \cdot \E[X(d)]] \le \Pr[X(d) \ge \E[X(d)]] \le 1/16
  \]
  by Markov's inequality. If either $u$ or $v$ is a not leader at step $t + X$, then the claim follows.
  Suppose that both are leader candidates at step $t + X$. This means that only one of the nodes interacts at step $t + X$: if both interact at step $t + X$, then the initiator will complete its streak, increasing its level to $\ell+1$, and the responder will fail to complete its streak and become a follower due to Rule (2).

Without loss of generality, suppose that node $v$ interacts at step $t + X$. Denote $T$ to be the number of steps it takes for broadcast from $v$ to reach $u$. Then $T \sim T(v,u)$. Note that if by the time $t + X + T$ node $u$ has not completed a streak, then it is at level at most $\ell$ and will be eliminated by the broadcast from $v$ by Rule (2). Let $Y$ be the number of steps after time $t + X$ until $u$ completes a streak. 

Note that the probability that $u$ is a leader at time $t + X + 16 \cdot \tbroadcastg$ is at most $\Pr[T > 16 \cdot \tbroadcastg] + \Pr[Y \le 16 \cdot \tbroadcastg]$.
We show that this is at most $3/4$. The bound on the first term follows from Markov's inequality and the fact that $\tbroadcastg$ is the worst-case expected broadcast time, as
\[
\Pr[ T \ge 16 \cdot \tbroadcast(G) ]] \le \Pr[ T \ge 16 \cdot \E[T]] \le 1/16.
\]
To bound the second term, let $A$ denote the event that during its first interaction after time $t+X$, node $u$ resets its streak counter (i.e., $u$ is a responder). Clearly, $\Pr[A] = 1/2$. Note that $Y \mid A \preceq X(d) \preceq X(\Delta)$, and in particular by \lemmaref{lemma:k-expectation-and-stuff}b we have
\[
\E[Y \mid A] \ge \E[X(\Delta)] \ge 2^8 \cdot \frac{\tbroadcast(G) \Delta}{m} \cdot \frac{m}{\Delta} = 256 \cdot \tbroadcast(G).
\]
Then, by law of total probability and \lemmaref{lemma:k-expectation-and-stuff}c, we have
\begin{align*}
  \Pr[Y \le 16 \cdot \tbroadcast(G)] &= \Pr[ Y \le 16 \cdot \tbroadcast(G) \mid \overline{A}] \cdot \Pr[\overline{A}] + \Pr[ Y \le 16 \cdot \tbroadcast(G) \mid A] \cdot \Pr[A] \\ 
  &\le \frac{1}{2} + \frac{1}{2}\Pr[ X(\Delta) \le 1/16 \cdot \E[X(\Delta)]] \\
  &\le \frac{1}{2} + \frac{1}{2} \cdot \left(\frac{1}{4} + 2^{1 - \stateparam/16}\right) \le \frac{5}{8} + 2^{-\stateparam/16} \le \frac{3}{4},
\end{align*}
as $\stateparam \ge 48$ for large enough $n$. Thus, $u$ is a leader  by time $t + X + 8 \cdot \tbroadcastg$ with probability at most $3/4 + 1/16 = 13/16$. Since $\Pr[X \ge 16 \cdot \E[X(d)]] \le 1/16$, the lemma follows by union bound.
\end{proof}

The second technical lemma leverages this to show a higher concentration result for eliminating all-but-one candidate from contention, assuming all nodes are in the elimination phase. 

\begin{restatable}{lemma}{whppruninglemma}
\label{lemma:whp-pruning}
  Suppose all leader candidates have degree at least $d$ and all nodes are in the elimination phase. Let 
  \[
  t(d) = 16(\tau+2) \cdot \log_{8/7} n \cdot \left( \E[ X(d) ] + \tbroadcast(G) \right).
  \]
  If no node reaches level $\alpha(\tau) L$ by time $t + t(d)$, then exactly one leader candidate remains at time step $t + t(d)$ with probability at least $1 - O(n^{-\tau})$.
\end{restatable}
\begin{proof}
  Let $A(u,v)$ denote the event that at time $t +t(d)$ both $u$ and $v$ are leader candidates. The claim follows by using \lemmaref{lemma:pairwise-pruning} and the union bound over all $u \neq v$, as
  \[
\sum_{u \neq v} \Pr[A(u,v)] \le \sum_{u \neq v} \left(\frac{7}{8}\right)^{(\tau+2) \cdot \log_{8/7} n } \le n^2/n^{\tau+2} \le 1/n^{\tau}. \qedhere
\]
\end{proof}

Next, we provide an upper bound on the time when, with high probability, all nodes are in the elimination phase, and all nodes of small degree have become followers. 

\begin{restatable}{lemma}{pruninglemma}
\label{lemma:pruning}
  There exist constants $\lambda \ge 1$ and $\gamma \ge 1$ such that at time $\lambda L \cdot \tbroadcast(G)$  the following holds with probability $1-O(n^{-\tau})$:
\begin{enumerate}
 \item all nodes are in the elimination phase, and
 \item all nodes of degree at most $\Delta/\gamma$ are followers.
\end{enumerate}
\end{restatable}
\begin{proof}
  Let $\lambda_0 \ge 1$ be a constant such that $c(\lambda_0) \ge \tau$, where $c(\lambda) = \lambda  - 1 - \ln \lambda$. Define $S(d,\ell)$ as in \lemmaref{lemma:log-streaks} and define $t$ as $t = 8\lambda^2_0 \E[S(\Delta, L)]$.
  Note that $t \in O(L \cdot \tbroadcast(G))$ by \lemmaref{lemma:log-streaks} and by choice of $\stateparam$. First we show that by time step $t$ at least one node has reached level $L$ with probability $1 - O(n^{-\tau})$. Consider a node $v$ with degree $\Delta$.  If at step $t$ it is not a leader, then there must exist another node at level at least $L$ by Rule (2). Otherwise, by \lemmaref{lemma:log-streaks}, after $t$ steps this node completed $L$ streaks with probability at least $1 - O(n^{-\tau})$. Then there is at least one node at level at least $L$ by time $t$ with probability $1 - O(n^{-\tau})$. Let $u$ be such node.
  
  Notice that once broadcast from $u$ reaches all other nodes, all nodes will be at level at least $L$. The probability that the broadcast from $u$ does not reach all other nodes after $2\tbroadcast(G)$ steps is at most $1/2$ by Markov's inequality. Then the probability that broadcast from $u$ reaches all other nodes by time $\tbroadcast(G) \cdot L$ is at least $1 - (1/2)^{\tau \log n} \in 1 - O(n^{-\tau})$ since $L \ge 2 \tau \log n$. Thus, all nodes will be at level at least $L$ within $t  +  \tbroadcast(G)\cdot L$ steps with probability $1 - O(n^{-\tau})$. As $t \in O(L \cdot \tbroadcast(G))$, for large enough $n$, we have that $t +   \tbroadcast(G) \cdot L \le \lambda L \cdot \tbroadcast(G)$ for some constant $\lambda \ge 1$. This completes the proof of the claim (1).

  For claim (2), choose a constant $\lambda_1 \le 1$ such that $c(\lambda_1) \ge \tau + 1$ and let $\gamma = 32\lambda_0^2 / \lambda_1^2$. Clearly, $\gamma$ is a constant at least $1$, since $\lambda_0 \ge 1$ and $\lambda_1 \le 1$ are constants. Consider a node $w$ with degree $\deg(w) = d \le \Delta/\gamma$. Let $S \sim S(d, L)$ be the number of steps until $w$ completes $L$ streaks. Then by \lemmaref{lemma:log-streaks} we have that
  \begin{align*}
    \Pr[S \le t] &= \Pr[S \le 8\lambda^2_0 \E[S(d, L)]] \le \Pr[S \le 8\lambda^2_0/\gamma \E[S]] \\
    &= \Pr[S \le \lambda_1^2/4 \E[S]] \le 2/n^{c(\lambda_1)} \le 2/n^{\tau + 1} .
  \end{align*}
  Claim (2) now follows by the union bound over all nodes with degrees less than $\Delta/\gamma$.
\end{proof}

Finally, we put everything together to obtain a w.h.p.\ bound on the time by which there is a single leader candidate left. 

\begin{restatable}{lemma}{faststabilisationlemma}
\label{lemma:fast-stabilisation}
  There exist constants $\alpha(\tau)$ and $C = C(\tau)$ such that there is exactly one leader candidate at time step $C \cdot \tbroadcast(G) \cdot \log n$ with probability at least $1-O(n^{-\tau})$.
\end{restatable}
\begin{proof}
Let $\lambda$ and $\gamma$ by constants given by \lemmaref{lemma:pruning} and $d = \lceil\Delta / \gamma\rceil$. Then with probability $1 - O(n^{-\tau})$ before $\lambda L \cdot \tbroadcast(G)$ steps all nodes are at level at least $L$ and all nodes with degrees less than $d$ are followers by \lemmaref{lemma:pruning}. For the remaining part of the proof, suppose this holds.
We now lower bound the time it takes for some node to reach the maximum level $\alpha(\tau) L$. Note that by \lemmaref{lemma:k-expectation-and-stuff} we get that
\begin{align*}
\E[X(d)] &= (2^{\stateparam+1} - 2)\cdot \frac{m}{d} \\
         &\le 2^{\stateparam+1}\cdot\frac{m}{d} \\
         &\le 512 \cdot \tbroadcast(G) \cdot \frac{\Delta}{m} \cdot \frac{m}{d} 
\le 512\gamma \cdot \tbroadcast(G).
\end{align*}
Let $\tau' \le 1$ be a constant such that $c(\tau') \ge \tau + 1$. Set $t$ as 
\begin{align*}
t &= 16(\tau+2)\left( 512 \gamma + 1 \right) \cdot \log_{8/7} n  \cdot \tbroadcast(G) \\
  &\ge 16(\tau+2) \cdot \log_{8/7} n \cdot \left( \E[ X(d) ] + \tbroadcast(G) \right)
\end{align*}
and set $t'$ as
\[
t' = 4096 (\tau')^2 \cdot \tbroadcast(G) \cdot (\alpha(\tau) - 1) L.
\]
We can pick $\alpha(\tau)$ to be a large enough constant such that $t' > t$ holds for all large enough $n$, since both $t$ and $t'$ are of order $\Theta(\tbroadcast(G) \cdot \log n)$.

Suppose there is a node $v$ that reached level $\alpha(\tau) L$. We aim to show this happened after time $\lambda L \cdot \tbroadcast(G) + t'$ with probability $1 - O(n^{-\tau})$. As we mentioned above, we assume all nodes reached level $L$ before time $\lambda L \cdot \tbroadcast(G)$. After that, to reach level $\alpha(\tau)L$ node $v$ needs to complete a streak at least $(\alpha(\tau) - 1) L$ times. Let $S$ be the time for that to happen. Then $S(\Delta, (\alpha(\tau) - 1)L) \preceq S$. Then by \lemmaref{lemma:log-streaks}
\begin{align*}
\Pr[S \le t'] &\ge \Pr[S(\Delta, (\alpha(\tau) - 1)L) \le t']\\ &\ge \Pr[S(\Delta, (\alpha(\tau) - 1)L) \le (\tau')^2 \cdot 8\E[S(\Delta, (\alpha(\tau) - 1)L)] ] \\
&\ge 1 - n^{-c(\tau')} \in 1 - O(n^{-\tau - 1}).
\end{align*}
Then, by union bound, in $t'$ steps none of the nodes have reached the maximal level $\alpha(\tau) L$ with probability $1 - O(n^{-\tau})$.
Since $t' > t$, \lemmaref{lemma:whp-pruning} implies that a unique leader will be elected before time $t$ with probability $1-O(n^{-\tau})$. It remains to note that there exists a constant $C(\tau)$ such that $t < C(\tau) \tbroadcast(G) \log n$ since $L \in \Theta(\log n)$.
\end{proof}

Finally, to guarantee finite expected stabilization time, the protocol includes a backup phase following the same approach as in~\cite{alistarh2021graphical}. The first node to reach level $\alpha(\tau) L$ must be a leader candidate. When a node $v$ reaches level $\alpha(\tau) L$, it switches to executing the constant-state token-based leader election protocol. When this happens, node initializes the constant-state protocol with the input $\var{status}(v) \in \{ \mathsf{leader}, \mathsf{follower} \}$ and starts running the protocol while simultaneously continues broadcasting its $\var{level}(v)$ value using Rule~(3). Within $\tbroadcast(G)$ expected steps, all nodes are running the constant-state protocol. This protocol guarantees that eventually only one leader remains after polynomially many expected steps.

\fastle*
\begin{proof}
  The protocol uses $O(\stateparam L)$ states. By \theoremref{thm:broadcast-upper-bounds}, $\tbroadcastg \le C m/\beta \log n$ for some constant $C$. Thus,
  \[
\stateparam = 8 + \left \lceil  \log \left( \frac{\tbroadcast(G) \cdot \Delta}{m} \right) \right \rceil \le 8 + \left \lceil \log\left(\frac{Cm \log n \cdot \Delta}{m \beta} \right) \right\rceil \in O\left( \log \left( \frac{\log n \cdot \Delta}{\beta} \right) \right).
\]
Clearly, $L \in \Theta(\log n)$, so the claim on the state complexity follows.
By \lemmaref{lemma:fast-stabilisation}, the fast protocol stabilizes in $O(\tbroadcast(G) \cdot \log n)$ steps with probability at least $1-O(n^{-\tau})$.
By \theoremref{thm:6-state-le-cover-time}, the constant-state backup protocol stabilizes in $O(n^4 \log n)$ expected steps, as the worst-case hitting time of a classic random walk is $O(n^3)$ steps~\cite{levin2017markov}. With probability at most $O(n^{-\tau})$, at least two leader candidates enter the backup phase. Choose $\tau \ge 4$ and let $T(\tau)$ be stabilization time of the protocol. Then
\[
T(\tau) \in O\left(\tbroadcast(G) \cdot \log n + n^{-\tau} \cdot n^4 \log n\right),
\]
which is $O(\tbroadcast(G) \cdot \log n)$.
\end{proof}

\section{Time lower bounds for general protocols}\label{sec:lower}

In this section, we establish time lower bounds for stable leader election for general protocols with unbounded state space. First, we give a fairly general technique for constructing graphs, where leader election has a given time complexity of any order between $\Omega(n \log n)$ and $O(n^3)$. This technique can be also applied to specific graph families to characterize the complexity of leader election in these families. Finally, we also give a result that shows that in any sufficiently dense graph leader election requires $\Omega(n \log n)$ expected steps.

\subsection{The lower bound construction for renitent graphs}

We first introduce the notion of \emph{isolating covers}. The idea is that we can cover the nodes of the graph with at most $K$ subsets of the same size, each of which has isomorphic neighbourhood up to some distance $\ell \ge 0$, and that there are at least two such sets that are sufficiently far apart.

Let $G = (V,E)$ be a graph and $\mathcal{C} = \{ V_0, \ldots, V_{K-1} \} \subseteq 2^{V}$ be a collection of subsets of $V$. We say that $\mathcal{C}$ is a $(K,\ell)$-cover of the graph $G$ if
\begin{enumerate}[noitemsep]
\item for each $0 \le i < j < K$ there exists an isomorphism $\phi$ between 
 $G[B_{\ell}(V_i)]$ and $G[B_{\ell}(V_j)]$ such that $\phi(V_i) = V_j$, 
\item there exists some $V_i$ and $V_j$ such that $B_{\ell}(V_i) \cap B_{\ell}(V_j) = \emptyset$, and
\item  $V_0 \cup \cdots \cup V_{K-1} = V(G)$.
\end{enumerate}
That is, (1) the local neighbourhoods are isomorphic up to distance $\ell$ and this isomorphism maps vertices of $V_i$ to $V_j$, (2) there are two sets whose vertices are all far apart, and (3) the union of the sets covers the entire graph. 

We define 
\[
Y(\mathcal{C}) = \min \{ t : I_t(V_i) \setminus B_\ell(V_i) \neq \emptyset \textrm{ for some } V_i \}
\]
to be the \emph{isolation time} of the cover~$\mathcal{C}$. This is the minimum time until some node in $V_i$ is influenced by some node at distance greater than $\ell$ from all nodes of $V_i$.
We say that $\mathcal{C}$ is \emph{$t$-isolating} if $\Pr[Y(\mathcal{C}) \ge t] \ge 1/2$. This property states that it is unlikely that during the first $t$ steps, nodes in the set $V_i$ can be influenced by nodes that are far away from nodes in $V_i$.

Note that if the distance-$\ell$ propagation time on $G$ satisfies $\Pr[T_{\ell}(G) < t] \le 1/2$, then any $(K,\ell)$-cover of $G$ is $t$-isolating. Thus, we may bound the minimum propagation times to show that a cover is isolating.

Let $\mathcal{G}$ be an infinite family of graphs and $f \colon \mathbb{N} \to \mathbb{N}$ be an increasing function. We say that graphs in $\mathcal{G}$ are \emph{$f$-renitent} if there exists a constant $K \ge 2$ and function  $\ell \colon \mathbb{N} \to \mathbb{N}$ and such that every $n$-node graph $G \in \mathcal{G}$ has an $f(n)$-isolating $(K,\ell(n))$-cover. In this section, we prove the following result.

\begin{restatable}{theorem}{thmlelower}
\label{thm:renitent-lb}
If the graph $G$ is $f$-renitent, then any leader election protocol takes $\Omega(f)$ expected steps to stabilize on $G$.
\end{restatable}

Our approach is inspired by the lower bound construction for randomized leader election in  synchronous message-passing models by Kutten et al.~\cite[Theorem 3.13]{kutten2015complexity}. However, unlike in synchronous message-passing models, in the population model communication is both stochastic and asynchronous with sequential interactions, so we need to further refine the approach to make it work in our setting.

We prove our result in a stronger variant of the population model: we do not restrict the number of states used by the nodes and give each node access to its own (independent and infinite) sequence of random bits. Formally, we assume that each node $v \in V$ is given as input a random value $y(v)$ sampled independently and uniformly at random from the unit interval $[0,1)$. Since we do not restrict the state space of the nodes, the nodes can locally store this value to access an infinite sequence of i.i.d.\ random bits. The random bits assign nodes unique identifiers with probability 1. Any protocol that does not use these random bits can ignore them.

\paragraph{Proof of \theoremref{thm:renitent-lb}}
Suppose $G$ is $f$-renitent and $\mathcal{A}$ is leader election protocol on $G$ that stabilizes in $T$ steps.
Without loss of generality, assume $f(n) \ge 6$, as otherwise the claim of \theoremref{thm:renitent-lb} is trivially true.

Fix any $f(n)$-isolating $(K,\ell)$-cover $\mathcal{C} = \{V_0, \ldots, V_{K-1}\}$ of the graph $G$, where $K$ is a constant independent of $n$, and let $Y = Y(\mathcal{C})$ be the isolation time of the cover. Let $X \sim \operatorname{Poisson}(\lambda)$ be a Poisson random variable with mean $\lambda = f(n)/2$ and $\mathcal{E}$ be the event that $X < Y$. Here, $X$ represents a random time step (independent of $Y$) at which we investigate the state of the system. To this end, we define $L_i$ to be the event that some node $v \in V_i$ outputs that it is a leader at step $X$.
  
\begin{lemma}\label{lemma:leader-events}
The following hold:
  \begin{enumerate}[noitemsep,label=(\alph*)]
  \item $\Pr[L_0 \mid \mathcal{E} ] = \Pr[L_i \mid \mathcal{E}]$ for each $0 \le i < K$, and
\item $\Pr[L_i \cap L_j \mid \mathcal{E}] = \Pr[L_i \mid \mathcal{E}] \cdot \Pr[L_j \mid \mathcal{E}]$ for some $0 \le i < j < K$.
  \end{enumerate}
\end{lemma}
\begin{proof}
  Let $P(G) = \{(v,u), (u,v): \{v,u\} \in E(G)\}$ be the set of ordered pairs of nodes that can interact and let $(e_t)_{t \ge 1}$ be a stochastic schedule on $G$.
  Note that each $e_t$ is sampled from $P$ unfiromly at random independent of other interaction pairs. Let $\sigma = (e_1, \ldots, e_X)$ be the sequence of first $X$ interactions, where $X \sim \operatorname{Poisson}(\lambda)$ is a Poisson random variable with mean $\lambda = f(n)/2$.

First we note that the distribution of $\sigma$ can be equivalently expressed using the following continuous-time process. Suppose each pair $(v,u) \in P(G)$ is activated at unit rate \emph{independently} of all other elements of $P(G)$ and previous activations of $(v,u)$, that is, the number of activations of $(v,u)$ on time interval $[0,t)$ is a Poisson random variable with mean $t$. The sum $Z = Z_1 + Z_2$ of two independent Poisson random variables with mean $\lambda_1$ and $\lambda_2$ is a Poisson random variable with mean $\lambda_1 + \lambda_2$. Therefore, since $|P(G)| = 2m$, 
  the total number $Z$ of activated pairs during the interval $[0, \lambda/(2m))$ is Poisson random variable with mean $\lambda$. That is, $Z \sim X$ and the sequence $\sigma' = (e'_1, \ldots, e'_X)$, where $e'_i$ is the $i$th pair activated in the continuous-time process, has the same distribution as $\sigma$.

    For each $0 \le i < K$, let $U_i = B_\ell(V_i)$ and $\sigma_i$ be the sequence of pairs of nodes from $G[U_i]$ that are activated during the time interval $[0, \lambda/(2m))$ in the continuous time process. If the edge sets of $G[U_i]$ and $G[U_j]$ are disjoint, then $\sigma_i$ and $\sigma_j$ are independent. Moreover, since $G[U_i] \simeq G[U_j]$ for all $0 \le i \le j < K$, the sequences $\sigma_0, \ldots, \sigma_{K-1}$ are identically distributed up to isomorphism, i.e., $\Pr[\sigma_i = a] = \Pr[ \sigma_j = \phi(a)]$, where $\phi$ is the isomorphism between $G[U_i]$ and $G[U_j]$. Note that $\sigma_i$ corresponds to the longest subsequence of $(e_1, \ldots, e_X)$ of pairs of nodes which are both in $U_i$.

    Now suppose the event $\mathcal{E}$ occurs. Then the set of influencers satisfies $I_X(V_i) \subseteq U_i$ for each $0 \le i < K$. Property (1) of $(K,\ell)$-covers implies that $|V_i| = |V_j|$ for each $i\neq j$. Let $N = |V_0|$ and $x_{i,t}: V_i \to \Lambda$  be the configuration of nodes in the set $V_i$ at after $t$ steps. If $\mathcal{E}$ occurs, then $x_{i,t}$ can only depend on the sequence $\sigma_i$ of interactions between nodes in $U_i$ and the initial random values $y_i \in [0,1)^{U_i}$ assigned to nodes in $U_i$. In particular, we get that $x_{i,X}$ is a function of the pair $(y_i, \sigma_i)$.

Since the isomorphism $\phi$ takes $V_i$ to $V_j$, 
       $(y_i, \sigma_i)$ and $(y_j, \sigma_j)$ 
      are identically distributed  (but not necessarily independent). Hence, the configurations $x_{i,X}$ and $x_{j,X}$ are identically distributed given that $\mathcal{E}$ occurs. This implies claim (a) of the lemma. By the second property of $(K,\ell)$-covers, there exist some $i < j$ such that $B_\ell(V_i) \cap B_\ell(V_j) = \emptyset$ and so $U_i \cap U_j = \emptyset$. If the event $\mathcal{E}$ happens, then edges of $G[U_i]$ and $G[U_j]$ are disjoint, and so the random sequences $\sigma_i$ and $\sigma_j$ are conditionally independent given $\mathcal{E}$.
Thus,  $(y_i, \sigma_j)$ and $(y_j, \sigma_j)$ are also conditionally independent given $\mathcal{E}$, which implies that $x_{i,X}$ and $x_{j,X}$ are conditionally independent given $\mathcal{E}$. This implies  claim (b).
\end{proof}

\begin{lemma}\label{lemma:t-x-lb}
  There exists a constant $C(K) > 0$ such that the stabilization time $T$ of protocol $\mathcal{A}$ satisfies
  $\Pr[T > X] \ge C(K)$.
\end{lemma}
\begin{proof}
  Let $x$ be the configuration at the random time step $X$. Note that $T \le X$ if and only if $x$ is a stable configuration with a unique leader. Let $A$ denote the event $T \le X$ and define
  $\Pr[ \overline{A} ] = \alpha$,  $\Pr[\mathcal{E}] = \epsilon$ and $\Pr[L_0 \mid \mathcal{E}] = \gamma$.
Observe that 
\[
\Pr[ A \mid \mathcal{E} ] = \frac{ \Pr[A \cap \mathcal{E}] }{ \Pr[\mathcal{E}] } = \frac{ \Pr[A] - \Pr[A \mid \overline{\mathcal{E}}] \cdot \Pr[\overline{\mathcal{E}}] }{\Pr[\mathcal{E}]} \ge \frac{1-\alpha-(1-\epsilon)}{\epsilon} = 1 - \frac{\alpha}{\epsilon}.
\]
Since the configuration $x$ contains a node in a leader state, union bound yields
\[
\Pr[ A \mid \mathcal{E}] \le \sum_{i=0}^{K-1} \Pr[L_i \mid \mathcal{E}] = K\gamma,
\]
where the equality of probabilities follows from \lemmaref{lemma:leader-events}a.
Let $i < j$ be as given by \lemmaref{lemma:leader-events}b.
Since the configuration $x$ cannot be stable if at least two nodes are in a leader state, we have that
\[
\Pr[ \overline{A} \mid \mathcal{E}] \ge \Pr[L_i \cap L_j \mid \mathcal{E}] = \Pr[L_i \mid \mathcal{E}] \cdot \Pr[L_j \mid \mathcal{E}] = \gamma^2,
\]
because  by \lemmaref{lemma:leader-events}b the events $L_i$ and $L_j$ are conditionally independent given $\mathcal{E}$.  Combining all of the above yields 
\[
1 - \alpha/\epsilon \le \Pr[ A \mid \mathcal{E}] \le K \gamma \le K \sqrt{\alpha/\epsilon}.
\]
Thus, $\alpha$ needs to satisfy the quadratic inequality
$1 - \alpha/\epsilon \le K \sqrt{\alpha/\epsilon}$.  
The solution of this inequality for $\epsilon > 0$ is given by
\[
\sqrt{\alpha} \ge \frac{\sqrt{\epsilon}}{2}\left(\sqrt{K^2+4} - K \right) \ge 0.
\]
Finally, we show that $\epsilon = \Pr[\mathcal{E}] > 1/4$. Since the cover $V_0, \ldots, V_{K-1}$ is $f(n)$-isolating, we have $\Pr[Y \ge f(n)] \ge 1/2$. By %
the law of total probability,
  \begin{align*}
    \Pr[ Y < X ]  
    &\le \Pr[X > 2\lambda] + \Pr[ Y \le 2\lambda] \\
    &\le \exp(-f(n)/4) + \Pr[ Y \le f(n)] \\
    &\le \exp(-3/2) + 1/2 \le 3/4,
  \end{align*}
  where we applied the definition of $\lambda = f(n)/2$, \lemmaref{lemma:poisson-concentration-c}a with $c=2$, and the assumption $f(n) \ge 6$.
  Thus, $\alpha > 0$ is a positive constant denoting the probability of the event $T> X$.
\end{proof}

\thmlelower*
\begin{proof}
  We first show that there exist constants $n_0$ and $C < 1$ such that for any $n \ge n_0$, we have $\Pr[T \le \lambda/2] \le C$. By law of total probability and the definition of conditional probability, we get that 
  \begin{align*}
    \Pr[T \le \lambda/2 ]
    &\le  \Pr[T \le X] + \Pr[T \le \lambda/2 \land T > X] \\
    &\le \Pr[T \le X ] + \Pr[X \le \lambda/2] \\
    &\le 1 - C(K) + \exp(-f(n)/12) \le C.
  \end{align*}
  We can pick constant $C < 1$ satisfying the above since $C(K) > 0$ 
  is a constant  by \lemmaref{lemma:t-x-lb} and $\exp(-f(n)/12)$ is a term that tends to 0 as $n$ increases, as $f$ is an increasing function. 
  Finally, since $1-C > 0$ is a positive constant and $T$ is a non-negative random variable, we have
  \[
  \E[T] \ge \E[T \mid T \ge \lambda/2]\cdot \Pr[T \ge \lambda/2] \ge \frac{(1-C)\lambda}{2} = \frac{1 - C}{4}f(n). \qedhere
  \]
\end{proof}

\subsection{Constructing renitent graphs}

We now give examples of $f$-renitent graphs; by \theoremref{thm:renitent-lb} the expected stabilization time on these graphs will be $\Omega(f)$. For example, as a warmup, it is not hard to see that cycles are $\Omega(n^2)$-renitent: we can split the cycle into four paths $V_0, \ldots, V_3$ of length roughly $n/4$ and information propagation from set $V_1$ to $V_3$ requires $\Omega(n^2)$ steps with constant probability.

\begin{lemma}\label{lemma:cycle-renitent}
  Cycle graphs are $\Omega(n^2)$-renitent.
\end{lemma}
\begin{proof}
  Let $G$ be an $n$-cycle with nodes 
  $V = \{ v_0, \ldots, v_{n-1} \}$. Define $\ell = \lceil n/4 \rceil$ and 
  \[
  V_i = \{ v_{i\ell}, \ldots, v_{(i+1)\ell-1} \}
  \]
  for $i \in \{0,1,2,3\}$. Observe that each $B_\ell(V_i)$ is isomorphic to a path of the same length, $B_{\ell-1}(V_0) \cap B_{\ell-1}(V_2) = \emptyset$ and $V_0 \cup \cdots \cup V_3  =  V$.
\lemmaref{lemma:max-degree-propagation} implies that 
\[
\Pr[T_{\ell}(G) \le c\ell n] \le 1/n \le 1/2
\]
for some constant $c>0$. Thus, $\{V_0, \ldots, V_3\}$ is $\Omega(n^2)$-isolating $(4,\ell)$-cover.
\end{proof}

In fact, for any constant $k >0$, the above idea generalizes to higher dimensions: $k$-dimensional toroidal grids are $\Omega(n^{1+1/k})$-renitent; one can partition such grids into constantly many subcubes of diameter $\Theta(n^{1/k})$ and observe that propagating information to distance $\Theta(n^{1/k})$ in regular graphs requires $\Omega(n^{1+1/k})$ steps with constant probability.

The next lemma  allows us to obtain $\Omega(Dm)$-renitent graphs for essentially any diameter $D$ and number $m$ of edges.

\begin{lemma}\label{lemma:general-construction}
  Let $G$ be a connected graph with $n$ nodes, $m$ edges and diameter $D$. For any integer $\ell$ such that $D \le \ell \le n$, 
  there exists an $\Omega(\ell m)$-renitent graph $G'$ with $\Theta(n)$ nodes, $\Theta(m)$ edges and diameter $\Theta(\ell)$. In addition, $\tbroadcast(G') \in \Omega(\ell m)$.
\end{lemma}
\begin{proof}
  Fix a node $v^*$ of $G$. We construct the graph $G'$ as follows: take four copies $G_0, \ldots, G_3$ of the graph $G$ and connect the $i$th copy $v^*_i$ of $v^*$ to $v_{(i+1) \bmod 4}^*$ by a using a path $P_i$ of length $2\ell$.
  Define $V_i =  V(G_i) \cup V(P_i)$ for $i \in \{0,1,2,3\}$.
  By construction $\mathcal{C} = \{ V_0, \ldots, V_3 \}$ gives a $(4,\ell)$-cover of $G'$, as  $G'[B_k(V_i)] \simeq G'[B_k(V_j)]$ for all $i,j \in \{0,1,2,3\}$ and $k \ge 0$, $B_\ell(V_0) \cap B_\ell(V_2) = \emptyset$ and $V_0 \cup \cdots \cup V_3 = V(G')$.

  We now show that the cover is $t$-isolating for $t = \lambda \ell m$ for some constant $\lambda > 0$. For each $0 \le i \le 3$, let $\rho_i$ and $\rho'_i$ be the sequence of edges in $P_{(i+1) \bmod 4}$ and $P_{(i-1) \bmod 4}$ ordered towards $v^*_i$. Note that for $I_t(V_i) \setminus B_\ell(V_i) \neq \emptyset$ to hold, the scheduler must have sampled either at least half of sequence $\rho_i$ or at least half of $\rho'_i$ (both sequences have length $2\ell$), since $D \le \ell$. Denote by
  $\hat{\rho}_i$ and $\hat{\rho}'_i$
  the halves of $\rho_i$ and $\rho'_i$ respectively that are closest to $V_i$. Now for any $0 \le i \le 3$, we have that
  \begin{align*}
\Pr[ I_t(V_i) \setminus B_\ell(V_i) \neq \emptyset ] &\le \Pr[ X({\hat{\rho}_i}) > t ] + \Pr[ X({\hat{\rho}'_i}) > t] \\
& \le 2 \cdot \exp(-\ell c(\lambda)) \le 1/8
\end{align*}
by using \lemmaref{lemma:rho-concentration} and choosing $\lambda > 0$ to be a sufficiently small constant. By union bound, $\Pr[Y(\mathcal{C}) \ge \lambda \ell m] \ge 1/2$.
The above also implies that the distance $\ell$-propagation time of $v_i^*$ satisfies $\Pr[T_\ell(v_i^*) > t] \ge 7/8$. It now follows that $\tbroadcastg = \max \{ \E[T(u)] : u \in V\} \in \Omega(\ell m)$, as
\[
\E[T(v_i^*)] \ge \E[T_\ell(v_i^*) \mid T_\ell(v_i^*) > t] \cdot \Pr[T_\ell(v_i^*) > t] \ge \frac{7t}{8}. \qedhere
\]
\end{proof}

\begin{restatable}{theorem}{ConstructiveLB}\label{thm:constructive-lb}
  For any increasing function $T \colon \mathbb{N} \to \mathbb{N}$ such that $n \log n \le T(n) \le n^3$, there is an infinite  family of graphs in which stable leader election takes $\Theta(T(n))$ expected steps and the broadcast time satisfies $\tbroadcastg \in \Theta(T)$ .
\end{restatable}

\begin{proof}
  For any $N \ge 1$, we construct a graph $G$ with $n \ge N$ nodes as follows. We distinguish two cases:
  \begin{itemize}

\item First, if $T \in \omega(n^2 \log n)$, then take a clique $H$ of size $N$ and set $\ell = \lceil T(N) / N^2 \rceil$.
  \item 
Otherwise, if $T \in O(n^2 \log n)$, then set $\ell = \lceil \log N + T(N)/(N \log N) \rceil$, take a star graph and add $\Theta(T(N)/\ell)$ edges in an arbitrary fashion to obtain the graph $H$. Adding this many edges is always possible since $T(N)/\ell \in O(N^2)$.
\end{itemize}
In both cases, we apply \lemmaref{lemma:general-construction} with $H$ and $\ell$ to obtain a graph $G$. 
Note that  \lemmaref{lemma:general-construction} implies that the graph $G$ will be $\Omega(T)$-renitent and satisfy $\tbroadcastg \in \Omega(T)$. By \theoremref{thm:renitent-lb} stable leader election will take $\Omega(T)$ expected steps on this graph.
The graph $H$ has constant diameter, and since $\ell \in \Omega(\log n)$, the graph $G$ has diameter $\Omega(\log n)$. By construction, $Dm \in \Theta(T)$, so \theoremref{thm:broadcast-upper-bounds} implies that $\tbroadcastg \in O(T)$. Now
\theoremref{thm:le-in-broadcast-time} implies the upper bound for leader election time.
\end{proof}

\subsection{A lower bound for dense graphs}

The above construction gives graph families in which expected leader election and broadcast time are of the same order. However, this is not generally true. Leader election time can be much lower than broadcast time in graphs, where the local structure helps break symmetry fast. The star graph (i.e., a tree of depth one) is the simplest example: there is a trivial constant-state protocol that elects a leader in one interaction, but broadcast time in a star is $\Theta(n \log n)$ by a simple coupon collector argument.

The above example rules out, for example, the existence of a general $\Omega(n \log n)$ lower bound for leader election in sparse graphs. In this section, we show that in dense graphs with sufficiently high minimum degrees, we cannot easily exploit local graph structure to break symmetry fast, even if we use any number of states per node.

\begin{restatable}{theorem}{denselowerbound}
  \label{thm:dense-lower-bound-general}
  Let $0 < \lambda < 1$ and $0 < \phi < 1$ be constants. If $G$ has minimum degree $\delta \ge \lambda n^\phi$ and $m \ge \lambda n^2$ edges, then any stable leader election protocol requires $\Omega(n\log n)$ expected steps to stabilize on $G$. 
\end{restatable}

At its core, the argument is an extension of lower bound result of Sudo and Masuzawa~\cite{sudo2020leader} from cliques to general high-degree graphs. However, to deal with the general structure of the interaction graph, we introduce the two new concepts: multigraphs of influencers and leader generating interaction patterns.

Our proof strategy is roughly as follows. We assume that there is a fast protocol that stabilizes in $o(n \log n)$ steps in a graph with the properties as in the above theorem. First, we capture the spatial structure of the part of the graph that influences a node to be elected as a leader; we call such structures ``leader generating interaction patterns''. We show that there must be such patterns that are fairly small and almost tree-like. This means they can be unfolded into trees, without growing their size too much. Then we argue that, because the graph has high degrees, such a tree is likely to be found in the set of nodes that have not interacted by time $o(n \log n)$. Since a new leader can be generated in this part of the graph, this implies that any configuration reached in $o(n \log n)$ steps is unlikely to be stable.

Recall that $I_t(v)\subseteq V$ denotes the set of influencers of node $v$ at time $t$. We start with the following lemma showing that the sets of influencers grow slowly on dense graphs.

\begin{lemma}
  \label{lemma:set-of-influencers}
  Let $0 < \varepsilon < 1$ be a constant.
  There exists a constant $0 < c < c(\varepsilon,\lambda)$ such that for any node $v \in V$ and any $0 \le t \le c n \log n$, we have
  \[
  \Pr[ |I_{t}(v)| > n^\varepsilon] \in \exp\left(- \Omega\left(\sqrt{n^\varepsilon}\right)\right).
  \]
\end{lemma}
\begin{proof}
  Define
  \[
  K = \left\lfloor \lambda/2 \cdot \left(n^{\varepsilon/2} - 1\right)\right\rfloor \quad \ell = \left\lfloor 2/\lambda \cdot \left(n^{\varepsilon/2} - 1\right)\right\rfloor \quad r_k = \left \lfloor \frac{\lambda n}{2k} \right\rfloor \quad t_0 = \sum_{k=1}^K r_k.
  \]
  Note that $K \ell \le n^{\varepsilon} - 1$. 
  Assume that $n$ is large enough so that $t_0>0$.
  Let $J_0(v) = \{ v \}$, and for $0 \le t < t_0$
  define
  \[
J_{t+1}(v) = \begin{cases}
  J_t(v) \cup \{ w \} & \textrm{if } e_{t_0-t} = (u,w) \textrm{ or } e_{t_0-t} = (w,u) \textrm{ for some } w \in J_v(t) \\
  J_t(v) & \textrm{otherwise,}
    \end{cases}
\]
where we take $e_0 = \emptyset$. That is, $J_{t+1}(v) = J_t(v) \cup \{w \}$ if some node in $u \in J_v(t)$ interacted with node $w$ at step $t_0 - t$. We can think of $J_t(v)$ as an evolution of $I_t(v)$ played in reverse. Note that $J_{t_0}(v) = I_{t_0}(v)$. Hence, it is enough to show that the event $|J_{t_0}(v)| < n^\varepsilon$ is likely to happen.

We now lower bound the number of steps it takes for $J_t(v)$ to grow from 1 to $K\ell + 1 \le n^{\varepsilon}$.
Note that $J_{t}(v)$ changes by at most one every step. Moreover,
\[
\Pr[|J_{t+1}(v)| = k + 1 \mid |J_t(v)| = k] \le \frac{nk}{m} \le \frac{k}{\lambda n} = p(k),
\]
as $m \ge \lambda n^2$ and there are most $n|J_t(v)|$ edges
incident to nodes in $J_t(v)$. Let $t(a,b)$  be the time it takes for $|J_t(v)|$ to grow from size $a$ to size $b$. Note that when $J_t(v)$ has size at most $b$, then the probability of $J_t(v)$ growing in a single step is upper bounded by $p(b) = b/(\lambda n)$.

We divide  the interval $[1, K\ell + 1]$  into $K$ disjoint intervals of length $\ell$. Let $[a_k, b_k]$ to be $k$th such interval.
For the count to grow from $a_k$ to $b_k$ we need $b_k - a_k = \ell$ steps in which the size of $J_t(v)$ increases. 
The probability that it takes  no more than $r_k$ steps for the count to grow from $a_k$ to $b_k$ can be upper bounded by the probability of the event $Y_k \ge \ell$, where $Y_k \sim \Bin{r_k}{p(b_k)}$ is a binomial random variable. Since 
\[
\E[Y_k] = r_k \cdot p(b_k) \le \frac{\lambda n}{2k} \cdot \frac{b_k}{\lambda n} \le \frac{\ell k}{2k} = \ell/2,
\]
the Chernoff bound from \lemmaref{lemma:bernoulli-chernoff} implies that
\begin{align*}
  \Pr[t(a_k, b_k) \le r_k] \le \Pr\left[Y_k \ge \ell \right] 
                        \le \Pr[Y_k \ge 2 \cdot \E[Y_k]] 
                         \in \exp\left(- \Omega\left(\sqrt{n^\varepsilon}\right)\right),
\end{align*}
as $\E[Y_k] \in \Theta(\sqrt{n^\varepsilon})$.
Now by using the union bound over all $K$ intervals, we have
\begin{align*}
 \Pr[|J_{t_0}(v)| \le n^\varepsilon] &\le \Pr\left[t(1, n^\varepsilon) \le t_0 \right] \\
  &\le \Pr\left[t(1, K\ell + 1) \le \sum_{k=1}^K r_k\right] \\
  &\le \sum_{k=1}^K \Pr[t(a_k, b_k) \le r_k] 
  \in \exp\left(- \Omega\left(\sqrt{n^\varepsilon}\right)\right),
\end{align*}
since $\E[Y_k], K, \ell \in \Theta(\sqrt{n^\varepsilon})$.
 Now we observe that for all sufficiently large $n$, we have
 \begin{align*}
t_0 = \sum_{k = 1}^{K} r_k &\ge \sum_{k = 1}^{K} \left(\frac{\lambda n}{2k} - 1\right) = \frac{\lambda n}{2}\left(\sum_{k = 1}^{K} \frac{1}{k}\right) - K 
= \frac{\lambda n H_k}{2} - K \ge c n \log n,
\end{align*}
where $H_K$ is the $K$th harmonic number and $c>0$ is some constant such that
\[
c \log n \le \frac{\lambda H_K}{2} - \frac{K}{n},
\]
for all sufficiently large $n$. Such a constant exists, since $K \in \Theta(\sqrt{n^{\varepsilon}})$ and $0 < \varepsilon < 1$ is a constant. Thus, the claim of the lemma follows for this $c$.
\end{proof}

Next we show that for any set $U$ of nodes of size $N$ and any constant $\varepsilon > 0$, there likely is a subset of $\Omega(N^{1-\varepsilon})$ nodes in $U$ that have not interacted by time step $o(n \log n)$.

\begin{lemma}\label{lemma:do-it-all}
  Let $0 < \varepsilon \le 1$ be a constant. Fix a set $U \subseteq V$ of size $N$. Let $X(t)$ be the number of nodes in $U$ that have not interacted by step $t \ge 0$. Then for any small enough constant $0 < c < c(\varepsilon,\lambda)$ and for all $0 \le t < c n \log N$, we have 
  \[
  \Pr[X(t) \le N^{1 - \varepsilon}] \in \exp\left(-\Omega\left(N^{1-\varepsilon}\right)\right).
  \]
\end{lemma}
\begin{proof}
  Clearly $X(0) = N$. We estimate the time until $X(t) < N^{1 - \varepsilon}$ holds.
Note that
\[
 \Pr[X(t+1) < X(t) \mid X(t) = k] \le \frac{kn}{m} \le \frac{k}{\lambda n} = p(k),
\]
as $m \ge \lambda n^2$. Moreover, $0 \le X(t)-X(t+1) \le 2$, as in each iteration the number of nodes that have interacted can increase by at most two. Let $t(a,b)$ denote the number of steps until $X(t)$ decreases from $a$ to $b$. Given that $X(t) =a$, the probability that the number of nodes that have not yet interacted decreases is at most $p(a)$. For their count to drop below $b$, we need at least $(a-b)/2$ steps in which the count decreases.

Define 
\[
\ell = 2\left\lceil N^{1-\varepsilon} \right\rceil \quad K = \left\lfloor \frac{N}{\ell} \right\rfloor - 1 \quad t_k = \left\lfloor \frac{\lambda n}{4(k+1)}\right\rfloor \quad t = \sum_{k=1}^K t_k
\]
and partition the interval $[(K+1)\ell, \ell]$ into $K$ disjoint intervals of length $\ell$. Let $[a_k,b_k]$ be the $k$th interval, where $a_k = (k+1)\ell$ and $b_k = k \ell$.

The probability of the event $t(a_k, b_k) \le t_k$ is at most the probability of the event $Y_k \le (a_k - b_k)/2 = \ell/2$, where $Y_k \sim \Bin{t_k}{p(a_k)}$ is a binomial random variable.
From the Chernoff bound of \lemmaref{lemma:bernoulli-chernoff} and the fact that $\E[Y_k] \le \ell/4$, we get that
\begin{align*}
    \Pr[t(a_k, b_k) \le t_k]  &= \Pr\left[Y_k \ge \frac{\ell}{2}\right] \le \Pr[Y_k \ge 2 \cdot \E[Y_k]] \in \exp\left(\Omega\left(N^{1-\varepsilon}\right)\right),
\end{align*}
as $\E[Y_k] \in \Theta(N^{1-\varepsilon})$.
Taking the union bound over all $K \in \Theta(N^\varepsilon)$ intervals yields
\begin{align*}
  \Pr[ X(t) \le N^{1-\varepsilon}] &= \Pr\left[ t(N, N^{1 - \varepsilon}) \le \sum_{k = 1}^K t_k \right] \\
  &\le \sum_{k=1}^K \Pr[t(a_k, b_k) \le t_k] 
    \in \exp\left(-\Omega\left(N^{1-\varepsilon}\right)\right).
\end{align*}
Now observe that for all sufficiently large $N$, we have
\begin{align*}
  t = \sum_{k = 1}^K t_k 
  & \ge \sum_{k = 1}^K \left( \frac{\lambda n}{4(k+1)} - 1\right)
  \ge \frac{\lambda n}{4} \left( \sum_{k = 1}^K \frac{1}{k+1} \right) - K \\
  &\ge \frac{\lambda n ( H_K - 1)}{4} - K 
  \ge cn\log N,
\end{align*}
where $H_K$ is the $K$th harmonic number and $c>0$ is some constant. Such a constant exists as $K \in \Theta(N^\varepsilon)$ and $0 < \varepsilon < 1$ is a constant.
\end{proof}

Now we show that the set of nodes that remain in their initial state after $o(n \log n)$ steps is likely to contain a large induced subgraph that contains a polynomially-sized tree. Let $S(t)$ denote the set of nodes that have not interacted by step $t$.

\begin{lemma}\label{lemma:tree-embedding}
  Suppose $G$ has $m \ge \lambda n^2$ edges and minimum degree $\delta \ge \lambda n^\phi$.  
  There exist constants $\varepsilon > 0$ and $c > 0$ such that for all $t \le c n \log n$, the subgraph induced by nodes in $S(t)$ contains any tree of size $n^{\varepsilon + c}$ with high probability.
\end{lemma}
\begin{proof}
  Fix $\varepsilon < 1/2$ to be a small enough constant such that $\phi(1 - \varepsilon) > 2\varepsilon$. Note that for all sufficiently large $n$, we have the following properties:
\begin{enumerate}[label=(\alph*)]

    \item By applying \lemmaref{lemma:do-it-all} with $U = V$, we get that there exists a positive constant $c_1 = c_1(\varepsilon, \lambda)$ such that the event
\[
      |S(t)| \ge n^{1-\varepsilon}
\]
happens with high probability for any $0 \le t \le c_1 n \log n$.

\item Note that the minimum degree of $G$ is at least $\delta \ge \lambda n^\phi$ in $G$. For any $v \in V$, by applying \lemmaref{lemma:do-it-all} with the set $U = B(v)$, we get that there exists a positive constant $c_2 = c_2(\varepsilon,\lambda)$ such that the event
\[
|B(v) \cap S(t)| \ge |B(v)|^{1 - \varepsilon} \ge \delta^{1 - \varepsilon} \ge n^{\phi(1-\varepsilon)}
\]
happens with high probability for any $0 \le t \le c_2 n \log n$. By union bound, the event happens for all $v \in V$ with high probability.
\end{enumerate}
Let $c = \min\{c_1,c_2, \varepsilon\}$ and $t \le c n \log n$. Note that $c$ is a positive constant. We condition the rest of the proof that the events in (a) and (b) happen for this $t$.

Assume that $n$ is large enough so that $n^{\phi(1-\varepsilon)} > 2n^{2\varepsilon} > 2n^{\varepsilon + c} = 2k$ holds and suppose $T$ is any tree of size $k = n^{c + \varepsilon}$. We show that there is an isomorphic tree in the subgraph of $G$ induced by $S(t)$. Let $U = \{u_1, \ldots, u_k\}$ be the nodes of the tree $T$ ordered by a breadth-first search and $U_i = \{ u_1, \ldots, u_i\}$. We define $T_i$ to be the subgraph of $T$ induced by the nodes $U_i$. We show by induction that we can map each node $u_i \in V(T)$ to a node $v_i \in V(G)$ such that the set $\{v_1, \ldots, v_i\}$ induces a subgraph in $T$ that contains the edges of $T_i$.
  
  For the base case, we map $u_1$ to an arbitrary node $v_1 \in S(t)$. For the inductive step, suppose that $\{v_1, \ldots, v_i\}$ induces a subgraph that contains the edges of $T_i$. Clearly $u_{i+1}$ has at most $k$ neighbours in $T$, as $T$ has at most $k$ nodes. Let $u_j \in U_i$ be the parent of $u_{i+1}$ in $T$. Note that $j \le i$. From (b), we have that
  \[
  |S(t) \cap B(v_j)| - |U_i| \ge n^{\phi(1-\varepsilon)} - k > k,
  \]
  so we can find a node $v_{i+1} \in S(t) \cap B(v_j)$ that is not in $\{v_1, \ldots, v_i\}$. Map $u_{i+1}$ to this node. Now the subgraph induced by $\{v_1, \ldots, v_{i+1}\}$ contains the edges of $T_{i+1}$.
\end{proof}

\paragraph{Multigraphs of influencers}

To finish the proof, we need more knowledge about how the sets of influencers look like. For each $v \in V$ and $t$, we construct a directed \emph{multigraph of influencers} $\mathcal I_t(v)$ that keeps tracks of the interactions that have influenced $v$ in a more fine-grained manner than the set $I_t(v)$ of influencers alone would. The latter only keeps track of the nodes that can influence the state of $v$ at step $t$, not the sequence of interactions that have influenced the node.

We define this multigraph as follows. Let $(e_t)_{t\ge 1}$ be the stochastic schedule on the interaction graph $G$.  The set of nodes of $\mathcal I_t(v)$ will be the set of influencers $I_t(v)$ and the \emph{multiset} of (directed) edges is given by the following recurrence. Define $\mathcal I_t(v) = \emptyset$ and for all $t \ge 0$, define
\[
  \mathcal I_{t+1}(v) = \begin{cases}
    \mathcal I_{t}(v) \cup \mathcal I_{t}(u) \cup \{(v,u)_{t+1}\} & \textrm{if } e_{t+1} = (v,u) \\
    \mathcal I_{t}(v) & \textrm{otherwise.} \\
  \end{cases}
\]
In the first line, $(v,u)_{t+1}$ is understood to be a directed edge timestamped with the time step $t+1$ of its activation. That is, $\mathcal I_t(v)$ is a \textit{directed multigraph} where edges of multiplicity $> 1$ are allowed as long as they have different timestamps. However, note that there can be no self-loops.

The multigraph $\mathcal I_t(v)$ \emph{fully} captures interactions that happened before time step $t$ and which can affect the state of $v$ at step $t$. Hence, given $\mathcal I_t(v)$, we can determine the state of node $v$ at time $t$. This gives us more information than sets of influencers alone.

It will be easier to analyze $\mathcal I_t(v)$, when its elovution is played in reverse. Fix $t_0 \ge 0$ and define $\mathcal J_t(v)$ to be the induced subgraph constructed on edges of $\mathcal I_t(v)$ that were added after step $t_0 - t$ for $t \le t_0$. Observe that $\mathcal J_{t_0}(v) = \mathcal I_{t_0}(v)$ and that $\mathcal J_t(v)$ follows the recurrence given by  $\mathcal J_0(v) = (\{v\}, \emptyset)$ and
\[
\mathcal J_{t+1}(v) = \begin{cases}
  \mathcal J_t(v) \cup \{(u,w)_{t_0 - t}\}  & \textrm{if } e_{t_0-t} = (u,w) \textrm{ and } \{u,w\} \cap V(\mathcal J_t(v)) \neq \emptyset \\
  \mathcal J_t(v) & \textrm{otherwise}.
\end{cases}
\]
In the first line, $V(\mathcal{J}_t(v))$ denotes the nodes of $\mathcal{J}_t(v)$, and $\mathcal J_t(v) \cup \{(u,w)_{t_0 - t}\}$ is interpreted as adding a directed edge to the multiset as edges of $\mathcal J_t(v)$ along with adding $u$ and $w$ to the set of nodes of $\mathcal J_t(v)$, if one of them does not belong to $V(\mathcal J_t(v))$ yet. We also keep original timestamps on edges of the multigraph.

Call an interaction of two nodes $u$ and $w$ at time step $t - t_0$ \textit{internal} if both $u$ and $w$ are contained in $\mathcal J_v(t)$. Internal interactions create cycles in $\mathcal J_v(t)$. We show that internal interactions almost never happen before $\Omega(n \log n)$ steps, which means that $\mathcal J_v(t_0)$ is tree-like for small enough $t_0$.

\begin{lemma}
  \label{lemma:internal-interactions}
  For any $\varepsilon>0$, any small enough $0 < c < c(\varepsilon, \lambda)$ and any $t_0 \le cn \log n$, with high probability the following events happen for all $v \in V$ simultaneously %
  \begin{itemize}
    \item there are at most $c \log n$ internal interactions in $\mathcal J_{t_0}(v)$, and
      \item $\mathcal J_{t_0}(v)$ has size at most $n^\varepsilon$.
  \end{itemize}
\end{lemma}
\begin{proof}
  Consider some $v \in V$. By \lemmaref{lemma:set-of-influencers}, for any small enough $c$, we have that 
  $I_{t_0}(v)$ contains at most $n^\varepsilon$ nodes with high probability for any $t_0 \le c n \log n$. Hence, $\mathcal J_{t_0}(v)$ also contains at most $n^\varepsilon$ nodes. Thus, for all $0 \le t \le t_0$ we know that $\mathcal J_{t}(v)$ contains at most $n^\varepsilon$ nodes, since the size of $\mathcal J_t(v)$ does not exceed the size of $\mathcal J_{t_0}(v)$. Therefore, the probability that any two nodes in $\mathcal J_t(v)$ interact during any time step $1\le t \le t_0$ is at most
  \[
p = \frac{n^{2\varepsilon}}{2m} \le \frac{1}{\lambda n^{2 - 2\varepsilon}}.
  \]
  The number of internal interactions during the interval $\{1, \ldots, t_0 \}$ is stochastically dominated by a binomial random variable $X \sim \Bin{t_0}{p}$ with mean $\E[X] = t_0p$. By \lemmaref{lemma:bernoulli-chernoff} we have
  \begin{align*}
  \Pr[ X \ge c \log n] &\le \Pr[X \ge \lambda n^{1-2\varepsilon} \cdot \E[X]]
                       \le \exp\left(- \E[X] \cdot (\lambda n^{1-2\varepsilon})^2/3 \right) \\
                       &\le \exp\left(- \frac{\lambda c n^{1-2\varepsilon}\log n }{3} \right) 
                       \in \exp(- \omega(\log n) ).
  \end{align*}
  Now taking the union bound over all $v$ implies the claim.
\end{proof}

Similarly to $\mathcal I_{t_0}(v)$, given $\mathcal J_{t_0}(v)$, we can determine the state of $v$ at time step $t_0$. We say that a multigraph $\mathcal J(v)$ is a \emph{leader generating interaction pattern for node $v$} if $\mathcal{J}_{t_0}(v) = \mathcal{J}(v)$ implies that node $v$ is a leader at step $t_0$. 
We now show that we can unfold such multigraphs into trees, without increasing their size too much. 

\begin{lemma}\label{lemma:unfold}
Suppose $\mathcal J(v)$ is a leader generating interaction pattern for node $v$ with $N$ distinct nodes and $k$ internal interactions. Then there exists a leader generating interaction pattern $\mathcal K(v)$ for node $v$ that has $k-1$ internal interactions and has size at most $2N$.
\end{lemma}
\begin{proof}
  Consider an internal interaction between nodes $u$ and $w$ in $\mathcal J(v)$ at time $r$, such that $r$ has the smallest value.
  Consider all interactions that influenced nodes $u$ and $w$ before step $r$. Because the interaction between $u$ and $w$ is an internal interaction with the smallest timestamp, the interactions that influenced nodes $u$ and $w$ before step $r$ are given by some multigraphs of influencers $\mathcal I(u)$ and $\mathcal I(w)$, which are trees disjoint in edges and vertices. We now construct a new multigraph $\mathcal{K}$ as follows:
\begin{itemize}

   \item First, remove the interaction between $u$ and $w$ at step $r$ from $\mathcal{J}(v)$. Additionally, for every interaction with timestamp larger than $r$ (i.e. interaction that happened after $u$ and $w$ interacted), shift the value of its timestamp ahead by $2r+1$. This way, every interaction in the current multigraph has timestamp either less than $r$ or greater than $3r+1$. 

   \item Second, create isomorphic and disjoint copies $\mathcal I(u')$ and $\mathcal I(w')$ of the multigraphs of influencers $\mathcal I(u)$ and $\mathcal I(w)$, respectively. Add these interactions (and the new nodes that take part in these interactions) to the multigraph $\mathcal{K}$. Shift all timestamps of $\mathcal I(u')$ ahead by $r$ and shift all timestamps of $\mathcal I(w')$ ahead by $2r$. That way, timestamps of all interactions given by $\mathcal I(u')$ is between $r$ and $2r -1$, timestamps of all interactions given by $\mathcal I(w')$ are between $2r$ and $3r -1$.

   \item Finally, add the interactions $(u,w')$ and $(u',w)$ at time step $3r$ and $3r+1$.
\end{itemize}

\begin{figure}[t]
    \centering
    \begin{subfigure}[b]{0.4\textwidth}
        \centering
        \includegraphics[width=\textwidth]{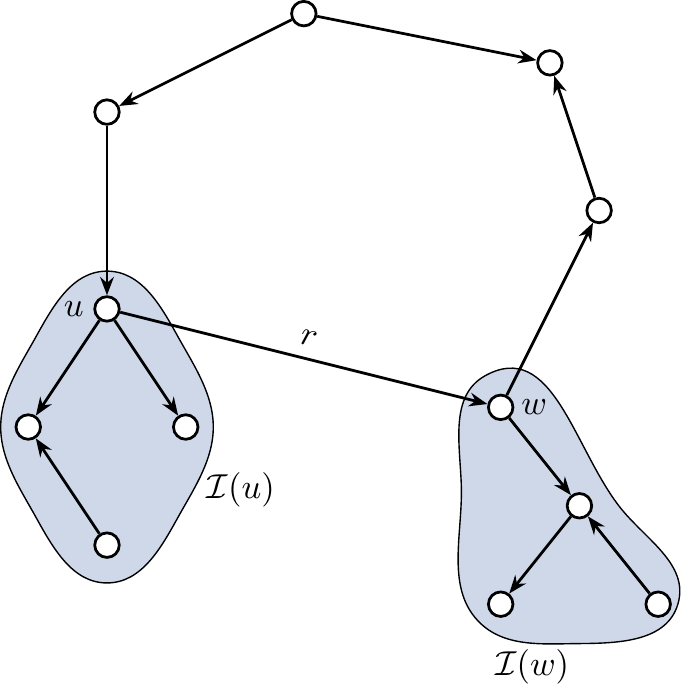}
        \vspace{0.5cm}
        \caption{}
        \label{fig:unfold-before}
    \end{subfigure}
    \hfill
    \begin{subfigure}[b]{0.5\textwidth}
        \centering
        \includegraphics[width=\textwidth]{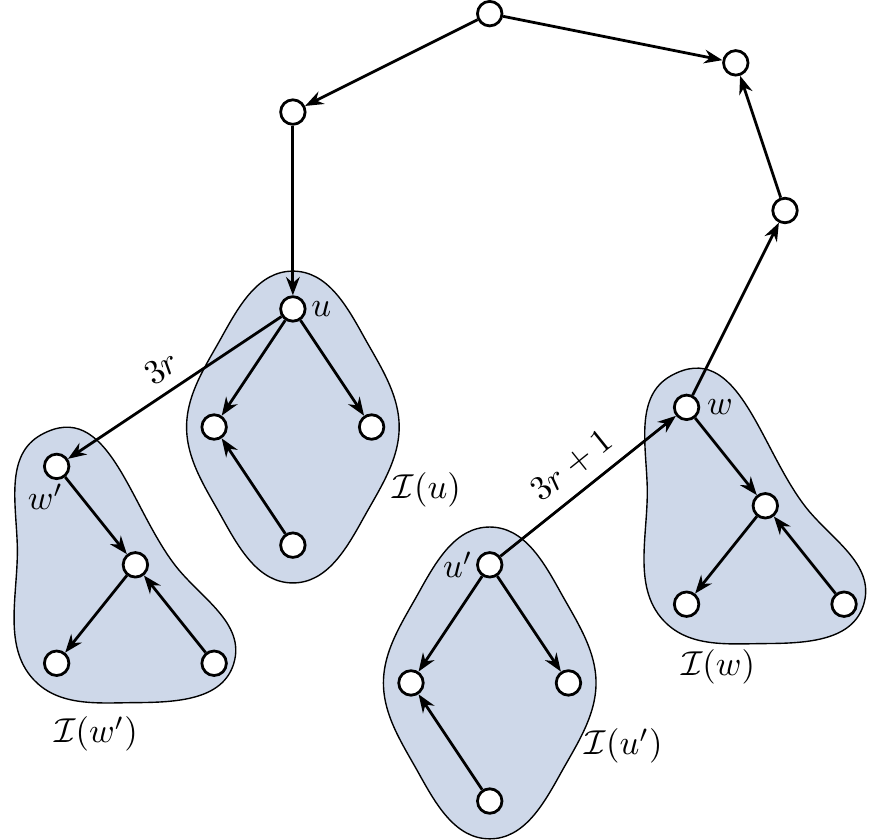}
        \caption{}
        \label{fig:unfold-after}
    \end{subfigure}
    
    \caption{Illustration of the proof strategy of \lemmaref{lemma:unfold}. (a) Multigraph $\mathcal J(v)$ with an internal interaction between $u$ and $w$. (b) Multigraph $\mathcal K$ with one internal interaction less.}
    \label{fig:unfold}
\end{figure}

See \figureref{fig:unfold} for an illustration of the above construction. First, note that interactions in the graph $\mathcal K$ have distinct timestamps. Additionally, note that the new interactions at steps $3r$ and $3r+1$ will update the states of nodes $u$ and $w$ the same way in $\mathcal{K}$, as the interaction $(u,w)$ did previously, because $u'$ and $w'$ are influenced by isomorphic interaction sequences as $u$ and $w$ were in $\mathcal{J}(v)$, and since the nodes are anonymous in the population protocol model and have no inputs in the leader election problem.
In particular, this means that node node $v$ will have the same state after the system has executed interactions in $\mathcal{K}$ as it would have when executing interactions from $\mathcal{J}(v)$. Thus, $\mathcal{K}$ is also a leader generating pattern for $v$.

Finally, observe that $\mathcal{K}$ has at most $2N$ vertices: by adding $\mathcal I(u')$ and $\mathcal I(w')$ to $\mathcal{K}$ we at most doubled the number of vertices in $\mathcal J(v)$, as $\mathcal J(v)$ contained its copies of $\mathcal I(u)$ and $\mathcal I(w)$ initially. Moreover, $\mathcal{K}$ has at most $k-1$ internal interactions, as we removed at least one internal interaction of $\mathcal{J}(v)$ without adding any new ones. 

\end{proof}

We are now ready to prove the lower bound.

\denselowerbound*
\begin{proof}
  By \lemmaref{lemma:tree-embedding}, there exist constants $\varepsilon > 0$ and $c_0>0$ such that for any $t \le c_0 n \log n$, the set $S(t)$ contains any tree of size $n^{\varepsilon+c_0}$ with high probability. Moreover, let $c_1>0$ be the constant obtained from \lemmaref{lemma:internal-interactions} with this $\varepsilon$. Let $c = \min \{c_0, c_1\}$.
  Let $\mathcal{P}$ be any a leader election protocol on $G$. Consider the configuration $x$ at any fixed time step $t \le c n \log n$.  Note that for such $t$, the statements of \lemmaref{lemma:internal-interactions} and \lemmaref{lemma:tree-embedding} hold with high probability; we condition the rest of the proof on the event that both statements hold. We show that in such case $x$ is not a stable configuration. 
  
  Suppose $x$ is stable and $v$ is the elected leader. Now $\mathcal J(v) = \mathcal J_{t}(v)$ is a leader generating pattern. By \lemmaref{lemma:internal-interactions}, $\mathcal J(v)$ has at most $c \log n$ internal interactions and size at most $n^\varepsilon$ with high probability.
  By applying \lemmaref{lemma:unfold} repeatedly for $c \log n$ times, we can obtain a new pattern $\mathcal J'(v)$ with size $n^{\varepsilon + c}$ without any internal interactions. That is, the edges in $\mathcal J'(v)$ (when the orientation of the edges is ignored) form a tree $T$ of size at most $n^{\varepsilon + c}$; note that absence of internal interactions also means absence of edges of multiplicity higher than 1.

  By \lemmaref{lemma:tree-embedding}, the subgraph induced by $S(t)$ contains a subgraph isomorphic to $T$. In particular, this means that we can construct a leader generating interaction pattern $\mathcal J(w)$ isomorphic to $\mathcal J(v)$ for some node $w \in S(t)$ using only nodes of $S(t)$.
  This means that there is an execution, where node $w \in S(t)$ will be elected as a leader. 
  Hence, $x$ is not a stable configuration with high probability.
Thus, the stabilization time $T_\mathcal{P}$ of the protocol $\mathcal{P}$ satisfies
\[
\E[T_\mathcal{P}] = \sum_{i=0}^\infty \Pr[T_\mathcal{P} > i] \ge \sum_{i=0}^t \Pr[ T_\mathcal{P} > i] \ge Ct 
\]
for some constant $C>0$. Thus, $\E[T_\mathcal{P}] \in \Omega(n \log n)$.
\end{proof}

\section{Lower bound for constant-state protocols}\label{sec:constant-lb}

In this section, we give lower bounds for constant-state protocols that stabilize in finite expected time on any connected graph. In the clique, the classic approach has been to utilize the so-called \emph{surgery technique} of Doty and Soloveichik~\cite{doty2018stable}, later extended by Alistarh et al.~\cite{alistarh2017time} to show lower bounds for super-constant state protocols. 
Roughly, surgeries consists of carefully ``stitching together'' transition sequences, in order to completely eliminate states whose count decreases too fast (e.g., the leader state), thus resulting in incorrect executions (e.g., executions without a leader).

When moving beyond cliques, applying surgeries is difficult. The key challenge is that in addition to keeping track of the counts of states, we also need to control for the \emph{spatial distribution} of generated states in order to determine if a given configuration is stable. For example, if we know that an interaction $(a,b) \to (c,d)$ produces a leader, then in the case of a clique it suffices to check if states $a$ and $b$ are present in the overall population to determine if this interaction can produce a new leader. However, in the case of general interaction graphs, the rule $(a,b) \to (c,d)$ can only produce a leader if some nodes with states $a$ and $b$ in the configuration are \emph{adjacent}.

We circumvent this obstacle by considering a \emph{random graph setting}, where the interaction graph itself is \emph{probabilistic}.  Instead of showing a lower bound for a given graph, we give a lower bound that holds in \emph{most} graphs, where ``most'' is interpreted as having graphs coming from a certain probability distribution. We will  focus on the  Erd\H{o}s--R\'enyi random graph model $G_{n,p}$. {As we are only interested in connected graphs $G$, we adopt the convention that stabilization time $T_{\mathcal{P}}(G) = \infty$ if $G$ is disconnected.} 
Our main result is the following.

\begin{restatable}{theorem}{nsquarelowerbound}\label{thm:nsquared-lb}
  Fix a constant $p > 0$ and let $G \sim G_{n,p}$ conditioned on $G$ being connected. Then there is no protocol $\mathcal P$ such that all of the following hold:
  \begin{enumerate}
      \item the protocol $\mathcal P$ stabilizes on any connected graph;
      \item transition function of $\mathcal P$ is independent of underlying communication graph;
      \item the expected stabilization time of $\mathcal P$ is $o(n^2)$.
  \end{enumerate}
\end{restatable}

 This result generalizes the lower bound of Doty and Soloveichik~\cite{doty2018stable} from cliques to {dense} random graphs. As such, we follow a similar approach, but provide new ideas to deal with the structure of the interaction graph. First, we show that any protocol starting from a uniform initial configuration, passes through a ``fully dense'' configuration with very high probability on a sufficiently-dense Erd\H{o}s--R\'enyi random graph. Second, we show that if the protocol stabilizes too fast, then there exist reachable configurations with many states in low count. Finally, we use  surgeries to show that such a  protocol must fail even on the clique.

\subsection{First step: Reaching fully dense configurations}

Fix the protocol $\mathcal{P}$ as in \theoremref{thm:nsquared-lb}. Let $\Lambda$ be the set of states of the protocol.
Without loss of generality, we assume that all nodes start in the same initial state $s_\text{init}$, and that every state in $\Lambda$ is producible by some schedule.

For $\alpha > 0$, a configuration $x$ is said to be $\alpha$-dense with respect to a set $\Lambda' \subseteq \Lambda$ if every state in $\Lambda'$ is present in count at least $\alpha n$. A configuration is \emph{fully $\alpha$-dense} with respect to  $\Lambda'$ if additionally all other states have count zero. We now show that the protocol reaches a fully $\alpha$-dense configuration with respect to $\Lambda$ with probability $1 - \exp(-\Omega(n))$ on a random Erd\H{o}s--R\'enyi graph. For any two sets $S,T \subseteq V(G)$, we write $e(S,T)$ for the number edges between nodes in $S$ and $T$.

\begin{restatable}{lemma}{DenseErdosGraphs}
  \label{lemma:dense-erdos-graphs}
  Let $p,c > 0$ be constants and $G \sim G_{n,p}$. Then for any disjoint $S, T \subseteq V(G)$ such that $|S|, |T| > cn$ the events 
  \begin{itemize}
  \item $p|S||T|/2 \le e(S,T) \le 3p|S||T|/2$ and  \item $p|S|^2/4 \le  e(S,S) \le 3p|S|^2/4$ 
  \end{itemize}
  hold with probability $1 - \exp(-\Omega(n))$.
\end{restatable}
\begin{proof}
  Note that in an  Erd\H{o}s--R\'enyi random graph, the number $e(S,T)$ of edges between $S$ and $T$ follows the binomial distribution $\Bin{|S||T|}{p}$ with mean $\mu = p |S||T| \in \Theta(n^2)$. Using standard Chernoff bounds, we get that
\[
\Pr\left[\left|e(S,T) - \mu \right| < \frac{\mu}{2}\right] > 1 - 2 \cdot \exp\left(-\frac{\mu}{12}\right) \in 1 - \exp\left(-\Omega\left(n^2\right)\right).
\]
The first claim follows by applying the union bound over at most $4^n$ pairs of sets. 
Observe that $e(S,S)$ follows a binomial distribution of $\Bin{|S|\left(|S| - 1\right)/2}{p}$ with mean $\mu' = p |S|\left(|S| - 1\right)/2 \in \Theta(n^2)$.
Therefore, the second claim also follows by an analogous argument.
\end{proof}

\begin{lemma}\label{lemma:fully-dense-config}
  Let $p>0$ be a constant and $G \sim G_{n,p}$. Then there exists a constant $\alpha = \alpha(p, |\Lambda|)$ such that an execution of $\mathcal{P}$ on $G$ reaches a fully $\alpha$-dense configuration with respect to $\Lambda$ in $O(n)$ steps  with probability $1-\exp(-\Omega(n))$. 
\end{lemma}
\begin{proof}
  Suppose all nodes start in the initial state $s_\text{init}$. Without loss of generality, assume that every state in $\Lambda$ is producible by some transition sequence from the initial configuration. Then there exists a sequence of sets $\{s_\text{init} \} = \Lambda_1 \subsetneq \cdots \subsetneq \Lambda_{|\Lambda|} = \Lambda$ such that the state $b_{i+1} \in \Lambda_{i+1} \setminus \Lambda_{i}$ is producible using some interaction between nodes with states from $\Lambda_{i}$.
  Let $f(1)=p/4$ and for $i \ge 1$ define
  \[
f(i+1) = f^2(i) \cdot \frac{p}{960}.
  \]
Note that $f$ is decreasing and $f(i+1) \le f(i)/2$ for every $i \ge 1$.
We assume that for the constant $c=f(|\Lambda|+1)$ the events given in \lemmaref{lemma:dense-erdos-graphs} happen. This implies that all of the following hold with probability $1-\exp(-\Omega(n))$:
\begin{enumerate}
\item The number of edges in $G$ satisfies $pn^2/4 \le m \le 3pn^2/4$.
\item For any disjoint sets $U, U' \subseteq V$ of size at least $nf(i)$, we have that $e(U,U') \ge |U| |U'|p/2$ for $1 \le i \le |\Lambda|+1$.
\item For any set $U \subseteq V$ of size at least $nf(i)$, we have $e(U,U) \ge |U|^2p/4$ for $1 \le i \le |\Lambda|+1$.
\end{enumerate}
We say that an event holds with very high probability (w.v.h.p.) if it holds with probability $1-\exp(-\Omega(n))$. We now prove by induction on $i$ that the protocol reaches an $f(i)$-dense configuration with respect to $\Lambda_i$ with very high probability. This  implies that we reach a $f(|\Lambda|)$-dense configuration with very high probability. The base case $i=1$ of the induction is vacuous, as the initial configuration is $1$-dense with respect to $\Lambda_1 = \{ s_\text{init} \}$.
, and hence, also $f(1)$-dense with respect to $\Lambda_1$. We say that a state $s$ has density $\rho$ in a configuration if there are $\rho n$ nodes in state $s$.

For the inductive step, suppose we have reached a configuration that is $f(i)$-dense with respect to $\Lambda_{i}$ for some $1 \le i < |\Lambda|$. Suppose this happens at step $t$. Let $s \in \Lambda_i$ and define $N = np/32$. By the induction assumption, the density of $s$ is at least $f(i)$ at step~$t$.

First, we show that the density of $s$ remains at least $f(i)/2 \ge f(i+1)$ for the next $N=np$ steps with very high probability. Suppose that the density of $s$ goes below $f(i)/2$ before step $t+N$ and let $t' < t + N$ be the last time when the density is at least $f(i)$. Note that if the count of $s$ decreases, it can decrease by at most two. The claim follows if we can show that w.v.h.p.\ the count of $s$ decreases at most $nf(i)/4$ times during the interval $I = \{ t', t'+1, \ldots, t+N\}$, which would be a contradiction. During this interval, the density of $s$ is at most $f(i)$ and the probability of the density decreasing in a single step is at most the probability that a stochastic scheduler samples an edge incident with a node in state $s$. Since every node in state $s$ has degree at most $n$, the probability is at most 
\[
r(i) = \frac{f(i)n^2}{m} \le \frac{4f(i)}{p} \le 1,
\]
where the bound on $m$ follows from (1) and the fact that $f(i) \le p/4$. The number $X(s)$ of times the density of $s$ decreases during the interval $I$ is stochastically dominated by the binomial random variable $Y \sim \Bin{N}{r(i)}$ with mean
\[
\E[Y] = Nr(i) = \frac{npr(i)}{32} \le \frac{np}{32} \cdot \frac{4f(i)}{p} = \frac{nf(i)}{8}.
\]
By standard Chernoff bounds (\lemmaref{lemma:bernoulli-chernoff}), we get that
\begin{align*}
\Pr\left[X(s) \ge \frac{nf(i)}{4}\right] &\le \Pr\left[Y \ge \frac{nf(i)}{4}\right] \\
&\le \Pr[Y \ge 2 \cdot \E[Y]] \\
&\le \exp(-\E[Y] \cdot 2/3) \in \exp(-\Omega(n)).
\end{align*}
By union bound, the density of every $s \in \Lambda_i$ remains at least $f(i)/2 \ge f(i+1)$  until step $t+N+1$ with very high probability.

To complete the induction, we show that the density of $b_{i+1} \in \Lambda_{i+1} \setminus \Lambda_i$ reaches $f(i+1)$ before step $t + N$ with very high probability. Suppose that the density of $b_{i+1}$ remains less than $f(i+1)$ until step $t+N$. Similarly, as in the argument above, the probability that the count of $b_{i+1}$ decreases given that its density is at most $f(i+1)$ is at most
\[
r(i+1) = \frac{f(i+1)n^2}{m} \le \frac{4f(i+1)}{p} \le 1.
\]
Let $s,s' \in \Lambda_{i}$ be a pair of states whose interaction generates the state $b_{i+1}$ (such states must exist by definition of $b_{i+1}$). The probability that the density of $b_{i+1}$ increases is the probability that some pair of nodes $(v,u)$ is sampled by the stochastic scheduler, where $v$ is in state $s$ and $u$ is in state $s'$. Above, we showed that the density of $s$ and $s'$ is at least $f(i)/2$, so the probability that the density of $b_{i+1}$ increases is at least
\[
q(i+1) = \frac{(nf(i)/2)^2p}{2 \cdot 2m} \ge \frac{(nf(i))^2 p}{16} \cdot \frac{4}{3pn^2} =  \frac{f^2 (i)}{12} \ge \frac{960f(i+1)}{12p} = \frac{80f(i+1)}{p},
\]
where we have used properties (1)--(3), the fact that $f(i)/2 \ge f(i+1) \ge f(|\Lambda|+1)$, and the definition $f(i+1) = f^2(i)p/960$. Suppose the density of $b_{i+1}$ is at most $f(i+1)$ during the interval $t, \ldots, t+N$. We will show that this implies that with very high probability by time step $t + N$ the number of steps during which the count of $b_{i+1}$ increases is at least $2f(i+1)$, and the number of decrease events is at most $f(i+1)/2$. Since every step that decreases the count of $b_{i+1}$ can decrease its count by at most two, this implies that the density is at least $f(i+1)$ by time $t+N$.

Given that the density of state $b_{i+1}$ remains at most $f(i+1)$ during the interval $t, \ldots, t+N$, the number of steps that increase the density stochastically dominates the binomial random variable $Y \sim \Bin{N}{q(i+1)}$. The number of steps that decrease the count of $b_{i+1}$ is in turn stochastically dominated by a binomial random variable $Z \sim \Bin{N}{r(i+1)}$. Now using standard Chernoff bounds, we have that
\begin{align*}
\Pr[Y \le \E[Y]/2] \in  \exp(-\Omega(n)) \quad \textrm{ and } \quad \Pr[Z \ge 2 \cdot \E[Z]] \in \exp(-\Omega(n)).
\end{align*}
That is, w.v.h.p.\ the number of increases is at least $\E[Y]/2$ and the number of decreases at most $2 \cdot \E[Z]$. From the above bounds on $r(i+1)$ and $q(i+1)$, we get that 
\begin{align*}
  \E[Y]/2 - 2\cdot \E[Z] &= \frac{Nq(i+1)}{2} - 2Nr(i+1) \\
  &\ge N \cdot \left(\frac{80f(i+1)}{2p} - \frac{8f(i+1)}{p} \right) 
  = nf(i+1),
\end{align*}
since $N=np/32$. Therefore, the density of $b_{i+1}$ must be at least $f(i+1)$ by step $t + N$ with very high probability.
\end{proof}

In the proof above, the only information we use about the interaction graph is the information provided by \lemmaref{lemma:dense-erdos-graphs}. We note that \lemmaref{lemma:fully-dense-config} also holds for deterministic graphs having similar connectivity properties as guaranteed by \lemmaref{lemma:dense-erdos-graphs}. %
In particular, any dense graph with small spectral gap would also produce fully dense configurations with high probability.  

\subsection{Second step: Controlling for low count states}

For the remainder of the proof, we assume that for infinitely many $n$ protocol $\mathcal{P}$ stabilizes on a random graph $G \sim G_{n,p}$ in $o(n^2)$ expected steps and show that this leads to  a contradiction.  
A set $S \subseteq \Lambda$ is \emph{leader generating} if $S$ either contains a leader state (i.e. a state in which node's output is $\mathsf{leader}$), or given enough nodes in each state of $S$, there exists a finite sequence $\sigma = (e_1, \ldots, e_t)$ of interactions such that these agents generate a leader on the clique.

Alistarh et al.~\cite[Lemma A.7]{alistarh2017time} showed that for any leader generating set $S$ we need at most $2^{|\Lambda|}$ nodes in each state to generate a leader on the clique, i.e., a clique of size $|S| \cdot 2^{|\Lambda|}$ suffices to generate a new leader state. We show that if the protocol stabilizes in $o(n^2)$ steps, then in every leader generating set there is a state of count $< 2^{|\Lambda|}$. Let $A = \{ (u,v) : u,v \in V, u \neq v\}$. For any $\sigma = (e_1, \ldots, e_t) \in A^t$, let $\exec{\sigma}$ be the event that sequence $\sigma$ was sampled by the scheduler in the first $t$ steps, i.e., $\sigma$ is a prefix of the stochastic schedule.

In the following, we want to consider all graphs $G \sim G_{n,p}$ in which $\sigma$ could be executed and infer something about the presence of edges in such graphs. However, we cannot  simply state that $e \in E(G)$ with probability $p$, because the fact that $\sigma$ was sampled by the scheduler introduces bias into the random graph $G \sim G_{n,p}$, i.e., the events $e\in E(G)$ and $\exec{\sigma}$ are correlated. 

The next lemma shows that this correlation is small conditioned on the event that $|E| = m$ is sufficiently large. We write $\sigma \subseteq E$ for the event that all edges in a sequence $\sigma$ are in $E$, i.e., for $e_i = (v,u)$ we know that $\{v,u\}$ is an edge of $G$ for every $1 \le i \le t$.

\begin{lemma}
  \label{lemma:exec-prob}
  Let $\sigma \in A^t$, $N \le n(n-1)/2$ and $F \subseteq A$. Then the random graph $G = (V,E) \sim G_{n,p}$ satisfies
\[
\Pr[F \subseteq E \mid \exec{\sigma} \land m \ge N] \ge p^{|F|}\left(1 - \frac{t}{N}\right)^{|F|}.
\]
\end{lemma}
\begin{proof}
  For ease of exposition, we implicitly condition every probability in the proof on the event $m = |E| \ge N$.  The event $\exec{\sigma}$ implies that pairs in $\sigma$ are contained in $E$, that is, for every $(v,u) \in \sigma$ we have $\{v,u\} \in E$. In other words, all necessary edges to perform the interactions in $\sigma$ are present in the random graph. Use $\sigma \subseteq E$ to denote the event that all edges sampled by $\sigma$ are contained in $E$. Using the definition of conditional probability, we get
\begin{align*}
  \Pr[F \subseteq E \mid \exec{\sigma}] &=\Pr[F \subseteq E \mid \exec{\sigma} \land \sigma \subseteq E ] \\
  &= \frac{\Pr[\exec{\sigma} \land F\subseteq E \land \sigma \subseteq E ]}{\Pr[\exec{\sigma} \land \sigma \subseteq E]} \\
 &= \frac{\Pr[ \exec{\sigma} \mid F\subseteq E \land \sigma \subseteq E] \cdot \Pr[F\subseteq E \land \sigma \subseteq E ]}{\Pr[\exec{\sigma} \mid \sigma \subseteq E] \cdot \Pr[\sigma \subseteq E] }\\
 &= \frac{\Pr[\exec{\sigma} \mid F\subseteq E \land \sigma \subseteq E]}{\Pr[\exec{\sigma} \mid \sigma \subseteq E]} \cdot \Pr[F \subseteq E \mid \sigma \subseteq E ]\\ 
 &\ge \frac{\Pr[\exec{\sigma} \mid F\subseteq E \land \sigma \subseteq E]}{\Pr[\exec{\sigma} \mid \sigma \subseteq E] } \cdot p^{|F|},
\end{align*}
where last step is due to the fact that events $F \subseteq E$ and $\sigma \subseteq E \land m \ge N$ are non-negatively correlated.

Let $F = \{ f_1, \ldots, f_h \}$.  We use $Q_0$ to denote the event that $\sigma \subseteq E$ happens. For $i \ge 1$, define the events  $Q_{i} = Q_{i - 1} \land f_i \in E$ and $Q_{i}^* = Q_{i - 1} \land f_i \not \in E$.  Next, we show that the following holds for all $i \ge 1$:
\[
\frac{\Pr[\exec{\sigma} \mid Q_{i}]}{\Pr[\exec{\sigma} \mid Q_{i - 1}]} \ge \left(1 -\frac{t}{N}\right).
\]
The statement is trivial if $f_{i} \in \sigma$, as in that case $Q_i = Q_{i-1}$. Suppose then that $f_{i} \not \in \sigma$. We now want to lower bound the ratio
\[
\frac{\Pr[\exec{\sigma} \mid Q_i]}{\Pr[\exec{\sigma} \mid Q_{i}^*]}.
\]
Note that if we know $G$ contains exactly $m$ edges and that $\sigma \subseteq E$ holds, then $\Pr[\exec{\sigma}] = 1/(2m)^t$. It follows then that 
\[
\Pr[\exec{\sigma} \mid Q_i] = \E[(2m)^{-t} \mid Q_i \land m \ge N].
\]
(Recall that all our probabilities are implicitly conditioned on $m \ge N$.) Similarly, $\Pr[\exec{\sigma} \mid Q_{i}^*] = \E[(2m)^{-t} \mid Q_i^* \land m \ge N]$. That is, we have
\[
\frac{\Pr[\exec{\sigma} \mid Q_i]}{\Pr[\exec{\sigma} \mid Q_{i}^*]} = \frac{\E[(2m)^{-t} \mid Q_i \land m \ge N]}{\E[(2m)^{-t} \mid Q_i^* \land m \ge N]} = \frac{\E[m^{-t} \mid Q_i \land m \ge N]}{\E[m^{-t} \mid Q_i^* \land m \ge N]}.
\]
Notice that the expectation in the numerator is decreasing in $N$. Hence, we have
\[
\frac{\Pr[\exec{\sigma} \mid Q_i]}{\Pr[\exec{\sigma} \mid Q_{i}^*]} \ge \frac{\E[m^{-t} \mid Q_i \land m \ge N + 1]}{\E[m^{-t} \mid Q_i^* \land m \ge N]}.
\]
Suppose that edges activated in $\sigma$ and $\{f_1, \ldots, f_{i - 1}\}$ together contain exactly $k$ distinct edges. Then the distribution of $m \mid Q_i^* \land m \ge N$ follows the distribution of a shifted binomial random variable $\Bin{\frac{n(n-1)}{2} - k - 1}{p} + k$ conditioned on the fact that it is at least $N$. Let $X$ to be such a random variable. Similarly, the distribution of $m \mid Q_i \land m \ge N + 1$ is exactly the same as of $\Bin{\frac{n(n-1)}{2} - k - 1}{p} + k + 1$ conditioned on the fact that it is at least $N + 1$.  Therefore, it is the same as the distribution of $X + 1$, and so we get 
\begin{align*}
\frac{\Pr[\exec{\sigma} \mid Q_i]}{\Pr[\exec{\sigma} \mid Q_{i}^*]} &\ge \frac{\E[(X + 1)^{-t}]}{\E[X^{-t}]} = \frac{\sum_{x \ge N} \Pr[X = x] (x + 1)^{-t}}{\sum_{x \ge N} \Pr[X = x] x^{-t}}\\ 
                                                                    & \ge \left(1 - \frac{1}{N}\right)^t \ge 1 - \frac{t}{N},
\end{align*}
where the second inequality follows since for any $x \ge N$ the ratio $\frac{(x+1)^{-t}}{x^{-t}}$ is lower bounded by $\left(1 - 1/N\right)^t$. The last inequality follows from the Bernoulli inequality.  Next, consider the ratio
\[
R_i =   \frac{\Pr[\exec{\sigma} \mid Q_i]}{\Pr[\exec{\sigma} \mid Q_{i-1}]}.
\]
Observe that
\begin{align*}
    R_i &= 
    \frac{\Pr[\exec{\sigma} \mid Q_i]}{\Pr[\exec{q} \mid Q_{i}] \cdot \Pr[f_i \in E \mid Q_{i-1}] + \Pr[\exec{\sigma} \mid Q_{i}^*] \cdot \Pr[f_i \not\in E \mid Q_{i-1}]} \\
    &= \frac{1}{\Pr[f_i \in E \mid Q_{i - 1}] + \frac{\Pr[\exec{\sigma} \mid Q_i^*]}{\Pr[\exec{\sigma} \mid Q_i]} \cdot \Pr[f_i \not\in E \mid Q_{i - 1}]}\\
    &\ge \frac{1}{\Pr[f_i \in E \mid Q_{i - 1}] + \frac{1}{1 - t/N} \cdot \Pr[f_i \not\in E \mid Q_{i - 1}]}\\
    &\ge \frac{1}{\frac{1}{1 - t/N}\cdot\left(\Pr[f_i \in E \mid Q_{i - 1}] + \Pr[f_i \not\in E \mid Q_{i - 1}]\right)} = 1 - \frac{t}{N}.
\end{align*}
Therefore,
\[
\frac{\Pr[\exec{\sigma} \mid F\subseteq E \land \sigma \subseteq E]}{\Pr[\exec{\sigma} \mid \sigma \subseteq E] } = \prod_{i = 1}^h R_i = \prod_{i = 1}^h \frac{\Pr[\exec{\sigma} \mid Q_i]}{\Pr[\exec{\sigma} \mid Q_{i-1}]} \ge \left(1 - \frac{t}{N}\right)^{|F|}. \qedhere
\]
\end{proof}

By slight abuse of notation, we define $x$ to be the event that configuration $x$ was reached by the protocol after $t$ steps.

\begin{lemma}
  \label{lemma:configuration-condition}
  Let $t> 0$ and let $x$ be the configuration reachable after $t$ steps.
  For any $N \le n(n-1)/2$ and any sequence of edges $F$, we have
  \[
  \Pr[F \subseteq E \mid x \land m \ge N] \ge p^{|F|}\left(1 -\frac{t}{N}\right)^{|F|}.
  \]
\end{lemma}
\begin{proof}
  Let $\sigma_1, \ldots, \sigma_k$ be the possible interaction sequences of length $t$ that result in the configuration $x$. 
  Let $H$ denote the event $x \land m \ge N$. 
  By the law of total probability and  \lemmaref{lemma:exec-prob}, we get that     \begin{align*}
      \Pr[F \subseteq E \mid H] &= \sum_{i = 1}^k \Pr[F \subseteq E \mid \exec{\sigma_i} \land H] \cdot \Pr[\exec{\sigma_i} \mid H] \\
    &\ge \sum_{i = 1}^k p^{|F|}\left(1 -\frac{t}{N}\right)^{|F|} \cdot\Pr[\exec{\sigma_i} \mid H]  \\
    &\ge p^{|F|}\left(1 -\frac{t}{N}\right)^{|F|} \cdot \sum_{i = 1}^k  \Pr[\exec{\sigma_i} \mid H] \\
    &= p^{|F|}\left(1 -\frac{t}{N}\right)^{|F|}. \qedhere
    \end{align*}
\end{proof}

The above lemma gives us the power to overcome the challenge described in the introduction:
we have control over the structure of the graph when conditioned on a given configuration. However, we only have meaningful control over a portion of the graph of constant size. Turns out, this is sufficient for us. Say that the state $s$ is \emph{present in low count} in the configuration $x$ if its count is less than $2^{|\Lambda|}$ in $x$. Otherwise, we will say that it is \emph{present in high count}.

Let $T_{\mathcal P}$ be the stabilization time of the protocol. By our assumption $\E[T_{\mathcal P}] \in o(n^2)$. This means there exists some $t \in o(n^2)$ such that $T_{\mathcal P} < t$ with probability $1 - o(1)$. Indeed, take $t = \sqrt{n^2 \E[T_{\mathcal P}]}$. Note that $t \in o(n^2)$ and 
\[
\Pr[T_{\mathcal P} > t] \le \sqrt{\frac{\E[T_{\mathcal P}]}{n^2}} = o(1)
\] 
by Markov inequality. With the above, we can show the next lemma. %

\begin{lemma}
\label{lemma:low-count-states}
Let $x$ be a stable configuration reached by the protocol after $t$ steps.
Let $\xi$ be a set of states in low count in $x$.
Then for any $t \in o(n^2)$, such that configuration is stable after $t$ steps with probability $1-o(1)$, and for all leader generating sets $S$, we have that  $\xi \cap S \not = \emptyset$ with probability $\Theta(1)$ bounded away from zero.\end{lemma}

\begin{proof}
  By \lemmaref{lemma:dense-erdos-graphs}, $|E| = m > \frac{1}{4}pn^2$ with probability $1 - \exp(-\Omega(n))$. We will condition the rest of the proof on this event. Suppose 
we can find in $x$ a leader generating set $S$ with all states $s \in S$ being present in high count in $x$. That is, every state $s \in S$ has count at least $2^{|\Lambda|}$ in $x$. Pick $2^{|\Lambda|}$ nodes of each state in configuration $x$ and call this set $U$. As mentioned above, by \cite[Lemma A.7]{alistarh2017time}, leader can be elected from nodes in $U$ given that all necessary edges between them are present. Let $F^*$ be set of edges between all nodes in $U$. Then, by \lemmaref{lemma:configuration-condition}
\begin{align*}
    \Pr[U \text{ is a clique} \mid x] &= \Pr[F^* \subseteq E \mid x] 
                                                \ge p^{|F^*|}\left(1 - \frac{4t}{pn^2}\right)^{|F^*|}.
\end{align*}
Since $t \in o(n^2)$, $|F^*| \in O(1)$ and $p \in \Theta(1)$, for large enough $n$ the above probability is lower bounded by some constant. Let $\stable{x}$ denote the event that $x$ is a stable configuration and $\alllowcount{x}$ the event that all leader generating sets in $x$ have at least one state in low count. Then the above implies
\[
\Pr[\overline{ \stable{x}} \mid \overline{\alllowcount{x}}] \in \Theta(1).
\]
By Bayes' law, we get that
\begin{align*}
    \Pr[ \alllowcount{ x} \mid \stable{x}] &= 1 - \Pr[ \overline{ \alllowcount{ x}} \mid \stable{ x}]  \\
 &= 1 - \frac{\Pr[ \stable{ x} \mid \overline{ \alllowcount{ x}} ] \cdot \Pr[\overline{ \alllowcount{ x}}]}{\Pr[\stable{ x}]}  \\
    &\in 1 - \frac{(1 - \Theta(1)) \cdot \Pr[\overline{ \alllowcount{ x}}]}{1 - o(1)},
\end{align*}
which is lower bounded by a constant. Finally, note that if the event $\alllowcount{x}$ happens, then $\xi \cap S \not = \emptyset$ for any leader generating set $S$.
\end{proof}

Finally we are able to prove statement analogous to Lemma 4.4 in \cite{doty2018stable}. Specifically, we borrow some definitions and notation from \cite{doty2018stable}. We call an interaction between states $a$ and $b$ a \emph{$k$-bottleneck} if at the time of the  interaction, the state counts of $a$ and $b$ are upper bounded by $k$. Also in the following, we use $x \Rightarrow y$ to denote that the configuration $y$ is reachable from the configuration $x$ on a clique. Moreover, by $x \Rightarrow_{\sigma} y$ we denote that after executing interactions in $\sigma$ starting from configuration $x$, we end up in configuration $y$.

\begin{lemma}
  \label{lemma:bottleneck-free-transition}
  Suppose $\mathcal{P}$ stabilizes in $o(n^2)$ expected steps on $G \sim G_{n,p}$ for infinitely many $n$. Then there exists an infinite set $K \subseteq \mathbb N$ of indices and set $\xi \subseteq \Lambda$ of states such that for every leader generating set $S$ we have $\xi \cap S \not = \emptyset$ and there exist an infinite sequence of configurations $((x_k, y_k))_{k \in K}$ along with transition sequences $(\sigma_k)_{k \in K}$ such that the following holds 
\begin{enumerate}
    \item $x_k \Rightarrow_{\sigma_k} y_k$ and $\sigma_k$ does not contain any $k$-bottleneck interactions,
    \item $x_k$ contains every state in count at least $k$, and
    \item $y_k$ contains all states from $\xi$ in low count.
\end{enumerate} 
\end{lemma}
\begin{proof}
  Let $T_\mathcal{P}$ be the stabilization time of the protocol $\mathcal{P}$ and suppose $n$ is such that $\mathcal{P}$ takes $o(n^2)$ expected steps to stabilize.
  Consider the time step
  \[
t = \max \left \{ n \log n , \sqrt{n^2 \cdot \E[T_\mathcal{P}]} \right\}.
  \]
  By Markov inequality, the protocol has stabilized at step $t$ with probability $1 - o(1)$. Note that clearly $t \in \Omega(n \log n)$ and $t \in o(n^2)$.

  We now condition on the event $|E| > pn^2/4$. Since this holds with probability $1 - o(1)$ by \lemmaref{lemma:dense-erdos-graphs}, this will not affect the asymptotics of the expected stabilization time of the protocol. 
  Next, we apply the union bound over \lemmaref{lemma:fully-dense-config} and \lemmaref{lemma:low-count-states} to conclude that with at least constant positive probability, the stochastic schedule $\sigma$ satisfies the following properties provided that $n$ is large enough:
    \begin{enumerate}[label=(\alph*)]
        \item $\sigma$ passes through a fully dense configuration by time $t$ (\lemmaref{lemma:fully-dense-config}), and %
        \item in step $t$, the system is in a stable configuration, where every leader generating set has a state in low count (\lemmaref{lemma:low-count-states}).
    \end{enumerate}
    Suppose every such sequence contains a $k$-bottleneck interaction for some constant $k \in \mathbb N$ indepdendent of $n$. 
    Note that at every time step there are at most $|\Lambda|^2$ pairs of states with state counts at most $k$. Additionally, two nodes in states $a$ and $b$ with state count at most $k$ interact with probability at most $\frac{k^2}{m}$. Then the probability of a $k$-bottleneck interaction occurring at an arbitrary time step is at most 
    \[
    \frac{8k^2 |\Lambda|^2}{pn^2}
    \]
    by definition of $k$-bottleneck transitions and 
    since $m \ge pn^2/4$.
    Since (a) and (b) hold with at least constant positive probability, this implies that the expected stabilization time is of order $\Omega(n^2)$ (similarly to ~\cite[Observation 4.1]{doty2018stable}), contradicting our assumption. Hence, for all large enough $n$, 
    \begin{itemize}
    \item by property (a), we get that there exists a dense configuration $x_k$, and 
    \item by property (b), we get that there exists a configuration $y_k$ where every leader generating set has a state in low count, along with a $k$-bottleneck free transition sequence $\sigma_k$ such that $x_k \Rightarrow_{\sigma_k} y_k$ (as shown above). 
    \end{itemize}
    From such $k$-bottleneck free transition sequences, we may form a sequence $(( x_k,  y_k, \sigma_k))_{k \in \mathbb N}$ of configurations and transition sequences which satisfy first two conditions of the lemma. Additionally, for every leader generating set, $y_k$ has a state from it in low count.
    Notice however that for every $k$ we may have a different set $\xi$ of low count states. This can be resolved using the pigeonhole principle: since there are at most $2^{|\Lambda|}$ possibilities for $\xi \subseteq \Lambda$, we can select $K \subseteq \mathbb{N}$ and an infinite subsequence $((x_k,  y_k, \sigma_k))_{k \in K}$ of $((x_k,  y_k, \sigma_k))_{k \in \mathbb N}$ such that $\xi$ is the same for all $y_k$.
\end{proof}

With \lemmaref{lemma:bottleneck-free-transition}, \theoremref{thm:nsquared-lb} now follows using standard surgery techniques~\cite{doty2018stable}. We give the details in  the next section.

\subsection{Last steps: Proof of Theorem~\ref{thm:nsquared-lb}}

Consider the execution of sequences from \lemmaref{lemma:bottleneck-free-transition} on the clique. This setting is the same as the one considered by Doty and Soloveichik~\cite{doty2018stable} with our \lemmaref{lemma:bottleneck-free-transition} corresponding to their Lemma 4.4. The rest of the proof is now almost identical to proof by Doty and Soloveichik~\cite{doty2018stable}. For the sake of completeness, we now state the necessary results needed to finish the proof of \theoremref{thm:nsquared-lb}; we refer to \cite{doty2018stable} for details.

First, the transition ordering lemma for $G_{n,p}$ follows from \lemmaref{lemma:bottleneck-free-transition} (this corresponds to Lemma 4.5 from \cite{doty2018stable}):

\begin{lemma}
\label{lemma:ordering}
Let $b_1, b_2$ be positive integers such that $b_2 > |\Lambda| b_1$ and $K$, $\xi$ and $(x_k, y_k, \sigma_k)_{k \in K}$ be as in \lemmaref{lemma:bottleneck-free-transition}.
For any $k \in K$, there is an order on $\xi = \{d_1, \ldots, d_l\}$ such that for all $i \in \{1,\ldots,l\}$ there exists a transition $\alpha_i$, %
where node in state $d_i$ transitions to state $o_1$ and node in state $s_i$ transitions to state $o_2$, such that 
\begin{itemize}
\item $s_i, o_1, o_2 \not \in \{d_1, \ldots, d_i\}$, and 
\item $\alpha_i$ occurs at least $(b_2 - b_1 |\Lambda|)/|\Lambda|^2$ times in $\sigma_k$.
\end{itemize}
\end{lemma}

Recall that a configuration $x$ is a map $V \to \Lambda$. We say two configurations $x_1: V_1 \to \Lambda$ and $x_2: V_2 \to \Lambda$ are \emph{disjoint} if $V_1 \cap V_2 = \emptyset$. Then, for two disjoint configurations $x_1: V_1 \to \Lambda$ and $x_2: V_2 \to \Lambda$, define $x_1 + x_2$ to be a configuration on $V_1 \cup V_2$, where
\[
(x_1 + x_2)(v) = \begin{cases}
    x_1(v) & \textrm{if } v \in V_1 \\
    x_2(v) & \textrm{otherwise.}
\end{cases}
\]

With the above, we get the following lemma, which corresponds to (slightly adapted) Claims 1--3 given in \cite{doty2018stable}:

\begin{lemma}\label{lemma:claims}
Take $((x_k, y_k, \sigma_k))_{k \in K}$ and $\xi$ as in \lemmaref{lemma:bottleneck-free-transition} and 
let $x \Rightarrow y$  denote that configuration $y$ is reachable from configuration $x$ on a clique. Then:
\begin{enumerate}
\item There exists a configuration $q$ and a sequence $(z_k)_{k \in K}$ such that for large enough $k$ we have $x_k + q \Rightarrow z_k$ and the counts of states from $\xi$ are zero in $z_k$.

\item For any configuration $q$, there exist configurations $q'$ and a sequence $(z_k)_{k \in K}$ such that
for large enough $k$ we have $q' + x_k \Rightarrow q' + z_k + q$ and the counts of states in $\xi$ are are zero in $z_k$.

\item  Let $x_k'$ be a copy of $x_k$ disjoint with $x_k$. For large enough $k$, there exists a configuration $w_k$ such that $x_k + x_k' \Rightarrow w_k$ and counts of states from $\xi$ are zero in $w_k$. \end{enumerate}
\end{lemma}

We note that while we use only   \lemmaref{lemma:claims}.3 in the next proof, \lemmaref{lemma:claims}.3 follows from  \lemmaref{lemma:ordering}, \lemmaref{lemma:claims}.1 and \lemmaref{lemma:claims}.2; we refer to \cite{doty2018stable} for the details. 
With the above, we are now ready to prove the lower bound result.

\nsquarelowerbound*
\begin{proof}
 For the sake of contradiction, suppose that $\mathcal{P}$ stabilizes in $o(n^2)$ steps in $G \sim G_{n,p}$ with probability $1-o(1)$. By \lemmaref{lemma:bottleneck-free-transition} and \lemmaref{lemma:claims}.3, %
 there exists a configuration $w$ that is reachable on the clique in which all states from $\xi$ have zero count, where $\xi$ is some set that intersects every leader generating set.

 The only states that may have a non-zero count in $w$ are sets from $\Lambda \setminus \xi$. Since every leader generating set has a non-empty intersection with $\xi$, the set $\Lambda \setminus \xi$ cannot completely contain any leader generating set. Then, by definition of leader generating sets, a leader cannot be generated from $w$. Additionally, $w$ has zero leaders since every leader state is itself a leader generating set. Therefore, $w$ is a stable reachable configuration without any leaders, a contradiction. 
\end{proof}
 
\section{Conclusions}

In this work, we performed the first focused investigation of time-space trade-offs in the complexity of leader election on general graphs, in the population model. We provided some of the first time and space-efficient protocols for leader election, and the first time complexity bounds that are tight up to logarithmic factors. We introduced  ``graphical'' variants of classic population protocol techniques, such as information dissemination and approximate phase clocks on the upper bound side, and indistinguishability arguments and surgeries for lower bound results.

Our work leaves open the question of \emph{tight} bounds for both space and time complexity on general graph families, particularly in the case of sparse graphs. Another direction is considering other fundamental problems, such as majority, in the same setting, for which our techniques should prove  useful.

\section*{Acknowledgements}

We thank anonymous reviewers for their helpful comments. We gratefully acknowledge funding from the European Research Council (ERC) under the European Union’s Horizon 2020 research and innovation programme (grant agreement No 805223 ScaleML).

\bibliographystyle{plainnat}
\bibliography{references}

\appendix

\section{Proof of Lemma~\ref{lemma:meet-broadcast}}\label{apx:meet-broadcast}

\begin{proof}
For a function $f \colon V \to \mathbb{R}$ let $f(\bar{x})$ denote the average of $f$ over all $\deg(x)$ neighbours $w \in B(x)$ of $x$. Note that for $u \neq v$, the expected hitting times satisfies
$\hit{u,v} = \hit{\bar{u},v} + 1$.
Using the fact that the random walks are reversible, it can be shown analogously to Lemma 2 in \cite{coppersmith1993collisions} that the hitting times in the population model satisfy
\begin{equation}\label{eq:hit}
\hit{x,y} + \hit{y,z} + \hit{z,x} = \hit{x,z} + \hit{z,y} + \hit{y,x}
\end{equation}
for any $x,y,z \in V$. To see why, the probability of sampling any sequence of edges $e_1, \ldots, e_k$ that has the walk start from $x$ and end in $x$ is $1/m^k$. The number of walks $x \leadsto y \leadsto z \leadsto x$ and  $x \leadsto z \leadsto y \leadsto x$ are equal. Hence, the expected length of both types of walks are equal, which yields (\ref{eq:hit}).
The above implies that the relation $x \le y$ of the nodes defined by $x \le y \Leftrightarrow \hit{x,y} \le \hit{y,x}$ gives a preorder on the nodes. In particular, there exists a minimum value $z$ in this preorder such that $\hit{x,z} \ge \hit{z,x}$ for all $x \in V$.
Fix such a minimum node $z$. For any $x,y \in V$, define
\begin{align*}
  \Phi(x,y) &= \hit{x,y} + \hit{y,z} - \hit{z,y} \\
            &= \hit{y,x} + \hit{x,z} - \hit{z,x} = \Phi(y,x),
\end{align*}
where the second line follows from the identity~(\ref{eq:hit}).
By choice of $z$, the potential $\Phi(x,y)$ is non-negative. Moreover, the potential satisfies $\Phi(x,y) = \Phi(\bar{x},y)+1 = \Phi(x, \bar{y}) + 1$.

Now we show that $\vec M(x,y) \le \Phi(x,y)$ for all $x,y \in V$, which implies the claim of the lemma, as $\Phi(x,y) \le 2 \cdot \hit{G}$. For the sake of contradiction, suppose this is not true. Let $\alpha > 0$ be the maximum value attained by $\vec M(x,y) - \Phi(x,y)$ over $x,y \in V$. Choose $x,y$ to be the closest pair of nodes that satisfies $\vec M(x,y) = \Phi(x,y) + \alpha$. Note that $x \neq y$ and $\vec M(x',y) \le \Phi(x',y) + \alpha$ for all $x' \in V$. However, there is some neighbour $w$ of $x$ that is closer to $y$ than $x$, so $\vec M(w,y) < \Phi(w,y) + \alpha$. However, this yields a contradiction, as
\[
\vec M(x,y) = 1 + \vec M(\bar{x}, y) < 1 + \Phi(\bar{x},y) + \alpha = \Phi(x,y) + \alpha = \vec M(x,y). \qedhere
\]
\end{proof}

\section{Proof of Lemma~\ref{lemma:k-domination}}\label{apx:k-domination-lemma}

Let $\stateparam > 0$ be an integer and $K$ be the number of fair coin flips needed to observe $\stateparam$ consecutive heads. We now show the following technical lemma.

\kdominationlemma*

Define $f(k) = \Pr[K \ge k]$ for all $k \ge 0$. Clearly, $f(k) = 1$ for $0 \le k \le \stateparam$.  We will now show that for all $k \ge \stateparam$ the function $f(k)$ satisfies
\[
\left(1 - \frac{1}{2^\stateparam} \right)^k \le f(k) \le \left(1 - \frac{1}{2^{\stateparam+1}} \right)^{k-\stateparam}.
\]
This then establishes \lemmaref{lemma:k-domination}, as the lower bound equals $\Pr[Z_0 \ge k]$ and the upper bound equals $\Pr[Z_1 + \stateparam \ge k]$.

\begin{lemma}\label{lemma:k-identity}
  For all $k \ge \stateparam$, we have the identity
  \[
  f(k+1) = f(k) - \frac{f(k-\stateparam)}{2^{\stateparam+1}}.
  \]
\end{lemma}
\begin{proof}
  Recall that $K$ is the length of the shortest prefix of the infinite sequence $(X_i)_{i \ge 1}$ of i.i.d. random variables that contains $\stateparam$ consecutive ones, where $X_i$ denotes if the nodes was a initiator on its $i$th interaction and $\Pr[X_i = 1] = 1/2$ for all $i \ge 1$.  Note the function $f \colon \mathbb{N} \to [0,1]$ is non-increasing. Suppose that $K=k$, then the last $\stateparam$ coin flips must have been heads and the $(\stateparam+1)$th coin flip from end must have been zero. Moreover, the subsequence $X_1, \ldots, X_{k-\stateparam-1}$ cannot contain $\stateparam$ consecutive ones. Therefore, for $k \ge \stateparam$, we have the identity
\[
  \Pr[K = k] = \frac{1}{2^{\stateparam+1}} \cdot \Pr[ K \ge k - \stateparam] = \frac{f(k-\stateparam)}{2^{\stateparam+1}}.
\]
For all $k \ge \stateparam$, this gives the recurrence relation
\[
f(k+1) = f(k) - \frac{f(k-\stateparam)}{2^{\stateparam+1}}. \qedhere
\]
\end{proof}

\begin{lemma}
  For all $k \ge 0$, the function $f(k)$ satisfies
  \[
\left(1 - \frac{1}{2^\stateparam} \right)^k \le f(k).
  \]
\end{lemma}
\begin{proof}
  We now show by induction that $f(k) > f(k-\stateparam) / 2$ for $k \ge \stateparam$. As the base case, note that for $\stateparam \le k \le 2\stateparam$  \lemmaref{lemma:k-identity} gives
\[
f(k) = f(\stateparam) - \frac{k-\stateparam}{2^{\stateparam+1}} \ge 1 - \frac{\stateparam}{2^{\stateparam+1}} > f(k-\stateparam)/2,
\]
since $f(i) = 1$ for all $i \le \stateparam$. As the induction hypothesis, suppose $f(k') > f(k'-\stateparam)/2$ holds for all $k \ge k' \ge \stateparam$. The inductive step now follows by observing that for $k > 2\stateparam$ we have

\begin{align*}
f(k) &= f(k - 1) - \frac{1}{2^{\stateparam+1}} f(k-\stateparam-1)\\
      &= f(k - 2) -\frac{1}{2^{\stateparam+1}} f(k-\stateparam-1) - \frac{1}{2^{\stateparam+1}} f(k-\stateparam-2) \\
      &\ldots\\
      &= f(k-\stateparam) - \frac{1}{2^{\stateparam+1}} f(k-\stateparam-1) - \frac{1}{2^{\stateparam+1}} f(k-\stateparam-2)  - \ldots - \frac{1}{2^{\stateparam+1}} f(k-2\stateparam)\\
      &\ge f(k-\stateparam) - \frac{\stateparam}{2^{\stateparam+1}} f(k - 2\stateparam) \\
      &> f(k - \stateparam) - \frac{\stateparam}{2^\stateparam} f(k - \stateparam) \\ 
      &\ge f(k - \stateparam) - \frac{1}{2}f(k - \stateparam) = \frac{1}{2} f(k-\stateparam),
\end{align*}
where the last inequality is given by the induction hypothesis. This completes the induction. Since $f$ is non-increasing, we get the that for any $k > \stateparam$
\[
f(k+1) = f(k) - \frac{f(k-\stateparam)}{2^{\stateparam+1}} \ge f(k) \cdot \left(1 - \frac{1}{2^\stateparam} \right).
\]
Moreover, note that above holds for $k \le \stateparam$ as well since $f(k) = 1$ for any $k \le \stateparam$. Thus,
\[
f(k) \ge \left(1 - \frac{1}{2^\stateparam} \right)^{k}. \qedhere
\]
\end{proof}

\begin{lemma}
For all $0 \le k \le \stateparam$, we have $f(k)=1$ and for all $k \ge \stateparam$, the function $f(k)$ satisfies
\[
f(k) \le \left(1 - \frac{1}{2^{\stateparam+1}} \right)^{k-\stateparam}.
  \]
\end{lemma}
\begin{proof}
The proof of the upper bound proceeds in a similar fashion. We prove by induction that
\[
f(k) \le \left(1 - \frac{1}{2^{\stateparam+1}}\right)^{k-\stateparam}
\]
for all $k \ge \stateparam$. For the base case of the induction, note that for $\stateparam \le k \le 2\stateparam$ we have from \lemmaref{lemma:k-identity} and the fact that $f(i)=1$ for $0 \le i \le \stateparam$ that
\begin{align*}
  f(k) = f(\stateparam) - \frac{k-\stateparam}{2^{\stateparam+1}} = 1 - \frac{k-\stateparam}{2^{\stateparam+1}} \le  \left(1 - \frac{1}{2^{\stateparam+1}}\right)^{k-\stateparam},
\end{align*}
where the inequality follows from the Bernoulli inequality $1+rx \le (1+x)^r$ for $x \ge -1$ and $r \ge 0$. For the inductive step, we get that
\begin{align*}
  f(k+1) &= f(k + 1 - \stateparam) - \frac{1}{2^{\stateparam+1}} \cdot \sum_{i=1}^\stateparam f(k + 1 - \stateparam - i) \\
  &\le f(k+1-\stateparam) - \frac{\stateparam}{2^{\stateparam+1}} \cdot f(k+1-\stateparam) \\
  &\le f(k+1-\stateparam) \cdot \left(1-\frac{\stateparam}{2^{\stateparam+1}}\right) \\
  &\le f(k+1-\stateparam) \cdot \left(1-\frac{1}{2^{\stateparam+1}}\right)^\stateparam \\
  &\le  \left(1-\frac{1}{2^{\stateparam+1}}\right)^{k + 1 - 2\stateparam} \cdot \left(1-\frac{1}{2^{\stateparam+1}}\right)^\stateparam \\
  &= \left(1-\frac{1}{2^{\stateparam+1}}\right)^{k + 1 - \stateparam}
\end{align*}
by using the Bernoulli inequality and the induction hypothesis.
\end{proof}

\kdominationlemma*
\begin{proof}
  The claim follows using the definition of $f(k) = \Pr[K \ge k]$ and stochastic domination, since
  \begin{align*}
\Pr[Z_0 \ge k] &= \left(1 - \frac{1}{2^\stateparam} \right)^k \le f(k) = \Pr[K \ge k]  \\
&\le \left(1 - \frac{1}{2^{\stateparam+1}} \right)^{k-\stateparam} 
= \Pr[Z_1 + \stateparam \ge k]. \qedhere
\end{align*}
\end{proof}

\end{document}